\documentclass[aip]{revtex4-2}
\draft 

\usepackage[utf8]{inputenc}

\usepackage{amsmath}	
\usepackage{amssymb}	
\usepackage{amsfonts}
\usepackage{amsthm}		

\usepackage{mathtools}
\mathtoolsset{showonlyrefs}

\usepackage{xcolor,color}	

\usepackage{soul}				
\sethlcolor{yellow}				

\usepackage{graphicx}			
\graphicspath{{figures/}}

\usepackage{float}

\allowdisplaybreaks

\renewcommand{\vec}[1]{\boldsymbol{#1}}
\let\oldhat\hat
\renewcommand{\hat}[1]{\oldhat{\boldsymbol{#1}}}

\newcommand{\de}{\ensuremath{\partial}}

\newcommand{\field}[1]{\mathbb{#1}}
\newcommand{\ip}[2]{\ensuremath{ \left< \left. #1 \right| #2 \right> } }

\newcommand{\bra}[1]{\ensuremath{ \left< #1 \right| } }
\newcommand{\ket}[1]{\ensuremath{ \left| #1 \right> } }

\DeclareMathOperator{\ran}{ran}

\newtheorem{definition}{Definition}		

\newtheorem{remark}{Remark}
\newtheorem{theorem}{Theorem}
\newtheorem{lemma}{Lemma}
\newtheorem{proposition}{Proposition}

\numberwithin{lemma}{section}
\numberwithin{example}{section}
\numberwithin{proposition}{section}
\numberwithin{theorem}{section}
\numberwithin{corollary}{section}
\numberwithin{remark}{section}
\numberwithin{definition}{section}
\numberwithin{assumption}{section}

\let\diag\relax
\DeclareMathOperator{\diag}{diag}
\let\mod\relax
\DeclareMathOperator{\mod}{mod}

\begin{document}


\title[Existence of Magic Angle for Twisted Bilayer Graphene]{Existence of the first magic angle for the chiral model of bilayer graphene}



\author{Alexander B. Watson}
\email[]{watso860@umn.edu}
\affiliation{School of Mathematics, University of Minnesota Twin Cities, U.S.A.}

\author{Mitchell Luskin}
\email[]{luskin@umn.edu}
\affiliation{School of Mathematics, University of Minnesota Twin Cities, U.S.A.}


\date{\today}

\begin{abstract}
We consider the chiral model of twisted bilayer graphene introduced by Tarnopolsky-Kruchkov-Vishwanath (TKV). TKV have proved that for inverse twist angles $\alpha$ such that the effective Fermi velocity at the moir\'e $K$ point vanishes, the chiral model has a perfectly flat band at zero energy over the whole Brillouin zone. By a formal expansion, TKV found that the Fermi velocity vanishes at $\alpha \approx .586$. In this work, we give a proof that the Fermi velocity vanishes for at least one $\alpha$ between $.57$ and $.61$ by rigorously justifying TKV's formal expansion of the Fermi velocity over a sufficiently large interval of $\alpha$ values. The idea of the proof is to project the TKV Hamiltonian onto a finite dimensional subspace, and then expand the Fermi velocity in terms of explicitly computable linear combinations of modes in the subspace, while controlling the error. The proof relies on two propositions whose proofs are computer-assisted, i.e., numerical computation together with worst-case estimates on the accumulation of round-off error which show that round-off error cannot possibly change the conclusion of the computation.  The propositions give a bound below on the spectral gap of the projected Hamiltonian, an Hermitian $80 \times 80$ matrix whose spectrum is symmetric about $0,$ and verify that two real 18th order polynomials, which approximate the numerator of the Fermi velocity, take values with definite sign when evaluated at specific values of $\alpha$. Together with TKV's work our result proves existence of at least one perfectly flat band of the chiral model.
\end{abstract}


\pacs{}

\maketitle 



%
%

%


\section{Introduction} \label{sec:introduction}


\subsection{Outline} 
Twisted bilayer graphene (TBG) is formed by stacking one layer of graphene on top of another in such a way that the Bravais lattices of the layers  are twisted relative to each other. For generic twist angles, the atomic lattices will be incommensurate so that the resulting structure will not have periodic structure. Bistritzer-MacDonald (BM) \cite{Bistritzer2011} have introduced an approximate model (BM model) for the electronic states of TBG which is periodic over the scale of the bilayer moir\'e pattern, where the twist angle enters as a parameter. Using this model, BM showed that the Fermi velocity, the velocity of electrons at the Fermi level, vanishes at particular twist angles known as ``magic angles.'' The largest of these angles, known as the first magic angle, is at $\theta \approx 1.1$ degrees. Numerical computations on the BM model show the stronger result that at magic angles the Bloch band of the BM model at zero energy is approximately flat \emph{over the whole Brillouin zone} \cite{Bistritzer2011,San-Jose2012}. The flatness of the zero energy Bloch band is thought to be a critical ingredient for recently observed superconductivity of TBG \cite{Cao2018}, although the precise mechanism for superconductivity in TBG is not yet settled.

Aiming at a simplified model which explains the nearly-flat band of TBG, Tarnopolsky-Kruchkov-Vishwanath (TKV) \cite{Tarnopolsky2019} have introduced a simplification of the BM model which has an additional ``chiral'' symmetry, known as the chiral model. TKV showed analytically that at magic angles (of the chiral model, still defined by vanishing of the Fermi velocity), the chiral model has \emph{exactly} flat bands over the whole Brillouin zone. Using a formal perturbation theory (for the chiral model the natural parameter is the reciprocal of twist angle up to a constant) TKV have derived approximate values for the magic angles of the chiral model. It is worth noting that the first magic angles of the chiral model and the BM model are nearby, but the higher magic angles are not very close. 

Becker et al. have introduced a spectral characterization of magic angles of the TKV model where the role of a non-normal operator is emphasized (the operator $D^\alpha$ appearing in \eqref{eq:chiral_H}). Using this characterization, they have numerically computed precise values for the magic angles of the TKV model (see the discussion below \eqref{eq:Fermi_v_high_order}) \cite{becker2020mathematics}. In the same work, they also proved that the lowest band of the TKV model becomes exponentially close to flat even away from magic angles, as the natural small parameter tends to zero. The same authors have also investigated flat bands of the TKV model with more general interlayer coupling potentials, and the spectrum of other special cases of the BM model\cite{Becker2020}.

In this work we study the chiral model introduced by TKV and consider the problem of rigorously proving existence of the first magic angle.
We do this by justifying the formal perturbation theory of TKV to make a rigorous expansion of the Fermi velocity to high enough order, and over a large enough parameter range, so that we can prove existence of a zero. By numerically verifying that the resulting expansion attains a negative value and proving that the result continues to hold when the effect of round-off error is included (Proposition \ref{as:Fermi_v_zero}), we obtain existence of the magic angle (Theorem \ref{th:Fermi_v_zeros}). 

The proof of validity of the expansion is challenging because the reciprocal of the twist angle at the zero of the Fermi velocity is large relative to the spectral gap of the unperturbed Hamiltonian, which means that the magic angle falls outside of the interval of twist angles where the perturbation series for the Fermi velocity is obviously convergent. To overcome this difficulty, we start by representing the chiral model Hamiltonian in a basis which takes full advantage of model symmetries. Then, using a rigorous bound on the high frequency components of the error, we reduce the error analysis to analysis of the eigenvalues of the chiral model projected onto finitely many low frequencies. The final stage of the error analysis (Theorem \ref{th:error_theorem}) is to prove a proposition about the eigenvalues of the projected chiral model by a numerical computation that we prove continues to hold when the accumulation of round-off error is considered (Proposition \ref{as:PHP_gap}). We discuss the limitations of our methods, and in particular whether our methods might be generalized to the more general settings considered by Becker et al.\cite{becker2020mathematics,Becker2020} in Remarks \ref{rem:higher_angles}, \ref{rem:other_potentials}, and \ref{rem:BM_model}.

\subsection{Code availability}

We have made code for the numerical computations used in our proofs available at \texttt{github.com/abwats/magic\_angle}. We give references to specific scripts in the text.

\section{Statement of results} \label{sec:results}

\subsection{Tarnopolsky-Kruchkov-Vishwanath's chiral model}

The chiral model, like the Bistritzer-MacDonald model (B-M model) from which it is derived, is a formal continuum approximation to the atomistic tight-binding model of twisted bilayer graphene. The BM and chiral models aim to capture physics over the length-scale of the bilayer moir\'e pattern, which is, for small twist angles, much longer than the length-scale of the individual graphene layer lattices. Crucially, even when the graphene layers are incommensurate so that the bilayer is aperiodic on the atomistic scale, the chiral model and BM model are periodic (up to phases) with respect to the moir\'e lattice, so that they can be analyzed via Bloch theory.

We define the moir\'e lattice to be the Bravais lattice
\begin{equation}
    \Lambda = \left\{ m_1 \vec{a}_1 + m_2 \vec{a}_2 : (m_1,m_2) \in \field{Z}^2 \right\}
\end{equation}
generated by the moir\'e lattice vectors
\begin{equation}
    \vec{a}_1 = \frac{2 \pi}{3} \begin{pmatrix} \sqrt{3} , 1 \end{pmatrix}, \quad \vec{a}_2 = \frac{2 \pi}{3} \begin{pmatrix} - \sqrt{3} , 1 \end{pmatrix},
\end{equation}
and denote a fundamental cell of the moir\'e lattice by $\Omega$.
The moir\'e reciprocal lattice is the Bravais lattice
\begin{equation}
    \Lambda^* = \left\{ n_1 \vec{b}_1 + n_2 \vec{b}_2 : (n_1,n_2) \in \field{Z}^2 \right\}
\end{equation}
generated by the moir\'e reciprocal lattice vectors defined by $\vec{a}_i \cdot \vec{b}_j = 2 \pi \delta_{ij}$, given explicitly by
\begin{equation}
    \vec{b}_1 = \frac{1}{2} \begin{pmatrix} \sqrt{3}, 3 \end{pmatrix}, \quad \vec{b}_2 = \frac{1}{2} \begin{pmatrix} - \sqrt{3}, 3 \end{pmatrix}.
\end{equation}
We define $\vec{q}_1 = \begin{pmatrix} 0, - 1 \end{pmatrix}$, which is the (re-scaled) difference of the $K$ points (Dirac points) of each layer, and
\[\vec{q}_1 = \begin{pmatrix} 0, - 1 \end{pmatrix},\quad \vec{q}_2=\vec{q}_1+\vec{b}_1=\frac{1}{2} \begin{pmatrix} \sqrt{3}, 1 \end{pmatrix},\quad \vec{q}_3=\vec{q}_1+\vec{b}_2=\frac{1}{2} \begin{pmatrix} -\sqrt{3}, 1 \end{pmatrix}.
\]
We write $\Omega^*$ for a fundamental cell of the moir\'e reciprocal lattice, and refer to such a cell as the Brillouin zone.

Let $\phi := \frac{2 \pi}{3}$. Tarnopolsky-Kruchkov-Vishwanath's chiral Hamiltonian is defined as 
\begin{equation} \label{eq:chiral_H}
    H^\alpha = \begin{pmatrix} 0 & D^{\alpha\dagger} \\ D^\alpha & 0 \end{pmatrix}, \quad D^\alpha = \begin{pmatrix} - 2 i \overline{\de} & \alpha U(\vec{r}) \\ \alpha U(-\vec{r}) & - 2 i \overline{\de} \end{pmatrix},
\end{equation}
where $\overline{\de} = \frac{1}{2} ( \de_x + i \de_y )$, $U(\vec{r}) = e^{- i \vec{q}_1 \cdot \vec{r}} + e^{i \phi} e^{- i \vec{q}_2 \cdot \vec{r}} + e^{- i \phi} e^{- i \vec{q}_3 \cdot \vec{r}}$, $\dagger$ denotes the adjoint (Hermitian transpose), and $\alpha$ is a real parameter which we will take to be positive $\alpha \geq 0$ throughout (see \eqref{eq:alpha_ratio}). The chiral Hamiltonian $H^\alpha$ is an unbounded operator on $\mathcal{H} = L^2(\field{R}^2;\field{C}^4)$ with domain $H^1(\field{R}^2;\field{C}^4)$. We will write functions in $\mathcal{H}$ as
\begin{equation} \label{eq:psi_densities}
    \psi(\vec{r}) = \begin{pmatrix} \psi^A_1(\vec{r}), \psi^A_2(\vec{r}), \psi^B_1(\vec{r}), \psi^B_2(\vec{r}) \end{pmatrix},
\end{equation}
where $|\psi^\sigma_\tau(\vec{r})|^2$ represents the electron density near to the $K$ point (in momentum space) on sublattice $\sigma$ and on layer $\tau$. The diagonal terms of $D^\alpha$ arise from Taylor expanding the single layer graphene dispersion relation about the $K$ point of each layer, while the off-diagonal terms of $D^\alpha$ couple the $A$ and $B$ sublattices of layers $1$ and $2$. The chiral model is identical to the BM model except that inter-layer coupling between sublattices of the same type is turned off in the chiral model. The precise form of the interlayer coupling potential $U$ can be derived under quite general assumptions on the real space interlayer hopping \cite{Bistritzer2011,Catarina2019}. The parameter $\alpha$ is, up to unimportant constants, the ratio
\begin{equation} \label{eq:alpha_ratio}
    \alpha \sim \frac{ \text{interlayer hopping strength between $A$ and $B$ sublattices} }{ \text{twist angle} }.
\end{equation}
Although the limit $\alpha \rightarrow 0$ can be thought of as the limit of vanishing interlayer hopping strength at fixed twist, 
it is physically more interesting to view the limit as modeling increasing twist angle at a fixed interlayer hopping strength.

\subsection{Rigorous justification of TKV's formal expansion of the Fermi velocity and proof of existence of first magic angle} \label{sec:intro_TKVexpansion}

Bistritzer and MacDonald studied the effective Fermi velocity of electrons in twisted bilayer graphene modeled by the BM model, and computed values of the twist angle such that the Fermi velocity vanishes, which they called ``magic angles.'' One can similarly define an effective Fermi velocity for the chiral model, and refer to values of $\alpha$ such that the Fermi velocity vanishes as ``magic angles'' (although technically $\alpha$ is related to the reciprocal of the twist angle \eqref{eq:alpha_ratio}).

TKV proved the remarkable result that, at magic angles, the chiral model has a perfectly flat Bloch band at zero energy. Let $L^2_K$ denote the $L^2$ space on a single moir\'e cell $\Omega$ with moir\'e $K$ point Bloch boundary conditions. The starting point of TKV's proof is an expression for the Fermi velocity as a function of $\alpha$, $v(\alpha)$, as a functional of one of the Bloch eigenfunctions, $\psi^\alpha \in L^2_K$, of $H^\alpha$:
\begin{equation} \label{eq:Fermi_v}
    v(\alpha) := \frac{ | \ip{\psi^{\alpha*}(-\vec{r})}{\psi^\alpha(\vec{r})} | }{ | \ip{\psi^\alpha}{\psi^\alpha} | },
\end{equation}
where $\ip{.}{.}$ denotes the $L^2_K$ inner product. We give precise definitions of $L^2_K$, $\psi^\alpha$, and $v(\alpha)$ in Definition \ref{def:L2K_spaces}, Proposition \ref{prop:analytic_zeromodes}, and Definition \ref{def:Fermi_v}, respectively. We give a systematic formal derivation of why \eqref{eq:Fermi_v} is the effective Fermi velocity at the moir\'e $K$ point in Appendix \ref{sec:K_Dirac}.
To complete the proof, TKV showed that zeros of $v(\alpha)$ imply zeros of $\psi^\alpha$ at special ``stacking points'' of $\Omega$, and that such zeros of $\psi^\alpha$ allow for Bloch eigenfunctions with zero energy to be constructed for all $\vec{k}$ in the moir\'e Brillouin zone.

To derive approximate values for magic angles, TKV computed a formal perturbation series approximation of $\psi^\alpha$:
\begin{equation} \label{eq:TKV_expansion_psi}
    \psi^\alpha = \Psi^0 + \alpha \Psi^1 + ...
\end{equation}
and then substituted this expression into the functional for $v(\alpha)$ to obtain an expansion of $v(\alpha)$ in powers of $\alpha$:
\begin{equation} \label{eq:TKV_expansion_Fermi_v}
    v(\alpha) = \frac{ 1 - 3 \alpha^2 + \alpha^4 - \frac{111}{49} \alpha^6 + \frac{143}{294} \alpha^8 + ... }{ 1 + 3 \alpha^2 + 2 \alpha^4 + \frac{6}{7} \alpha^6 + \frac{107}{98} \alpha^8 + ... }.
\end{equation}
By setting $v(\alpha) = 0$ one obtains an approximation for the smallest magic angle: $\alpha \approx .586$.

Although TKV proved that flat bands occur at magic angles, they did not prove the existence of magic angles, and hence they did not prove the existence of flat bands. The contribution of the present work is to prove rigorous estimates on the error in the approximation \eqref{eq:TKV_expansion_psi} which are sufficiently high order and precise that, once substituted into \eqref{eq:Fermi_v}, they suffice to rigorously prove the existence of a zero of $v(\alpha)$, and hence, via TKV's proof, the existence of at least one perfectly flat band.

It turns out to be relatively straightforward to prove that the series \eqref{eq:TKV_expansion_psi} and \eqref{eq:TKV_expansion_Fermi_v} are uniformly convergent, and to derive precise bounds on the error in truncating the series, for $|\alpha| < \frac{1}{3}$; see Proposition \ref{prop:series_convergence}. The basic challenge, then, is to derive similar error bounds for $\alpha$ over an interval which includes the expected location of the first magic angle, at $\approx \frac{1}{\sqrt{3}}$. The first main theorem we will prove, roughly stated, is the following. See Theorem \ref{th:error_theorem_2} for the more precise statement. The theorem relies on existence of a spectral gap for an $80 \times 80$ Hermitian matrix which requires numerical computation for its proof, see Proposition \ref{as:PHP_gap}.
\begin{theorem} \label{th:error_theorem}
The $K$ point Bloch function $\psi^\alpha$ satisfies
\begin{equation} \label{eq:expansion}
    \psi^\alpha = \sum_{n = 0}^8 \alpha^n \Psi^n + \eta^\alpha
\end{equation}
where $\eta^\alpha \perp \sum_{n = 0}^8 \alpha^n \Psi^n$ with respect to the $L^2_K$ inner product, and 
\begin{equation}\label{eq:ub}
    \| \eta^\alpha \|_{L^2_K} \leq \frac{ 3 \alpha^9 }{ 15 - 20 \alpha }
    \quad \text{for all}\quad 0 \leq \alpha \leq \frac{7}{10}.
\end{equation}
\end{theorem}
The functions $\Psi^n$ for $0 \leq n \leq 8$ are derived recursively: see Appendix \ref{sec:TKV_expansion}. We stop at $8$th order in the expansion because this is the minimal order such that we can guarantee existence of a zero of $v(\alpha)$, but the functions $\Psi^n$ are well defined by a recursive procedure for arbitrary positive integers $n$, see Proposition \ref{prop:series_prop}. 

Substituting \eqref{eq:expansion} into the functional for the Fermi velocity \eqref{eq:Fermi_v} and using $\eta^\alpha \perp \sum_{n = 1}^8 \alpha^n \Psi^n$ we find
\begin{equation} \label{eq:Fermi_v_expansion}
    v(\alpha) = \frac{ v_{\mathcal{N}}(\alpha) }{ v_{\mathcal{D}}(\alpha) }
\end{equation}
where
\begin{equation} \label{eq:v_N}
\begin{split}
    v_{\mathcal{N}}(\alpha) := &\ip{ \sum_{n = 0}^8 \alpha^n \Psi^{n *}(-\vec{r}) }{ \sum_{n = 0}^8 \alpha^n \Psi^n(\vec{r}) } \\
    &+ \ip{ \eta^{\alpha *}(-\vec{r}) }{ \sum_{n = 0}^8 \alpha^n \Psi^n(\vec{r}) } + \ip{ \sum_{n = 0}^8 \Psi^{n *}(-\vec{r}) }{ \eta^\alpha(\vec{r}) } \\
    &+ \ip{ \eta^{\alpha *}(-\vec{r}) }{ \eta^\alpha(\vec{r}) },
\end{split}
\end{equation}
and
\begin{equation} \label{eq:v_D}
    v_{\mathcal{D}}(\alpha) := \ip{\sum_{n = 0}^8 \alpha^n \Psi^n }{ \sum_{n = 0}^8 \alpha^n \Psi^n } + \ip{ \eta^\alpha }{ \eta^\alpha }.
\end{equation}
where $\ip{.}{.}$ denotes the $L^2_K$ inner product and $\eta^\alpha$ satisfies \eqref{eq:ub}. The following is a straightforward calculation.  
\begin{proposition} \label{prop:Fermi_v_expansion}
The following identities hold:
\begin{equation} \label{eq:numerator_exp}
\begin{split}
	&\ip{ \sum_{n = 0}^8 \alpha^n \Psi^{n *}(-\vec{r}) }{ \sum_{n = 0}^8 \alpha^n \Psi^n(\vec{r}) }  \\
	&= 1 - 3 \alpha^2 + \alpha^4 - \frac{111}{49} \alpha^6 + \frac{143}{294} \alpha^8 - \frac{7536933}{11957764} \alpha^{10} \\
	&\quad + \frac{4598172331}{47460365316} \alpha^{12} - \frac{30028809212865451}{520327364608478700} \alpha^{14} + \frac{49750141858992227}{12487856750603488800} \alpha^{16},
\end{split}
\end{equation}
\begin{equation} \label{eq:denominator_exp}
\begin{split}
	&\ip{\sum_{n = 0}^8 \alpha^n \Psi^n }{ \sum_{n = 0}^8 \alpha^n \Psi^n } 	\\
	&= 1 + 3 \alpha^2 + 2 \alpha^4 + \frac{6}{7} \alpha^6 + \frac{107}{98} \alpha^8 + \frac{5119}{48412} \alpha^{10}    \\
	&\quad + \frac{62026511}{356844852} \alpha^{12} + \frac{355691470247}{113410497953025} \alpha^{14} + \frac{2481663780475871}{337509641908202400} \alpha^{16}.
\end{split}
\end{equation}
\end{proposition}
We prove Proposition \ref{prop:Fermi_v_expansion} in Appendix \ref{sec:Fermi_v_expansion}. Na\"ively, the expansions \eqref{eq:numerator_exp} and \eqref{eq:denominator_exp} approximate the formal infinite series expansions of $\ip{ \sum_{n = 0}^\infty \alpha^n \Psi^{n *}(-\vec{r}) }{ \sum_{n = 0}^\infty \alpha^n \Psi^n(\vec{r}) }$ and $\ip{ \sum_{n = 0}^\infty \alpha^n \Psi^n }{ \sum_{n = 0}^\infty \alpha^n \Psi^n }$ up to terms of order $\alpha^9$. We prove in Proposition \ref{prop:approx_of_exp} that, because of some simplifications, expansions \eqref{eq:numerator_exp} and \eqref{eq:denominator_exp} agree with the infinite series up to terms of order $\alpha^{10}$. 

We are now in a position to state and prove our second result. This result also relies on a proposition which requires numerical computation for its proof: that one real 18th order polynomial in $\alpha$ attains a negative value, and another attains a positive value, when evaluated at specific values of $\alpha$, see Proposition \ref{as:Fermi_v_zero}.
\begin{theorem} \label{th:Fermi_v_zeros}
There exist positive numbers $\alpha_{min}$ and $\alpha_{max}$ with $.57 < \alpha_{min} < \alpha_{max} < .61$ such that the Fermi velocity $v(\alpha)$ defined by \eqref{eq:Fermi_v} has a zero $\alpha^*$ satisfying $\alpha_{min} \leq \alpha^* \leq \alpha_{max}$. 
\end{theorem}
\begin{proof}
Equation \eqref{eq:v_N} and Cauchy-Schwarz imply that
\begin{equation}
    \left| v_{\mathcal{N}}(\alpha) - \ip{ \sum_{n = 0}^8 \alpha^n \Psi^{n *}(-\vec{r}) }{ \sum_{n = 0}^8 \alpha^n \Psi^n(\vec{r}) } \right| \leq 2 \| \eta^\alpha \| \sum_{n = 0}^8 \alpha^n \left\| \Psi^n \right\| + \| \eta^\alpha \|^2.
\end{equation}
Using Theorem \ref{th:error_theorem} and Proposition \ref{prop:Psi_n_norms}, we see that $v_{\mathcal{N}}(\alpha)$ is bounded above by the polynomial 
\begin{equation} \label{eq:worst_case}
\begin{split}
    &1 - 3 \alpha^2 + \alpha^4 - \frac{111}{49} \alpha^6 + \frac{143}{294} \alpha^8 - \frac{7536933}{11957764} \alpha^{10} \\
    &\quad + \frac{4598172331}{47460365316} \alpha^{12} - \frac{30028809212865451}{520327364608478700} \alpha^{14} + \frac{49750141858992227}{12487856750603488800} \alpha^{16} \\
    &\quad + \mathcal{E}(\alpha),
\end{split}
\end{equation}
where
\begin{equation}
\begin{split}
    \mathcal{E}(\alpha) := \; &\frac{6 \alpha^9}{15 - 20\alpha} \left( 1 + \sqrt{3} \alpha + \sqrt{2} \alpha^2 + \frac{\sqrt{14}}{7} \alpha^3 + \frac{ \sqrt{258} }{ 42 } \alpha^4 + \frac{ \sqrt{1968837} }{3458} \alpha^5 \right.    \\
    &\quad \quad \quad \left. + \frac{ \sqrt{106525799} }{31122} \alpha^6 + \frac{ 2\sqrt{2129312323981473} }{ 624696345 } \alpha^7 + \frac{ \sqrt{183643119755214454} }{ 4997570760 } \alpha^8 \right)    \\
    &\quad \quad \quad \quad + \frac{9 \alpha^{18}}{(15-20\alpha)^2},
\end{split}
\end{equation}
where we use Proposition \ref{prop:Psi_n_norms} to calculate the term in brackets, for all $0 \leq \alpha \leq \frac{7}{10}$. On the other hand, $v(\alpha)$ is bounded below for all $0 \leq \alpha \leq \frac{7}{10}$ by the polynomial
\begin{equation} \label{eq:best_case}
\begin{split}
    &1 - 3 \alpha^2 + \alpha^4 - \frac{111}{49} \alpha^6 + \frac{143}{294} \alpha^8 - \frac{7536933}{11957764} \alpha^{10} \\
    &\quad + \frac{4598172331}{47460365316} \alpha^{12} - \frac{30028809212865451}{520327364608478700} \alpha^{14} + \frac{49750141858992227}{12487856750603488800} \alpha^{16} \\
    &\quad - \mathcal{E}(\alpha).
\end{split}
\end{equation}
We now claim the following.
\begin{proposition} \label{as:Fermi_v_zero}
Expression \eqref{eq:worst_case}, or equivalently the 18th order polynomial obtained by multiplying \eqref{eq:worst_case} by $(15 - 20 \alpha)^2$, is negative at $\alpha = .61$. Similarly, expression \eqref{eq:best_case} is positive at $\alpha = .57$.
\end{proposition}
Proposition \ref{as:Fermi_v_zero} obviously implies by continuity that \eqref{eq:worst_case}, \eqref{eq:best_case}, and $v_{\mathcal{N}}(\alpha)$, each have at least one zero in the interval $.57 < \alpha < .61$. We denote the largest zero of \eqref{eq:worst_case} in the interval by $\alpha_{max}$, and the smallest zero of \eqref{eq:best_case} in the interval by $\alpha_{min}$. Since the zeroes of $v_{\mathcal{N}}(\alpha)$ must lie between those of \eqref{eq:worst_case} and \eqref{eq:best_case} we are done.
\end{proof}






\begin{proof}[Proof of Proposition \ref{as:Fermi_v_zero} (computer assisted)]
We will first prove that \eqref{eq:worst_case} attains a negative value at $.61$, then explain the modifications necessary to prove that \eqref{eq:best_case} is positive at $.57$. Evaluating using double-precision floating point arithmetic we find that at $\alpha = .61$, \eqref{eq:worst_case} attains the negative value $-0.020263$ (five significant figures, this value was computed by running the script \texttt{compute\_expansion\_symbolically.py} in the Github repo).
It is straightforward to bound the numerical error which accumulates when evaluating an 18th order polynomial using floating point arithmetic. Even the simplest exact bound, which doesn't account for error cancellation, see e.g. equation (8) of Oliver\cite{Oliver1979}, yields an upper bound on the possible accumulated round-off error in evaluation of an $n$th order polynomial $\sum_{j = 0}^n p_j \alpha^j$, for $\alpha \in [-1,1]$, as $(n + 1) \left[ e^{(2 n +1)\epsilon} - 1\right] \sup_{0 \leq j \leq n} |p_j|$, where $\epsilon$ is ``machine epsilon'': roughly speaking, the maximum possible round-off error generated in a single arithmetic operation. Bounding the maximum coefficient in \eqref{eq:worst_case} by 1000, taking $n = 18$, and bounding $\epsilon$ by $3 \times 10^{-16}$ (which was easily attained working in Python on our machine), the maximum possible numerical error in the evaluation of \eqref{eq:worst_case} is $\approx 10^{-11}$, which is much smaller than $0.020263$. We conclude that the first claim of Proposition \ref{as:Fermi_v_zero} must hold. Regarding the second, evaluating at $\alpha = .57$ we find that \eqref{eq:best_case} equals $0.029138$ (5sf). The same argument as before now shows that accumulated round-off error in the evaluation cannot possible change the sign of \eqref{eq:best_case} at $\alpha = .57$.
\end{proof} 


\begin{figure}
\includegraphics[scale=.5]{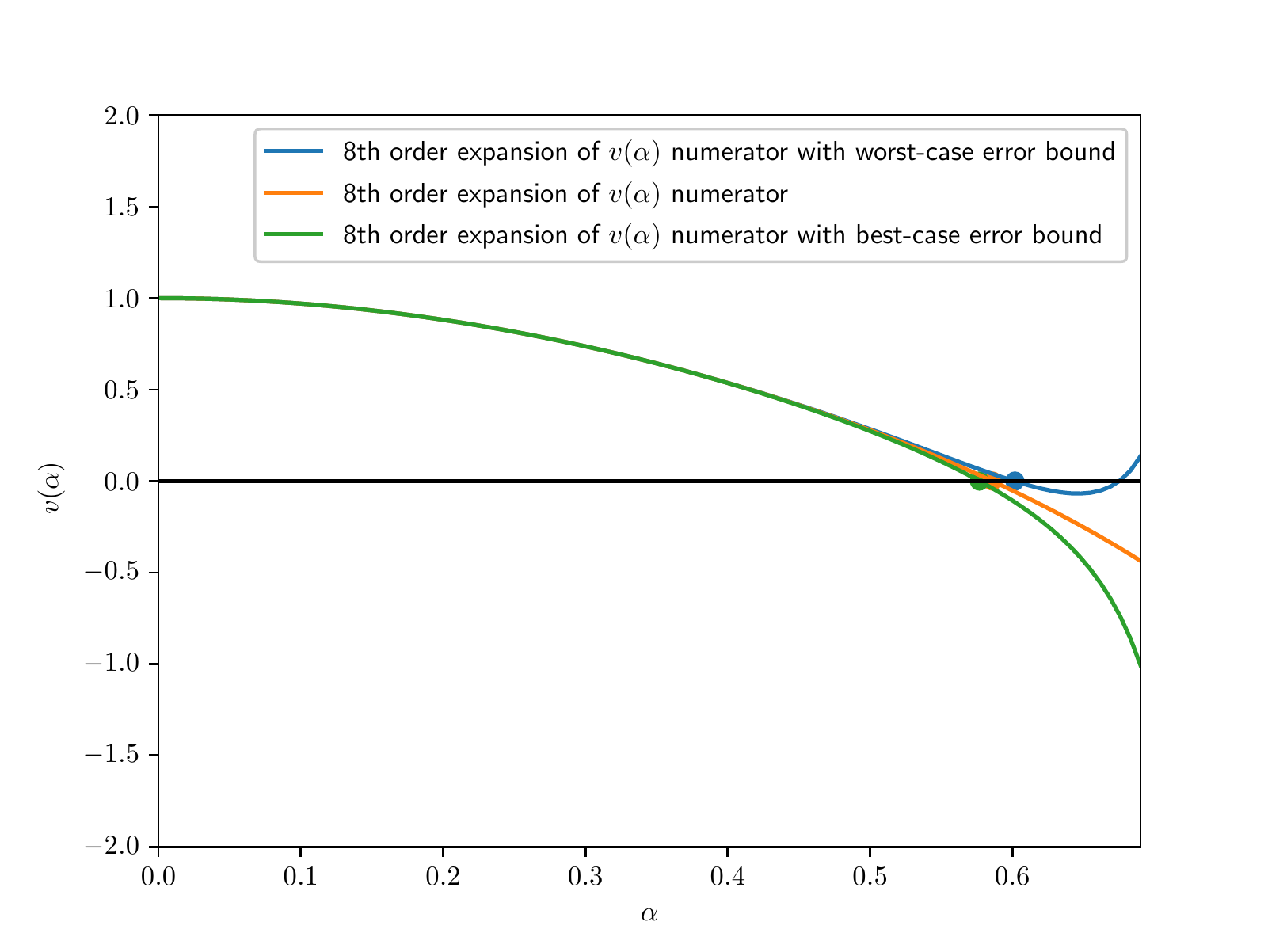}
\includegraphics[scale=.5]{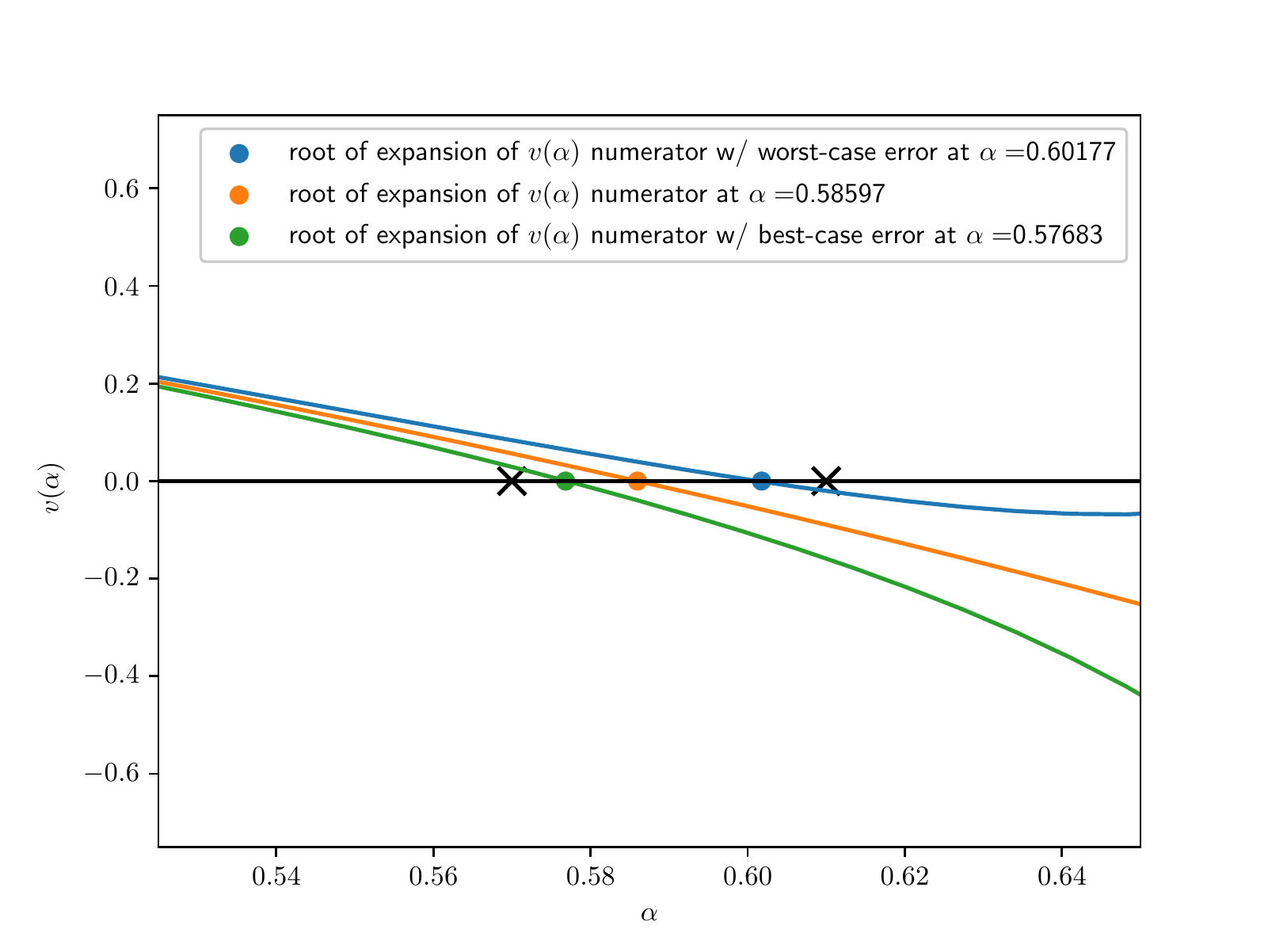}
\caption{
At left, plot of the numerator $v_{\mathcal{N}}(\alpha)$ of the Fermi velocity approximated by the 8th order TKV expansion \eqref{eq:TKV_expansion_Fermi_v} (orange), and of 8th order expansions with worst-case \eqref{eq:worst_case} (blue) and best-case \eqref{eq:best_case} (green) errors. At right, detail showing computed roots of these functions near to $\alpha = \frac{1}{\sqrt{3}}$. Numerically computing the zeroes of each curve yields $\alpha = 0.58597$ (5sf), $\alpha = 0.60177$ (5sf), and $\alpha = 0.57683$ (5sf), respectively. The values of $\alpha$ ($.57$ and $.61$) where we evaluated expressions \eqref{eq:best_case} (green line) and \eqref{eq:worst_case} (blue line) to prove that $v_{\mathcal{N}}(\alpha)$ has a zero between $.57$ and $.61$ are shown with black crosses.
}
\label{fig:check_zero}
\end{figure}
We do not attempt to rigorously estimate $\alpha_{min}$ and $\alpha_{max}$ precisely in this work, but numerically computing roots of the polynomials \eqref{eq:worst_case} and \eqref{eq:best_case} suggests $\alpha_{min} \approx 0.57683$ (5sf) and $\alpha_{max} \approx 0.60177$ (5sf) respectively, where (5sf) is an abbreviation for (five significant figures). Numerical computation of the first zero of $\ip{ \sum_{n = 0}^8 \alpha^n \Psi^{n*}(-\vec{r}) }{ \sum_{n = 0}^8 \alpha^n \Psi^n(\vec{r}) }$ gives $0.58597$ (5sf), see Figure \ref{fig:check_zero} (the zero values were computed by running the script \texttt{compute\_expansion\_symbolically.py} in the Github repo).

Using Proposition \ref{prop:H0_inv_H1} and the package Sympy\cite{10.7717/peerj-cs.103} for symbolic computation we can compute the formal expansion of $v(\alpha)$ up to arbitrarily high order in $\alpha$. In particular, we find the higher-order terms in the expansion \eqref{eq:TKV_expansion_Fermi_v} to be
\begin{equation} \label{eq:Fermi_v_high_order}
    v(\alpha) = \frac{ 1 - 3 \alpha^2 + \alpha^4 - \frac{111}{49} \alpha^6 + \frac{143}{294} \alpha^8 - \frac{10227257}{11957764} \alpha^{10} + \frac{6881137015}{47460365316} \alpha^{12} - \frac{130055941435858531}{520327364608478700} \alpha^{14} + ... }{ 1 + 3 \alpha^2 + 2 \alpha^4 + \frac{6}{7} \alpha^6 + \frac{107}{98} \alpha^8 + \frac{16011}{48412} \alpha^{10} + \frac{134058653}{356844852} \alpha^{12} + \frac{26407145691649}{226820995906050} \alpha^{14} + ... }.
\end{equation}
Truncating the numerator after order $\alpha^{40}$ and setting the numerator equal to zero yields $\alpha = 0.58566355838956$ (14sf) for the first zero of the Fermi velocity (to compute this value run \texttt{compute\_expansion\_symbolically.py} in the Github repo with $N = 40$). This is consistent with the numerical computation of Becker et al.\cite{becker2020mathematics}, who found $\alpha = 0.58566355838955$ truncated (not rounded) to 14 digits, by diagonalizing a non-normal but compact operator whose reciprocal eigenvalues correspond to magic angles. Note that we do not attempt to rigorously justify the series \eqref{eq:Fermi_v_high_order} to such large values of $\alpha$ and to such high order in this work, although see Remark \ref{rem:higher_terms}.

\begin{remark}[Higher magic angles] \label{rem:higher_angles}
The chiral model has been conjectured to have infinitely many magic angles \cite{becker2020mathematics}, but it isn't straightforward to extend our methods to prove existence of such higher magic angles. The problem is that calculating the perturbation series centered at $\alpha = 0$ requires diagonalizing the unperturbed operator $H^0$. In principle it might be possible to calculate the perturbation series to higher order in order to get an accurate approximation of the Fermi velocity near to the higher magic angles. However, this would require significantly more calculation compared with the present work, and we have no guarantee that the error can be made small enough to prove existence of another zero in that case. 
\end{remark}
\begin{remark}[More general interlayer hopping potentials] \label{rem:other_potentials}
The chiral model \eqref{eq:chiral_H} is an approximation to the full Hamiltonian of the twisted bilayer, even in the chiral limit where coupling between sublattices of the same type is turned off, because the interlayer hopping potential $U$ only allows for hopping between nearest neighbors in the momentum lattice (see Figure \ref{fig:H1_chi_G}). More general interlayer hopping potentials have been studied by Becker et al.\cite{Becker2020}. In principle, such models should be amenable to the analysis of this work, but longer-range hopping would lead to much more involved calculations, and the construction of the finite-dimensional subspace $\Xi$ of Proposition \ref{prop:mu_choice} would require more care: the fact that we can choose $\Xi$ so that $\| P_\Xi H^1 P^\perp_\Xi \| = 1$ depends on $H^1$ only coupling nearest-neighbors in the momentum lattice. Locality of hopping in the momentum space lattice has been exploited for efficient computation of density of states \cite{doi:10.1137/17M1141035} of twisted bilayers. 
\end{remark}
\begin{remark}[Generalization to BM model] \label{rem:BM_model}
Parts of our analysis should also apply to the full Bistritzer-MacDonald model. Specifically, one could study perturbation series for Bloch functions near to zero energy in powers of the inter-layer hopping strength, derive an equivalent expression for the Fermi velocity in terms of that series, and then study the zeroes of that series. However, there are various complications because of the lack of ``chiral'' symmetry. First, there is no reason for the continuation of the zero eigenvalue of the unperturbed operator to remain at zero. Second, the expression for the Fermi velocity in terms of the associated eigenfunction could be more complicated. Since zeros of the BM model Fermi velocity do not imply existence of flat bands for that model, we do not consider these complications in this work.
\end{remark}
\begin{remark}[Expanding to higher order] \label{rem:higher_terms}
Our methods could in principle be continued to justify the expansion of the Fermi velocity to arbitrarily high order and potentially over larger intervals of $\alpha$ values. However, these extensions aren't immediate: pushing the expansion to higher order or to a larger interval of $\alpha$ values would require a larger set $\Xi$ in Lemma \ref{lem:decompose}, and Proposition \ref{as:PHP_gap} would have to be re-proved for the new set $\Xi$. Note that the essential difficulty is justifying the perturbation series for large $\alpha$: the series are easily justified to all orders for $|\alpha| < \frac{1}{3}$, see Proposition \ref{prop:series_convergence}.
\end{remark}







\subsection{Structure of paper}

We review the symmetries, Bloch theory, and symmetry-protected zero modes of TKV's chiral model in Section \ref{sec:TKV_model}. We prove Theorem \ref{th:error_theorem} in Section \ref{sec:rigorous_expansion}, postponing most details of the proofs to the appendices. In Appendix \ref{sec:K_Dirac} we show why \eqref{eq:Fermi_v} corresponds to the effective Fermi velocity at the moir\'e $K$ point. In Appendix \ref{sec:chiral}, we construct an orthonormal basis, which we refer to as the chiral basis, which allows for efficient computation and analysis of TKV's formal expansion. We re-derive TKV's formal expansions in Appendix \ref{sec:TKV_expansion}. We give details of the proof of Theorem \ref{th:error_theorem} in Appendices \ref{sec:prop_proof} and \ref{sec:verify_gap_assump}. We prove Proposition \ref{prop:Fermi_v_expansion} in Appendix \ref{sec:Fermi_v_expansion}. In the supplementary material, we list the basis functions of the subspace onto which we project the TKV Hamiltonian, give the explicit forms of the higher-order corrections in the expansion \eqref{eq:expansion}, and present a derivation of the TKV Hamiltonian from the Bistritzer-MacDonald model.


\section{Symmetries, Bloch theory, and zero modes of TKV's chiral model} \label{sec:TKV_model}


\subsection{Symmetries of the TKV model}


In this section, we review the symmetries of the TKV model for the reader's convenience and to fix notation. Becker et al.\cite{becker2020mathematics} have given a group theoretical account of these symmetries, and further reviews can be found in the physics literature \cite{Zou2018,Po2018,Po2019}. Recall that $\phi = \frac{2 \pi}{3}$ and let $R_\phi$ denote the matrix which rotates vectors counter-clockwise by $\phi$, i.e.,
\begin{equation}
    R_\phi = \frac{1}{2} \begin{pmatrix} -1 & - \sqrt{3} \\ \sqrt{3} & -1 \end{pmatrix}.
\end{equation}
We define
\begin{definition}
For any $\vec{v} \in \Lambda$ we define a phase-shifted translation operator acting on functions $f \in \mathcal{H}$ by
\begin{equation} \label{eq:translate}
    \tau_{\vec{v}} f := \diag\left( 1,e^{i \vec{q}_1 \cdot \vec{v}},1,e^{i \vec{q}_1 \cdot \vec{v}} \right) \tilde{\tau}_{\vec{v}} f, \quad \tilde{\tau}_{\vec{v}} f(\vec{r}) = f(\vec{r} + \vec{v}).
\end{equation}
We define a phase-shifted version of the operator which rotates functions $f \in \mathcal{H}$ clockwise by $\phi$ by
\begin{equation} \label{eq:rotate}
    \mathcal{R} f := \diag\left( 1,1,e^{- i \phi},e^{- i \phi} \right) \tilde{\mathcal{R}} f, \quad \tilde{\mathcal{R}} f(\vec{r}) = f( R_\phi \vec{r} ).
\end{equation}
For any $f \in \mathcal{H}$ we finally define the ``chiral'' symmetry operator
\begin{equation} \label{eq:chiral}
    \mathcal{S} f := \diag\left(1,1,-1,-1\right) f.
\end{equation}
\end{definition}
We then have the following.
\begin{proposition}
The operators \eqref{eq:translate} and \eqref{eq:rotate} are symmetries in the sense that
\begin{equation} \label{eq:trans_sym}
    \left[ H^\alpha , \tau_{\vec{v}} \right] = H^\alpha\tau_{\vec{v}}-\tau_{\vec{v}}H^\alpha=0
\end{equation}
for all moir\'e lattice vectors $\vec{v} \in \Lambda$,
\begin{equation} \label{eq:rot_sym}
    \left[ H^\alpha , \mathcal{R} \right] = H^\alpha\mathcal{R}-\mathcal{R}H^\alpha = 0,
\end{equation}
and the operator \eqref{eq:chiral} is a ``chiral'' symmetry in the sense that
\begin{equation} \label{eq:chiral_sym}
    \left\{ H^\alpha , \mathcal{S} \right\} = H^\alpha\mathcal{S} +\mathcal{S} H^\alpha = 0.
\end{equation}
\end{proposition}
\begin{proof}
The first claim is a direct calculation using the facts that for any $\vec{v} \in \Lambda$
\begin{equation}
    \tilde{\tau}_{-\vec{v}} U(\vec{r}) \tilde{\tau}_{\vec{v}} = e^{- i \vec{q}_1 \cdot \vec{v}} U(\vec{r}), \quad \tilde{\tau}_{- \vec{v}} \overline{\de} \tilde{\tau}_{\vec{v}} = \overline{\de}.
\end{equation}
The second claim is a direct calculation using the facts that
\begin{equation}
    \tilde{\mathcal{R}}^{-1} U(\vec{r}) \tilde{\mathcal{R}} = e^{- i \phi} U(\vec{r}), \quad \tilde{\mathcal{R}}^{-1} \overline{\de} \tilde{\mathcal{R}} = e^{- i \phi} \overline{\de}.
\end{equation}
The final claim is trivial to check.
\end{proof}
The ``chiral'' symmetry \eqref{eq:chiral_sym} implies that the spectrum of $H^\alpha$ is symmetric about zero, because
\begin{equation}
    H^\alpha \psi = E \psi \iff H^\alpha \mathcal{S} \psi = - E \mathcal{S} \psi.
\end{equation}
The same calculation implies that zero modes of $H^\alpha$ can always be chosen without loss of generality to be eigenfunctions of $\mathcal{S}$.

\subsection{Bloch theory for the TKV Hamiltonian} \label{sec:Bloch_theory}

We now want to reduce the eigenvalue problem for $H^\alpha$ using the symmetries just introduced. The symmetry \eqref{eq:trans_sym} means that eigenfunctions of $H^\alpha$ can be chosen without loss of generality to be simultaneous eigenfunctions of $\tau_{\vec{v}}$ for all $\vec{v} \in \Lambda$. It therefore suffices to seek solutions of
\begin{equation} \label{eq:eigvalp}
    H^\alpha \psi = E \psi
\end{equation}
for $\vec{r}$ in a fundamental cell $\Omega := \field{R}^2/\Lambda$ of the moir\'e lattice in the symmetry-restricted spaces
\begin{equation} \label{eq:L2k_space}
    L^2_{\vec{k}} := \left\{ f \in L^2(\Omega;\field{C}^4) : f(\vec{r} + \vec{v}) = e^{i \vec{k} \cdot \vec{v}} \diag(1,e^{i \vec{q}_1 \cdot \vec{v}},1,e^{i \vec{q}_1 \cdot \vec{v}}) f(\vec{r})\ \forall \vec{v} \in \Lambda \right\}
\end{equation}
where $\vec{k}$ is known as the quasimomentum. Since $L^2_{\vec{k} + \vec{w}} = L^2_{\vec{k}}$ for any $\vec{w} \in \Lambda^*,$ it suffices to restrict attention to $\vec{k}$ in a fundamental cell of $\Lambda^*$ which we denote $\Omega^* := \field{R}^2/\Lambda^*$ and refer to as the Brillouin zone. We also define symmetry-restricted Sobolev spaces $H^s_{\vec{k}}$ for each $\vec{k} \in \Omega^*$ and positive integer $s$ by
\begin{equation}
    H^s_{\vec{k}} := \left\{ f \in H^s(\Omega;\field{C}^4) : f(\vec{r} + \vec{v}) = e^{i \vec{k} \cdot \vec{v}} \diag(1,e^{i \vec{q}_1 \cdot \vec{v}},1,e^{i \vec{q}_1 \cdot \vec{v}}) f(\vec{r})\ \forall \vec{v} \in \Lambda \right\}.
\end{equation}

We claim the following.
\begin{proposition}
For each fixed $\vec{k} \in \Omega^*$ and $\alpha \geq 0$, $H^\alpha$, defined on the domain $H^1_{\vec{k}}$, extends to an unbounded self-adjoint elliptic operator $L^2_{\vec{k}} \rightarrow L^2_{\vec{k}}$ with compact resolvent. In a complex neighborhood of every $\alpha \geq 0$, the family $H^\alpha$ is a holomorphic family of type (A) in the sense of Kato\cite{Kato}.
\end{proposition}
\begin{proof}
Ellipticity is immediate since the principal symbol of $H^\alpha$ is invertible. Self-adjointness is clear using the Fourier transform when $\alpha = 0$, and for $\alpha \neq 0$ because $\alpha H^1$ is a bounded symmetric perturbation of $H^0$ (see e.g. Theorem 1.4 of Cycon et al.\cite{CyconFroeseKirschSimon}). Elliptic regularity implies that the resolvent maps $L^2_{\vec{k}} \rightarrow H^1_{\vec{k}}$, and compactness of the resolvent then follows by Rellich's theorem (see e.g. Proposition 3.4 of Taylor \cite{Taylor1}). The family $H^\alpha$ is holomorphic of type (A) since the domain of $H^\alpha$ is independent of $\alpha$, and $H^\alpha f$ is holomorphic for every $f \in H^1_{\vec{k}}$ (see Kato Chapter 7 \cite{Kato}). 
\end{proof}

We now claim the following.
\begin{proposition} \label{prop:k_rot}
Let $f \in L^2_{\vec{k}}$. Then $\mathcal{R} f \in L^2_{R_\phi^* \vec{k}}$.
\end{proposition}
\begin{proof}
By definition, for any $\vec{v} \in \Lambda$,
\begin{equation}
    \mathcal{R} f(\vec{r} + \vec{v}) = \diag(1,1,e^{- i \phi},e^{- i \phi}) f( R_\phi \vec{r} + R_\phi \vec{v} ).
\end{equation}
By the definition of $L^2_{\vec{k}}$ we have
\begin{equation}
    \mathcal{R} f(\vec{r} + \vec{v}) = e^{i (R_\phi^* \vec{k}) \cdot \vec{v}} \diag(1,e^{i (R_\phi^* \vec{q}_1) \cdot \vec{v}},1,e^{i (R_\phi^* \vec{q}_1) \cdot \vec{v}}) \mathcal{R} f( \vec{r} ).
\end{equation}
The conclusion now follows from $R_\phi^* \vec{q}_1 = \vec{q}_1 + \vec{b}_2$ and $\vec{b}_2 \cdot \vec{v} = 0 \mod 2 \pi$ for all $\vec{v} \in \Lambda$.
\end{proof}
In particular, whenever $R_\phi^* \vec{k} = \vec{k}$ mod $\Lambda^*$, we have $\mathcal{R} L^2_{\vec{k}} = L^2_{\vec{k}}$. Regarding such $\vec{k}$, the following is a simple calculation.
\begin{proposition}
The moir\'e $K$ and $K'$ points $\vec{k} = 0$ and $\vec{k} = - \vec{q}_1$, and the moir\'e $\Gamma$ point $\vec{k} = \vec{q}_1 + \vec{b}_1$ satisfy $R_\phi^* \vec{k} = \vec{k}$ mod $\Lambda^*$.
\end{proposition}
The moir\'e $K$, $K'$, and $\Gamma$ points are shown in Figure \ref{fig:mBZ}. Note that the moir\'e $K$, $K'$, and $\Gamma$ points should not be confused with the single layer $K$, $K'$, and $\Gamma$ points. The moir\'e $K$ point corresponds to the $K$ point of layer 1, while the moir\'e $K'$ point corresponds to the $K$ point of layer 2. Interactions with the $K'$ points of layers 1 and 2 are formally small for small twist angles and are hence ignored.

In this work we will be particularly interested in Bloch functions at the moir\'e ${K}$ and ${K}'$ points. We therefore define
\begin{definition} \label{def:L2K_spaces}
\begin{equation}
    L^2_K := L^2_{\vec{0}}, \quad L^2_{K'} := L^2_{-\vec{q}_1}.
\end{equation}
\end{definition}
Let $\omega = e^{i \phi}$. Since the spaces $L^2_K$ and $L^2_{K'}$ are invariant under $\mathcal{R}$ they can be divided up into invariant subspaces corresponding to the eigenvalues of $\mathcal{R}$
\begin{equation}
    L^2_K = L^2_{K,1} \oplus L^2_{K,\omega} \oplus L^2_{K,\omega^*}, \quad L^2_{K'} = L^2_{K',1} \oplus L^2_{K',\omega} \oplus L^2_{K',\omega^*},
\end{equation}
where
\begin{equation} \label{eq:L2Ksigma}
    L^2_{K,\sigma} := \left\{ f \in L^2_K : \mathcal{R} f = \sigma f \right\} \quad \sigma = 1,\omega,\omega^*
\end{equation}
and $L^2_{K',\sigma}, \sigma = 1,\omega,\omega^*$, are defined similarly.

The following, which is trivial to prove, will be important for studying the zero modes of $H^\alpha$.
\begin{proposition}
The operator $\mathcal{S}$ commutes with ${\tau}_{\vec{v}}$ and $\mathcal{R}$ and hence maps the $L^2_{K,\sigma}$ and $L^2_{K',\sigma}$ spaces to themselves for $\sigma = 1,\omega,\omega^*$.
\end{proposition}
Since $\mathcal{S}$ has eigenvalues $\pm 1$, we can define the spaces
\begin{equation}
    L^2_{K,\sigma,\pm 1} = \left\{ f \in L^2_{K,\sigma} : \mathcal{S} f = \pm f \right\} \quad \sigma = 1,\omega,\omega^*
\end{equation}
and spaces $L^2_{K',\sigma,\pm 1}, \sigma = 1,\omega,\omega^*$ similarly. 
\begin{remark} \label{rem:A_B_sites}
Note that $+1$ eigenspaces of $\mathcal{S}$ correspond to wave-functions which vanish in their third and fourth entries, which correspond, through \eqref{eq:psi_densities}, to wave-functions supported only on $A$ sites of the layers. Similarly, $-1$ eigenspaces of $\mathcal{S}$ correspond to wave-functions which vanish in their first and second entries, which are supported only on $B$ sites of the layers.
\end{remark}

\begin{figure}
\centering
\includegraphics[scale=.5]{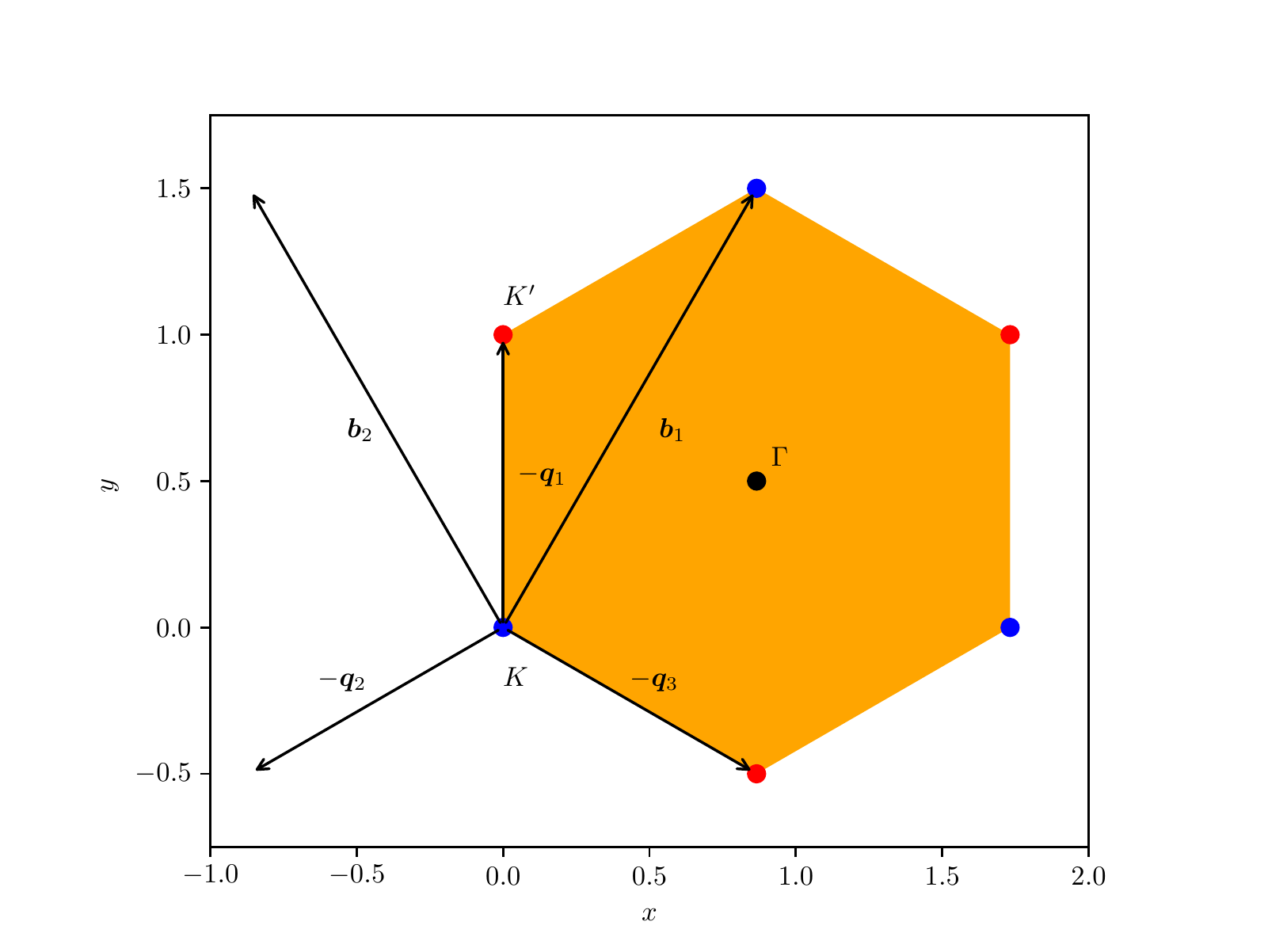}
\caption{Diagram showing locations of moir\'e $K$ (blue), $K'$ (red), and $\Gamma$ (black) points within the moir\'e Brillouin zone (orange).}
\label{fig:mBZ}
\end{figure}

\subsection{Zero modes of the chiral model}

We now want to investigate zero modes of $H^\alpha$ in detail. When $\alpha = 0$, there are exactly four zero modes given by ${e}_j , j = 1,2,3,4$ where ${e}_j$ equals $1$ in its $j$th entry and $0$ in its other entries. It is easy to check that
\begin{equation} \label{eq:0_zero_modes}
    e_1 \in L^2_{K,1}, \quad e_2 \in L^2_{K',1}, \quad e_3 \in L^2_{K,\omega^*}, \quad e_4 \in L^2_{K',\omega^*},
\end{equation}
and hence $0$ is a simple eigenvalue of $H^\alpha$ when restricted to each of these subspaces. Recall that zero modes can always be chosen as eigenfunctions of $\mathcal{S}$, and indeed we have
\begin{equation} \label{eq:0_zero_modes_S}
    e_1 \in L^2_{K,1,1}, \quad e_2 \in L^2_{K',1,1}, \quad e_3 \in L^2_{K,\omega^*,-1}, \quad e_4 \in L^2_{K',\omega^*,-1}.
\end{equation}
We now claim that these zero modes persist for all $\alpha$. This was already established by TKV \cite{Tarnopolsky2019}, and the following proposition is also similar to Proposition 3.1 of Becker et al.\cite{becker2020mathematics}. We re-state it using our notation and give the proof for completeness.
\begin{proposition} \label{prop:analytic_zeromodes}
There exist smooth functions $\psi^\alpha$ with $\| \psi^\alpha \| = 1$ in each of the spaces $L^2_{K,1,1}$, $L^2_{K',1,1}$, $L^2_{K,\omega^*,-1}$, $L^2_{K',\omega^*,-1}$ such that $\psi^0$ is as in \eqref{eq:0_zero_modes}, $\alpha \mapsto \psi^\alpha$ is real-analytic, and $H^\alpha \psi^\alpha = 0$ for all $\alpha$. The dimension of $\ker H^\alpha$ restricted to each of the spaces $L^2_{K,1}$, $L^2_{K',1}$, $L^2_{K,\omega^*}$, $L^2_{K,\omega^*}$ is always odd-dimensional.
\end{proposition}
\begin{proof}

Since $\mathcal{S}$ preserves each of the spaces $L^2_{K,1}$, $L^2_{K',1}$, $L^2_{K,\omega^*}$, $L^2_{K,\omega^*}$ and anti-commutes with $H^\alpha$, the spectrum of $H^\alpha$ restricted to each space must be symmetric about $0$ for all $\alpha$. Since $H^\alpha$ restricted to each space has compact resolvent and $H^\alpha$ is a holomorphic family of type (A), the spectrum of $H^\alpha$ consists of finitely-degenerate isolated eigenvalues depending real-analytically on $\alpha$, with associated eigenfunctions also depending real-analytically on $\alpha$ (although the real-analytic choice of eigenfunction at an eigenvalue crossing may not respect ordering); see Theorem 3.9 of Chapter 7 of Kato\cite{Kato}. The null space of $H^\alpha$ in each of the spaces is one-dimensional at $\alpha = 0$ by explicit calculation, with the zero modes given by (up to non-zero constants) \eqref{eq:0_zero_modes}. For small $\alpha > 0$, real-analyticity and the chiral symmetry force the null space to remain simple and it is clear how to define $\psi^\alpha$. For large $\alpha > 0$, the non-zero eigenvalues of $H^\alpha$ may cross $0$ at isolated values of $\alpha$, and in this case we define $\psi^\alpha$ to be the real-analytic continuation of the zero mode through the crossings. Note that real-analyticity prevents non-zero eigenvalues from equalling zero except at isolated points, so that the real-analytic continuation of the zero mode through the crossing must indeed be a zero mode. At crossings, the null space must be odd-dimensional in order to preserve symmetry of the spectrum of $H^\alpha$ about $0$. It remains to check that if $\psi^0$ is in, say, $L^2_{K,1,1}$, then $\psi^\alpha$ must remain in $L^2_{K,1,1}$ for all $\alpha > 0$. But this must hold because the $\mathcal{S}$-eigenvalue of $\psi^\alpha$ cannot change abruptly while preserving real-analyticity. Smoothness of $\psi^\alpha$ follows from elliptic regularity.
\end{proof}

In this work we will restrict attention to the moir\'e $K$ point, and especially the family $\psi^\alpha \in L^2_{K,1,1}$. We expect that our analysis would go through with only minor modifications if we considered instead the moir\'e $K'$ point. The zero modes in $L^2_{K,1,1}$ and $L^2_{K,\omega^*,-1}$ are related by the following symmetry.

\begin{proposition} \label{prop:K_symmetry}
Let $\psi^\alpha_{1}$ and $\psi^\alpha_{-1}$ denote the zero modes of $H^\alpha$ in the spaces $L^2_{K,1,1}$ and $L^2_{K,\omega^*,-1}$ respectively. Then $\psi^\alpha_1 = ( \Phi^\alpha , 0 )^\top$ where $\Phi^\alpha \in L^2(\Omega;\field{C}^2)$, $\Phi^\alpha(\vec{r} + \vec{v}) = \diag(1,e^{i \vec{q}_1 \cdot \vec{v}}) \Phi^\alpha(\vec{r})$ for all $\vec{v} \in \Lambda$, $\Phi^\alpha(R_\phi \vec{r}) = \Phi^\alpha(\vec{r})$. Up to gauge transformations $\psi^\alpha_{-1} \mapsto e^{i \phi(\alpha)} \psi^\alpha_{-1}$ which preserve real-analyticity of $\psi^\alpha_{-1}$, we have $\psi^\alpha_{-1}(\vec{r}) = ( 0, \Phi^{\alpha *}(-\vec{r}) )^\top$.
\end{proposition}
\begin{proof}
Since $\mathcal{S} \psi^\alpha_1 = \psi^\alpha_1$, the last two entries of $\psi^\alpha_1$ must vanish, so we can write $\psi^\alpha_1 = ( \Phi^\alpha , 0 )^\top$. That $\Phi^\alpha$ satisfies the stated symmetries follows immediately from $\psi^\alpha_1 \in L^2_{K,1}$. It is straightforward to check using the definitions of $\mathcal{R}$ and $\tau_{\vec{v}}$ that $( 0, \Phi^{\alpha *}(-\vec{r}) )^\top \in L^2_{K,\omega^*,-1}$. To see that $( 0, \Phi^{\alpha *}(-\vec{r}) )^\top$ is a zero mode, note that $\Phi^\alpha$ satisfies $D^\alpha \Phi^\alpha = 0$, which implies that $D^{\alpha \dagger} \Phi^{\alpha *}(-\vec{r}) = 0$ by a simple manipulation. To see that $\psi^\alpha_{-1}(\vec{r}) = ( 0, \Phi^{\alpha *}(-\vec{r}) )^\top$ (up to real-analytic gauge transformations) for all $\alpha$, note first that this clearly holds for $\alpha = 0$ (the zero modes are explicit \eqref{eq:0_zero_modes_S}). For $\alpha > 0$, the identity must continue to hold by uniqueness (up to real-analytic gauge transformations) of the real-analytic continuation of $\psi^\alpha_{-1}$ starting from $\alpha = 0$ and continuing first along the non-zero interval where $\psi^\alpha_1$ is non-degenerate in $L^2_{K,1,1}$, and then through eigenvalue crossings as in the proof of Proposition \ref{prop:analytic_zeromodes}.
\end{proof}

In Appendix \ref{sec:K_Dirac} we use Proposition \ref{prop:K_symmetry} to derive the effective Dirac operator with $\alpha$-dependent Fermi velocity which controls the Bloch band structure in a neighborhood of the moir\'e $K$ point. The Fermi velocity of the effective Dirac operator is given by the following. Note that we drop the subscript $+1$ when referring to the zero mode of $H^\alpha$ in $L^2_{K,1,1}$ since the zero mode of $H^\alpha$ in $L^2_{K,\omega^*,-1}$ plays no further role.
\begin{definition} \label{def:Fermi_v}
Let $\psi^\alpha \in L^2_{K,1,1}$ be as in Proposition \ref{prop:analytic_zeromodes}. Then we define
\begin{equation} \label{eq:Fermi_v_2}
    v(\alpha) := \frac{ | \ip{\psi^{\alpha*}(-\vec{r})}{\psi^\alpha(\vec{r})} | }{ | \ip{\psi^\alpha}{\psi^\alpha} | }
\end{equation}
where $\ip{.}{.}$ denotes the $L^2_K$ inner product.
\end{definition}

\section{Rigorous justification of TKV's expansion of the Fermi velocity} \label{sec:rigorous_expansion}

\subsection{Alternative formulation of TKV's expansion}

We now turn to approximating the zero mode $\psi^\alpha \in L^2_{K,1,1}$ by a series expansion in powers of $\alpha$. We write $H^\alpha = H^0 + \alpha H^1$ and formally expand $\psi^\alpha$ as a series
\begin{equation} \label{eq:psi_series}
    \psi^\alpha = \Psi^0 + \alpha \Psi^1 + ...
\end{equation}
where $H^0 \Psi^0 = 0$, and
\begin{equation} \label{eq:series_eqs}
    H^0 \Psi^n = - H^1 \Psi^{n-1}
\end{equation}
for all $n \geq 1$. To solve $H^0 \Psi^0 = 0$ we take $\Psi^0 = e_1$. We prove the following in Appendix \ref{sec:TKV_expansion}.
\begin{proposition} \label{prop:series_prop}
Let $P$ denote the projection operator in $L^2_{K,1}$ onto $e_1$, and $P^\perp = I - P$. The sequence of equations \eqref{eq:series_eqs} has a unique solution such that $\Psi^n \in L^2_{K,1,1}$ for all $n \geq 0$ and $P \Psi^n = 0$ for all $n \geq 1$ given by $\Psi^0 = e_1$ and
\begin{equation} \label{eq:H0_inv}
    \Psi^n = - P^\perp (H^{0})^{-1} P^\perp H^1 \Psi^{n-1}
\end{equation}
for each $n \geq 1$.
\end{proposition}
The expansion \eqref{eq:psi_series} appears different from the series studied by TKV, since we work only with the self-adjoint operators $H^0, H^1$, and $H^\alpha$ rather than the non-self-adjoint operator $D^\alpha$ (defined in \eqref{eq:chiral_H}). Since functions in $L^2_{K,1,1}$ vanish in their last two components, there is no practical difference. However, working with only self-adjoint operators allows us to use the spectral theorem, which greatly simplifies the error analysis. We compute the first eight terms in expansion \eqref{eq:psi_series} in Proposition \ref{prop:psi_alpha_expansion} after developing some necessary machinery in Appendix \ref{sec:chiral}.


\subsection{Rigorous error estimates for the expansion of the moir\'e $K$ point Bloch function}

In this section we explain the essential challenge in proving error estimates for the series \eqref{eq:psi_series} and explain how we overcome this challenge. Our goal is to prove the following. 
\begin{theorem} \label{th:error_theorem_2}
Let $\psi^\alpha \in L^2_{K,1,1}$ be as in Proposition \ref{prop:analytic_zeromodes}. Then 
\begin{equation} \label{eq:expansion_2}
    \psi^\alpha = \sum_{n = 1}^8 \alpha^n \Psi^n + \eta^\alpha
\end{equation}
where $\eta^\alpha \perp \sum_{n = 1}^8 \alpha^n \Psi^n$ with respect to the $L^2_K$ inner product, and
\begin{equation}\label{eq:ub_2}
    \| \eta^\alpha \|_{L^2_{K,1}} \leq \frac{ 3 \alpha^9 }{ 15 - 20 \alpha }
    \quad \text{for all}\quad 0 \leq \alpha \leq \frac{7}{10}.
\end{equation}
\end{theorem}

Proposition \ref{prop:series_prop} guarantees that the series \eqref{eq:psi_series} is well-defined up to arbitrarily many terms. A straightforward bound on the growth of terms in the series comes from the following proposition.
\begin{proposition} \label{prop:H1H0bound}
The spectrum of $H^0$ in $L^2_{K,1}$ is
\begin{equation}
    \sigma_{L^2_{K,1}}(H^0) = \{ \pm |\vec{G}|, \pm |\vec{q}_1 + \vec{G}| : \vec{G} \in \Lambda^* \}
\end{equation}
and hence
\begin{equation} \label{eq:H0_worst}
    \| P^\perp (H^0)^{-1} P^\perp \|_{L^2_{K,1} \rightarrow L^2_{K,1}} = 1.
\end{equation}
We also have
\begin{equation} \label{eq:H1_worst}
    \| H^1 \|_{L^2_{K,1} \rightarrow L^2_{K,1}} = 3.
\end{equation}
\end{proposition}
\begin{proof}
This proposition is a combination of Propositions \ref{prop:L2K1}, \ref{prop:H0_prop}, and \ref{prop:H1_norm_prop}, proved in Appendix \ref{sec:chiral}.
\end{proof}
Proposition \ref{prop:H1H0bound} implies that $\| P^\perp (H^0)^{-1} P^\perp H^1 \|_{L^2_{K,1} \rightarrow L^2_{K,1}} \leq 3$, which implies the following.
\begin{proposition} \label{prop:series_convergence}
The formal series \eqref{eq:psi_series} converges to $\psi^\alpha$ in $L^2_{K}$, with an explicit error rate, for all $|\alpha| < \frac{1}{3}$. The formal series for the Fermi velocity $v(\alpha)$ obtained by substituting the series expansion of $\psi^\alpha$ into \eqref{eq:Fermi_v_2} converges for the same range of $\alpha$, also with an explicit error rate.
\end{proposition}
\begin{proof}
For any non-negative integer $N$, let $\psi^{N,\alpha} := \sum_{n = 0}^N \alpha^n \Psi^n$ where the $\Psi^n$ are as in \eqref{eq:H0_inv}. Since $\Psi^0 \perp \Psi^n$ for all $n \geq 1$ and $\| \Psi^0 \| = 1$, we have that $\| \psi^{N,\alpha} \| \geq 1$ for all $N$. Let $\phi^{N,\alpha} := \frac{\psi^{N,\alpha}}{\| \psi^{N,\alpha} \|}$, then we can decompose $\phi^{N,\alpha} = c \psi^\alpha + \eta^\alpha$ for some constant $c$ and where $\eta^\alpha \perp \psi^\alpha$. Applying $H^\alpha$ to both sides we have that $H^\alpha \phi^{N,\alpha} = \frac{ \alpha^{N+1} H^1 \Psi^N }{ \| \psi^{N,\alpha} \| } = H^\alpha \eta^\alpha$. Now fix $\alpha \geq 0$ such that $|\alpha| < \frac{1}{3}$. Then $\alpha \| H^1 \| < 1$ and hence the first non-zero eigenvalue of $H^\alpha$ is bounded away from $0$ by $1 - 3 \alpha$ (recall that the first non-zero eigenvalues of $H^0$ are $\pm 1$). Since $\eta^\alpha \perp \psi^\alpha$, where $\psi^\alpha$ spans the eigenspace of the zero eigenvalue of $H^\alpha$, we have that $\| \eta^\alpha \| \leq \frac{ |\alpha^{N+1}| \| H^1 \Psi^N \| }{ |1 - 3 \alpha| \| \psi^{N,\alpha} \| }$. Using the bound $\| \Psi^N \| \leq (3 \alpha)^N$ and the bound below on $\| \psi^{N,\alpha} \|,$ we have that $\| \eta^\alpha \| \leq \frac{ (3 \alpha)^{N+1} }{ |1 - 3 \alpha| }$ which clearly $\rightarrow 0$ as $N \rightarrow \infty$, so that $\lim_{N \rightarrow \infty} \phi^{N,\alpha} = \psi^\alpha$ (up to a non-zero constant). Now consider $v(\alpha)$ defined by \eqref{eq:Fermi_v_2}. Assuming WLOG that $\| \psi^\alpha \| = 1$, substituting $\psi^\alpha = \phi^{N,\alpha} + \eta^\alpha$ we find immediately, using Cauchy-Schwarz, that $\left| v(\alpha) - \ip{ \phi^{N,*}(-\vec{r}) }{ \phi^N(\vec{r}) } \right| \leq 2 \| \eta^\alpha \| + \| \eta^\alpha \|^2$. In terms of $\psi^N$ we have $\left| v(\alpha) - \frac{ \ip{ \psi^{N,*}(-\vec{r}) }{ \psi^N(\vec{r}) } }{ \ip{ \psi^N }{ \psi^N } } \right| \leq 2 \| \eta^\alpha \| + \| \eta^\alpha \|^2$.
\end{proof}
Proposition \ref{prop:series_convergence} shows that for $|\alpha| < \frac{1}{3}$ the series \eqref{eq:psi_series} converges to $\psi^\alpha$ and can be used to compute the Fermi velocity. However, this restriction is too strong to prove that the Fermi velocity has a zero, which occurs at the larger value $\alpha \approx \frac{1}{\sqrt{3}}$. Of course, Proposition \ref{prop:H1H0bound} establishes only the most pessimistic possible bound on the expansion functions $\Psi^n$, and this bound appears to be far from sharp from explicit calculation of each $\Psi^n$, see Proposition \ref{prop:Psi_n_norms}. We briefly discuss a possible route to a tighter bound in Remark \ref{rem:remark_on_PsiN_norms}, but do not otherwise pursue this approach in this work.

We now explain how to obtain error estimates over a large enough range of $\alpha$ values to prove $v(\alpha)$ has a zero. We seek a solution of $H^\alpha \psi^\alpha = 0$ in $L^2_{K,1,1}$ with the form
\begin{equation} \label{eq:approx_sol}
    \psi^\alpha = \psi^{N,\alpha} + \eta^\alpha, \quad \psi^{N,\alpha} := \sum_{n = 0}^N \alpha^n \Psi^n.
\end{equation}
For arbitrary $\alpha$, let $Q^{\alpha}$ denote the projection in $L^2_{K,1}$ onto $\psi^{N,\alpha}$, and $Q^{\alpha,\perp} := I - Q^{\alpha}$ (note that $Q^0 = P$). Note that $Q^\alpha$ depends on $N$ but we suppress this to avoid clutter. We assume WLOG that $Q^{\alpha} \eta^\alpha(\vec{r}) = 0$. It follows that $\eta^\alpha$ satisfies
\begin{equation} \label{eq:eta_eq}
    Q^{\alpha,\perp} H^\alpha Q^{\alpha,\perp} \eta^\alpha = - \alpha^{N+1} Q^{\alpha,\perp} H^1 \Psi^N.
\end{equation}
To obtain a bound on $\eta^\alpha$ in $L^2(\Omega)$, we require a lower bound on the operator $Q^{\alpha,\perp} H^\alpha Q^{\alpha,\perp} : Q^{\alpha,\perp} L^2_{K,1} \rightarrow Q^{\alpha,\perp} L^2_{K,1}$. The following Lemma gives a lower bound on this operator in terms of a lower bound on the projection of this operator onto the finite dimensional subspace of $L^2_{K,1}$ corresponding to a finite subset of the eigenfunctions of $H^0$. The importance of this result is that, since $H^1$ only couples finitely many modes of $H^0$, for \emph{fixed} $N$, by taking the subset sufficiently large, we can always arrange that $\psi^{N,\alpha}$ lies in this subspace. 

\begin{lemma} \label{lem:decompose}
Let $P_\Xi$ denote the projection onto a subset $\Xi$ of the eigenfunctions of $H^0$ in $L^2_{K,1}$, and let $\mu \geq 0$ be maximal such that 
\begin{equation} \label{eq:mu_def}
    \| P^\perp_\Xi H^0 P^\perp_\Xi f \| \geq \mu \| f \| \quad \forall f \in H^1_{K,1}, \quad P^\perp_\Xi := I - P_\Xi,
\end{equation}
(with this notation the operator $P$ introduced in Proposition \ref{prop:series_prop} corresponds to $P_\Xi$ with $\Xi$ being the set $\{ e_1 \}$ and $\mu = 1$). Suppose that $Q^{\alpha} P_\Xi = P_\Xi Q^{\alpha} = Q^\alpha$, i.e., that $\psi^{N,\alpha}$ lies in $\ran P_\Xi$.
Define $g^\alpha$ by
\begin{equation} \label{eq:g_def}
\begin{split}
    g^\alpha &:= \min \left\{ |E| : \parbox{22.2em}{ $E$ is an eigenvalue of the matrix $Q^{\alpha,\perp} P_\Xi H^\alpha P_\Xi Q^{\alpha,\perp}$ \\ acting $Q^{\alpha,\perp} P_\Xi L^2_{K,1} \rightarrow Q^{\alpha,\perp} P_\Xi L^2_{K,1}$ } \right\}.
\end{split}
\end{equation}
We note that $P_\Xi Q^{\alpha,\perp}$ is the projection onto the subspace of $P_\Xi L^2_{K,1}$ orthogonal to $\psi^{N,\alpha}$.
As long as
\begin{equation}
    3 \alpha \leq \mu \text{ and } \alpha \| Q^{\alpha,\perp} P_\Xi H^1 P_\Xi^\perp \| < \min( g^\alpha , \mu - 3 \alpha ),
\end{equation}
then
\begin{equation} \label{eq:bound_below}
    \|Q^{\alpha,\perp} H^\alpha Q^{\alpha,\perp} \eta^\alpha \|
    \geq \left( \min(g^\alpha,\mu - 3 \alpha) - \alpha\| Q^{\alpha,\perp} P_\Xi H^1 P_\Xi^\perp \|  \right) \| Q^{\alpha,\perp} \eta^\alpha \|.
\end{equation}
\end{lemma}
Note that $g^\alpha$ would be identically zero if not for the restriction that the matrix acts on $Q^{\alpha,\perp} P_\Xi L^2_{K,1}$, since otherwise $\psi^{N,\alpha}$ would be an eigenfunction with eigenvalue zero for all $\alpha$. As it is, $g^0 = 1$ and $\alpha \mapsto g^\alpha$ is real-analytic so that $g^\alpha$ must be positive for a non-zero interval of positive $\alpha$ values.
\begin{proof}
Using $Q^{\alpha} P_\Xi = P_\Xi Q^\alpha$ we have $P_\Xi^\perp Q^{\alpha,\perp} = Q^{\alpha,\perp} P_\Xi^\perp = P_\Xi^\perp$ and hence
\begin{equation}
\begin{split}
    &\| Q^{\alpha,\perp} H^\alpha Q^{\alpha,\perp} \eta^\alpha \| = \| Q^{\alpha,\perp} (P_\Xi + P_\Xi^\perp) H^\alpha (P_\Xi + P_\Xi^\perp) Q^{\alpha,\perp} \eta^\alpha \| \\
    &= \| Q^{\alpha,\perp} P_\Xi H^\alpha P_\Xi Q^{\alpha,\perp} \eta^\alpha + \alpha Q^{\alpha,\perp} P_\Xi H^1 P_\Xi^\perp \eta^\alpha + \alpha P_\Xi^\perp H^\alpha P_\Xi Q^{\alpha,\perp} \eta^\alpha + P_\Xi^\perp H^\alpha P_\Xi^\perp \eta^\alpha \|.
\end{split}
\end{equation}
By the reverse triangle inequality
\begin{align}
    &\| Q^{\alpha,\perp} H^\alpha Q^{\alpha,\perp} \eta^\alpha \|  \label{eq:rev_tri} \\ \notag
    &\geq \left| \| Q^{\alpha,\perp} P_\Xi H^\alpha P_\Xi Q^{\alpha,\perp} \eta^\alpha + P_\Xi^\perp H^\alpha P_\Xi^\perp \eta^\alpha \| - \alpha \| Q^{\alpha,\perp} P_\Xi H^1 P_\Xi^\perp \eta^\alpha + P_\Xi^\perp H^\alpha P_\Xi Q^{\alpha,\perp} \eta^\alpha \| \right|.
\end{align}
We want to bound the second term above and the first term below. We start with the second term
\begin{equation}
\begin{split}
    &\| Q^{\alpha,\perp} P_\Xi H^1 P_\Xi^\perp \eta^\alpha + P_\Xi^\perp H^\alpha P_\Xi Q^{\alpha,\perp} \eta^\alpha \|^2 \\
    &= \| Q^{\alpha,\perp} P_\Xi H^1 P_\Xi^\perp \eta^\alpha \|^2 + \| P_\Xi^\perp H^1 P_\Xi Q^{\alpha,\perp} \eta^\alpha \|^2    \\
    &\leq \| Q^{\alpha,\perp} P_\Xi H^1 P_\Xi^\perp \|^2 \left( \| P_\Xi^\perp \eta^\alpha \|^2+\| P_\Xi Q^{\alpha,\perp} \eta^\alpha \|^2  \right)  \\
    &= \| Q^{\alpha,\perp} P_\Xi H^1 P_\Xi^\perp \|^2 \| Q^{\alpha,\perp} \eta^\alpha \|^2,
\end{split}
\end{equation}
where we use Pythagoras' theorem, $P_\Xi^\perp H^1 P_\Xi Q^{\alpha,\perp} \eta^\alpha=P_\Xi^\perp H^1 P_\Xi Q^{\alpha,\perp} P_\Xi Q^{\alpha,\perp}\eta^\alpha$ since $P_\Xi Q^{\alpha,\perp}$ is a projection, and $\| Q^{\alpha,\perp} P_\Xi H^1 P_\Xi^\perp \| = \| P_\Xi^\perp H^1 P_\Xi Q^{\alpha,\perp} \|$. Hence we can bound
\begin{equation} \label{eq:bd_above}
    \| Q^{\alpha,\perp} P_\Xi H^1 P_\Xi^\perp \eta^\alpha + P_\Xi^\perp H^\alpha P_\Xi Q^{\alpha,\perp} \eta^\alpha \| \leq \| Q^{\alpha,\perp} P_\Xi H^1 P_\Xi^\perp \| \| Q^{\alpha,\perp} \eta^\alpha \|.
\end{equation}
For the first term, first note that using Proposition \ref{prop:H1H0bound} and the spectral theorem
\begin{equation}
\begin{split}
    \| Q^{\alpha,\perp} P^\perp_\Xi H^\alpha P^\perp_\Xi Q^{\alpha,\perp} \eta^\alpha \| &\geq | \| Q^{\alpha,\perp} P^\perp_\Xi H^0 P^\perp_\Xi Q^{\alpha,\perp} \eta^\alpha \| - \alpha \| Q^{\alpha,\perp} P^\perp_\Xi H^1 P^\perp_\Xi Q^{\alpha,\perp} \eta^\alpha \| | \\
    &\geq (\mu - 3 \alpha) \| P^\perp_\Xi Q^{\alpha,\perp} \eta^\alpha \|
\end{split}
\end{equation}
as long as $\mu \geq 3 \alpha$. We now estimate
\begin{equation}
\begin{split}
    &\| Q^{\alpha,\perp} P_\Xi H^\alpha P_\Xi Q^{\alpha,\perp} \eta^\alpha + P_\Xi^\perp H^\alpha P_\Xi^\perp \eta^\alpha \|^2  \\
    &= \| Q^{\alpha,\perp} P_\Xi H^\alpha P_\Xi Q^{\alpha,\perp} \eta^\alpha \|^2 + \| P_\Xi^\perp H^\alpha P_\Xi^\perp \eta^\alpha \|^2  \\
    &\geq (g^\alpha)^2 \| Q^{\alpha,\perp} P_\Xi \eta^\alpha \|^2 + (\mu  - 3 \alpha)^2 \| P_\Xi^\perp \eta^\alpha \|^2    \\
    &\geq \min\left((g^\alpha)^2,(\mu - 3 \alpha)^2\right) \left( \| Q^{\alpha,\perp} P_\Xi \eta^\alpha \|^2 + \|P_\Xi^\perp \eta^\alpha \|^2 \right)   \\
    &= \min\left((g^\alpha)^2,(\mu - 3 \alpha)^2\right) \| Q^{\alpha,\perp} \eta^\alpha \|^2.
\end{split}
\end{equation}
It follows that as long as $3 \alpha \leq \mu$,
\begin{equation} \label{eq:bd_below}
    \| Q^{\alpha,\perp} P_\Xi H^\alpha P_\Xi Q^{\alpha,\perp} \eta^\alpha + P_\Xi^\perp H^\alpha P_\Xi^\perp \eta^\alpha \| \geq \min(g^\alpha,\mu-3\alpha) \|Q^{\alpha,\perp} \eta^\alpha\|.
\end{equation}
The conclusion now holds as long as $3 \alpha \leq \mu$ and $\alpha \| Q^{\alpha,\perp} P_\Xi H^1 P_\Xi^\perp \| \leq \min(g^\alpha,\mu-3 \alpha)$ upon substituting \eqref{eq:bd_above} and \eqref{eq:bd_below} into \eqref{eq:rev_tri}.
\end{proof}

For Lemma \ref{lem:decompose} to be useful, we must check that it is possible to choose $\Xi$ so that the bound \eqref{eq:bound_below} is non-trivial, i.e., so that the constant is positive. We will prove the following in Appendix \ref{sec:prop_proof}.
\begin{proposition} \label{prop:mu_choice}
There exists a subset $\Xi$ of the eigenfunctions of $H^0$ such that 
\begin{enumerate}
\item The maximal $\mu$ such that \eqref{eq:mu_def} holds is $\mu = 7$.
\item $\psi^{8,\alpha}$ defined by \eqref{eq:approx_sol} lies in $\ran P_\Xi$.
\item $\| P_\Xi H^1 P_\Xi^\perp \| = 1$ and hence $\| Q^{\alpha,\perp} P_\Xi H^1 P_\Xi^\perp \| \leq 1$.
\end{enumerate}
\end{proposition}
The set $\Xi$ constructed in Proposition \ref{prop:mu_choice} is the set of $L^2_{K,1}$-eigenfunctions of $H^0$ with eigenvalues with magnitude $\leq 4 \sqrt{3}$, augmented with two extra basis functions to ensure that $\| P_\Xi H^1 P_\Xi^\perp \| = 1$. Including all $L^2_{K,1}$-eigenfunctions of $H^0$ with eigenvalue magnitudes up to and including $4 \sqrt{3}$ ensures that $\psi^{8,\alpha}$ lies in $\ran P_\Xi$.

We now require the following. 
\begin{proposition} \label{as:PHP_gap}
Let $\Xi$ be as in Proposition \ref{prop:mu_choice}. Then $g^\alpha \geq \frac{3}{4}$ for all $0 \leq \alpha \leq \frac{7}{10}$.
\end{proposition}
\begin{proof}[Proof (computer assisted)]
Consider $H^\alpha_{\Xi} := Q^{\alpha,\perp} P_{\Xi} H^\alpha P_{\Xi} Q^{\alpha,\perp}$ acting on $P_{\Xi} L^2_{K,1}$. Assuming $\alpha$ is restricted to an interval such that the zero eigenspace of $H^\alpha_{\Xi}$ is simple, then, using orthogonality of eigenvectors corresponding to different eigenvalues and the fact that $Q^{\alpha}$ is the spectral projection onto the unique zero mode of $H^\alpha_{\Xi}$, $H^\alpha_{\Xi}$ has the same non-zero eigenvalues as the matrix $Q^{\alpha,\perp} P_{\Xi} H^\alpha P_{\Xi} Q^{\alpha,\perp}$ acting on $Q^{\alpha,\perp} P_{\Xi} L^2_{K,1}$. The matrix $H^\alpha_{\Xi}$ is an $81 \times 81$ matrix whose spectrum is symmetric about $0$ because of the chiral symmetry. When $\alpha = 0$ the spectrum is explicit: $0$ is a simple eigenvalue, and the smallest non-zero eigenvalues are $\pm 1$, both also simple. Proposition \ref{as:PHP_gap} is proved if we can prove that the first positive eigenvalue of $H^\alpha_{\Xi}$ is bounded away from zero by $\frac{3}{4}$ for all $0 \leq \alpha \leq \frac{7}{10}$. Note that if this holds, the zero eigenspace of $H^\alpha_{\Xi}$ must be simple for all $0 \leq \alpha \leq \frac{7}{10}$ and hence our basic assumption is justified. The strategy of the proof is as follows. 
\begin{enumerate}
\item Define a grid $\mathcal{G} := \left\{ \frac{7 n}{10 N} : n \in \{0,1,...,N\} \right\}$, where $N$ is a positive integer taken sufficiently large that the grid spacing $h := \frac{7}{10 N} < \frac{1}{388831}$ (the number $388831$ comes from Proposition \ref{prop:bound_de_alpha_H}).
\item Numerically compute the eigenvalues of $H^\alpha_{\Xi}$ for $\alpha \in \mathcal{G}$. We find that the numerically computed first positive eigenvalues of these matrices are uniformly bounded below by $\frac{8}{10} > \frac{3}{4}$.
\item Perform a backwards error analysis which fully accounts for round-off error in the numerical computation in order to prove that the exact first positive eigenvalues of the matrices $H^\alpha_{\Xi}$ must also be bounded below by $\frac{8}{10}$ at each $\alpha \in \mathcal{G}$.
\item Use perturbation theory to bound the exact first positive eigenvalue of $H^\alpha_{\Xi}$ below by $\frac{3}{4}$ over the whole interval of $\alpha$ values between $0$ and $\frac{7}{10}$.
\end{enumerate}
When discussing round-off error due to working in floating point arithmetic, we will denote ``machine epsilon'' by $\epsilon$. The significance of this number is that we will assume that all complex numbers $a$ can be represented by floating-point numbers $\tilde{a}$ such that $|a - \tilde{a}| \leq \epsilon a$. We will also make the standard assumption about creation of round-off error in floating-point arithmetic operations: if $\tilde{a}$ and $\tilde{b}$ are floating-point complex numbers, and if $(\tilde{a} \mathcal{O} \tilde{b})_{comp}$ and $\tilde{a} \mathcal{O} \tilde{b}$ represent the numerically computed value and exact value of an arithmetic operation on the numbers $\tilde{a}$ and $\tilde{b}$, then $(\tilde{a} \mathcal{O} \tilde{b})_{comp} = \tilde{a} \mathcal{O} \tilde{b} + e$ where $|e| \leq (\tilde{a} \mathcal{O} \tilde{b}) \epsilon$. In Python this is indeed the case, for all reasonably sized (such that stack overflow does not occur) complex numbers, with $\epsilon = 2.22044 \times 10^{-16}$ (5sf). We now present the main points of parts 2.-4. of the strategy, postponing proofs of intermediate lemmas to Appendix \ref{sec:verify_gap_assump}. 

For part 2. of the strategy, for each $\alpha \in \mathcal{G}$, we let $\tilde{H}_\Xi^\alpha$ denote $H_\Xi^\alpha$ (which is known exactly) evaluated as floating-point numbers. We generate numerically computed eigenpairs $\tilde{\lambda}_j, \tilde{v}_j$ for $1 \leq j \leq 81$ for each $\tilde{H}^\alpha_{\Xi}$ using numpy's Hermitian eigensolver \texttt{eigh}. We find that the smallest first positive eigenvalue of $\tilde{H}^\alpha_\Xi$ for $\alpha \in \mathcal{G}$ is 0.8147191261445436 (computed using \texttt{compute\_PHalphaP\_enclosures.py} in the Github repo). Note that the difference between this number and $\frac{8}{10}$ is bounded below by $0.01$.

The main tool for part 3. of the strategy is the following theorem.
\begin{theorem} \label{th:first_1}
Let $m$ and $n$ denote positive integers with $m \leq n$. Let $A$ be a Hermitian $n \times n$ matrix, and let $\{ v_j \}_{1 \leq j \leq m}$ be orthonormal $n$-vectors satisfying $(A - \lambda_j I) v_j = r_j$ for scalars $\lambda_j$ and $n$-vectors $r_j$ for each $1 \leq j \leq m$. Then there are $m$ eigenvalues $\{ \alpha_j \}_{1 \leq j \leq m}$ of $A$ which can be put into one-to-one correspondence with the $\lambda_j$s such that 
\begin{equation} \label{eq:enclosure_ints}
    | \lambda_j - \alpha_j | \leq 2 m \sup_{1 \leq i \leq m} \| r_i \|_2 \quad \text{for all $1 \leq j \leq m$.}
\end{equation}
\end{theorem}
\begin{proof}
See Appendix \ref{sec:prove_first_1}.
\end{proof}
Na\"ively, one would hope to be able to calculate enclosure intervals for every eigenvalue of $H^\alpha_{\Xi}$, and in particular a lower bound on the first positive eigenvalue of $H^\alpha_{\Xi}$, by directly applying Theorem \ref{th:first_1} with $A = H^\alpha_{\Xi}$, $m = 81$, and $\lambda_j$ and $v_j$ given for each $1 \leq j \leq 81$ by the approximate eigenpairs $\tilde{\lambda}_j, \tilde{v}_j$ computed in part 2. However, we can't directly apply the theorem because the $\{ \tilde{v}_j \}_{1 \leq j \leq 81}$ aren't exactly orthonormal because of round-off error. So we will prove existence of an exactly orthonormal set $\{ \oldhat{v}_j \}_{1 \leq j \leq 81}$ close to the set $\{ \tilde{v}_j \}_{1 \leq j \leq 81}$ and apply Theorem \ref{th:first_1} to the set $\{ \oldhat{v}_j \}_{1 \leq j \leq 81}$ (with the same $\tilde{\lambda}_j$) instead. Note that to carry out this strategy we must bound the residuals $\oldhat{r}_j := ( H^\alpha_\Xi - \tilde{\lambda}_j ) \oldhat{v}_j$. The result we need to implement this strategy is the following. Note that the result requires numerical computation of inner products and residuals, and we account for round-off error in these computations.
\begin{theorem} \label{th:second_1}
Let $m$ and $n$ be positive integers with $m \leq n$. Let $A$ be an $n \times n$ Hermitian matrix, let $\tilde{A}$ denote $A$ evaluated in floating-point numbers, and let $\tilde{v}_j, \tilde{\lambda}_j$ for $1 \leq j \leq m$ be a set of $n$-dimensional vectors and real numbers respectively. Let $\ip{\tilde{v}_i}{\tilde{v}_j}_{comp}$ denote their numerically computed inner products, and let $\tilde{r}_{j,comp} := \left[ \left( \tilde{A} - \tilde{\lambda}_j I \right) \tilde{v}_j \right]_{comp}$ denote their numerically computed residuals. Let $\epsilon$ denote machine epsilon, and assume $n \epsilon < 0.01$. Let $\mu$ be 
\begin{equation} \label{eq:mu}
    \mu := (1.01) n^2 \epsilon \left( \sup_{1 \leq i \leq m} \| \tilde{v}_i \|_\infty \right)^2 + \sup_{1 \leq i \leq m} | \ip{ \tilde{v}_i }{ \tilde{v}_i }_{comp} - 1 | + \sup_{\substack{i \neq j \\ 1 \leq i, j \leq m}} | \ip{ \tilde{v}_i }{ \tilde{v}_j }_{comp} |.
\end{equation}
Then, as long as $m \mu < \frac{1}{2}$, there is an orthonormal set of $n$-vectors $\{ \oldhat{v}_j \}_{1 \leq j \leq m}$ whose residuals $\oldhat{r}_j := ( A - \tilde{\lambda}_j I ) \oldhat{v}_j$ satisfy the bound
\begin{equation} \label{eq:bound_1}
\begin{split}
    \sup_{1 \leq j \leq m} \| \oldhat{r}_j \|_2 \leq & \; 2^{-1/2} n \left( \| A \|_2 + \sup_{1 \leq j \leq m} | \tilde{\lambda}_j | \right) \mu + n^{1/2} \sup_{1 \leq j \leq m} \left\| \tilde{r}_{j,comp} \right\|_\infty \\
    &+ (1.01) n^{5/2} \epsilon \left( \| \tilde{A} \|_{max} + \sup_{1 \leq j \leq m} |\tilde{\lambda}_j| \right) \sup_{1 \leq j \leq m} \| \tilde{v}_j \|_\infty + n \epsilon \| A \|_{max} \sup_{1 \leq j \leq m} \| \tilde{v}_j \|_\infty,
\end{split}
\end{equation}
where $\| A \|_{max}$ denotes the largest of the absolute values of the elements of the matrix $A$.
\end{theorem}
\begin{proof}
See Appendix \ref{sec:proof_second_1}.
\end{proof}
Numerical computation (using the script \texttt{compute\_PHalphaP\_enclosures.py} in the Github repo) shows that the maximum of $\sup_{1 \leq i \leq m} | \ip{ \tilde{v}_i }{ \tilde{v}_i }_{comp} - 1 |$ and $\sup_{\substack{i \neq j \\ 1 \leq i, j \leq m}} | \ip{ \tilde{v}_i }{ \tilde{v}_j }_{comp} |$ over $\alpha \in \mathcal{G}$ is bounded by 7$\times 10^{-15}$. Hence we can apply Theorem \ref{th:second_1} with $A = H^\alpha_{\Xi}$ and $\tilde{\lambda}_j, \tilde{v}_j$ given by the numerically computed eigenpairs of $\tilde{H}^\alpha_{\Xi}$ to obtain orthonormal sets $\{ \oldhat{v}_j \}_{1 \leq j \leq 81}$ whose residuals with respect to $H^\alpha_{\Xi}$ satisfy \eqref{eq:bound_1}. The following is straightforward.
\begin{proposition} \label{prop:Halpha_bounds}
\begin{equation}
    \sup_{0 \leq \alpha \leq \frac{7}{10}} \| H^\alpha_{\Xi} \|_2 \leq 10, \quad \sup_{0 \leq \alpha \leq \frac{7}{10}} \| H^\alpha_{\Xi} \|_{max} \leq 7 
\end{equation}
\end{proposition}
\begin{proof}
The first estimate follows from $\| P_{\Xi} H^0 P_{\Xi} \| \leq 7$ and $\| H^1 \| \leq 3$. The second estimate follows immediately from writing the matrix $H^\alpha_{\Xi}$ in the chiral basis.
\end{proof}
We can now apply Theorem \ref{th:first_1} with $A = H^\alpha_{\Xi}$ and $\lambda_j, v_j$ given by the numerically computed $\tilde{\lambda}_j$ from part 2. and the $\oldhat{v}_j$ coming from Theorem \ref{th:second_1}, in order to derive rigorous enclosure intervals for every eigenvalue of $H^\alpha_{\Xi}$. We find that (using the script \texttt{compute\_PHalphaP\_enclosures.py} in the Github repo) the suprema over $\alpha \in \mathcal{G}$ of $\sup_{1 \leq j \leq m} \| \tilde{v}_j \|_\infty$, $\sup_{1 \leq j \leq m} \| \tilde{r}_{j,comp} \|_\infty$,$\| \tilde{H}^\alpha_{\Xi} \|_{max}$, $\sup_{1 \leq j \leq m} |\tilde{\lambda}_j|$, are bounded by $1$, $5 \times 10^{-14}$, $7$, and $8$, respectively. It is then easy to see that $2 \times 81$ times the right-hand side of \eqref{eq:bound_1} is much smaller than $0.01$, and is hence smaller than the distance between the minimum over $\alpha \in \mathcal{G}$ of the numerically computed first positive eigenvalues of $\tilde{H}^\alpha_{\Xi}$ and $\frac{8}{10}$. We can therefore conclude that the first positive eigenvalues of $H^\alpha_{\Xi}$ are bounded below by $\frac{8}{10}$ at every $\alpha \in \mathcal{G}$.


The main tool for part 4. of the strategy is the following.
\begin{theorem} \label{thm:Taylor_for_eigenvalues_1}
Let $A^\alpha$ be an $n \times n$ Hermitian matrix depending real-analytically on a real parameter $\alpha$. Denote the ordered eigenvalues of $A^\alpha$ by $\lambda_j^\alpha$ for $1 \leq j \leq n$. Then for any $\alpha$ and $\alpha_0$,
\begin{equation}
    | \lambda_j^\alpha - \lambda_j^{\alpha_0} | \leq |\alpha - \alpha_0| \sup_{\beta \in [\alpha_0,\alpha]} \| \de_\beta A^\beta \|_2 \text{ for all $1 \leq j \leq n$.} 
\end{equation}
\end{theorem}
\begin{proof}
See Appendix \ref{sec:proof_Taylor_eigenvalues}.
\end{proof}
We would like to apply Theorem \ref{thm:Taylor_for_eigenvalues_1} to bound the variation of eigenvalues of $H^\alpha_{\Xi}$. To this end we require the following proposition, which bounds the derivative of $H^\alpha_{\Xi}$ with respect to $\alpha$ over the interval $0 \leq \alpha \leq \frac{7}{10}$.
\begin{proposition} \label{prop:bound_de_alpha_H}
\begin{equation}
    \sup_{0 \leq \alpha \leq \frac{7}{10}} \left\| \de_\alpha H^\alpha_{\Xi} \right\|_2 \leq 38883 
\end{equation}
\end{proposition}
\begin{proof}
See Appendix \ref{sec:proof_de_alpha_H}.
\end{proof}
Proposition \ref{prop:bound_de_alpha_H} combined with Theorem \ref{thm:Taylor_for_eigenvalues_1} explains the choice of distance $h = \frac{1}{388831}$ between grid points. Assuming that the first positive eigenvalue of $H^\alpha_{\Xi}$ is bounded below by $\frac{8}{10}$ at grid points between $0$ and $\frac{7}{10}$ separated by $h$, we see that as long as
\begin{equation}
    \frac{38883 h}{2} < \frac{8}{10} - \frac{3}{4} \iff h < \frac{1}{388830},
\end{equation}
Proposition \ref{prop:bound_de_alpha_H} and Theorem \ref{thm:Taylor_for_eigenvalues_1} guarantee that the first eigenvalue of $H^\alpha_{\Xi}$ must be greater than $\frac{3}{4}$ over the whole interval $0 \leq \alpha \leq \frac{7}{10}$.
\end{proof}
\begin{remark}
In the proof of Proposition \ref{as:PHP_gap} we bound the round-off error which can occur in our numerical computations in order to draw rigorous conclusions. A common approach to this is interval arithmetic, see Rump\cite{Rump2010} and references therein. Our approach applies directly to the present problem and is just as rigorous.
\end{remark}
The results of a  computation of the eigenvalues of $H^\alpha_{\Xi}$ are shown in Figure \ref{fig:show_PHP_gap}.

\begin{figure}
\includegraphics[scale=.4]{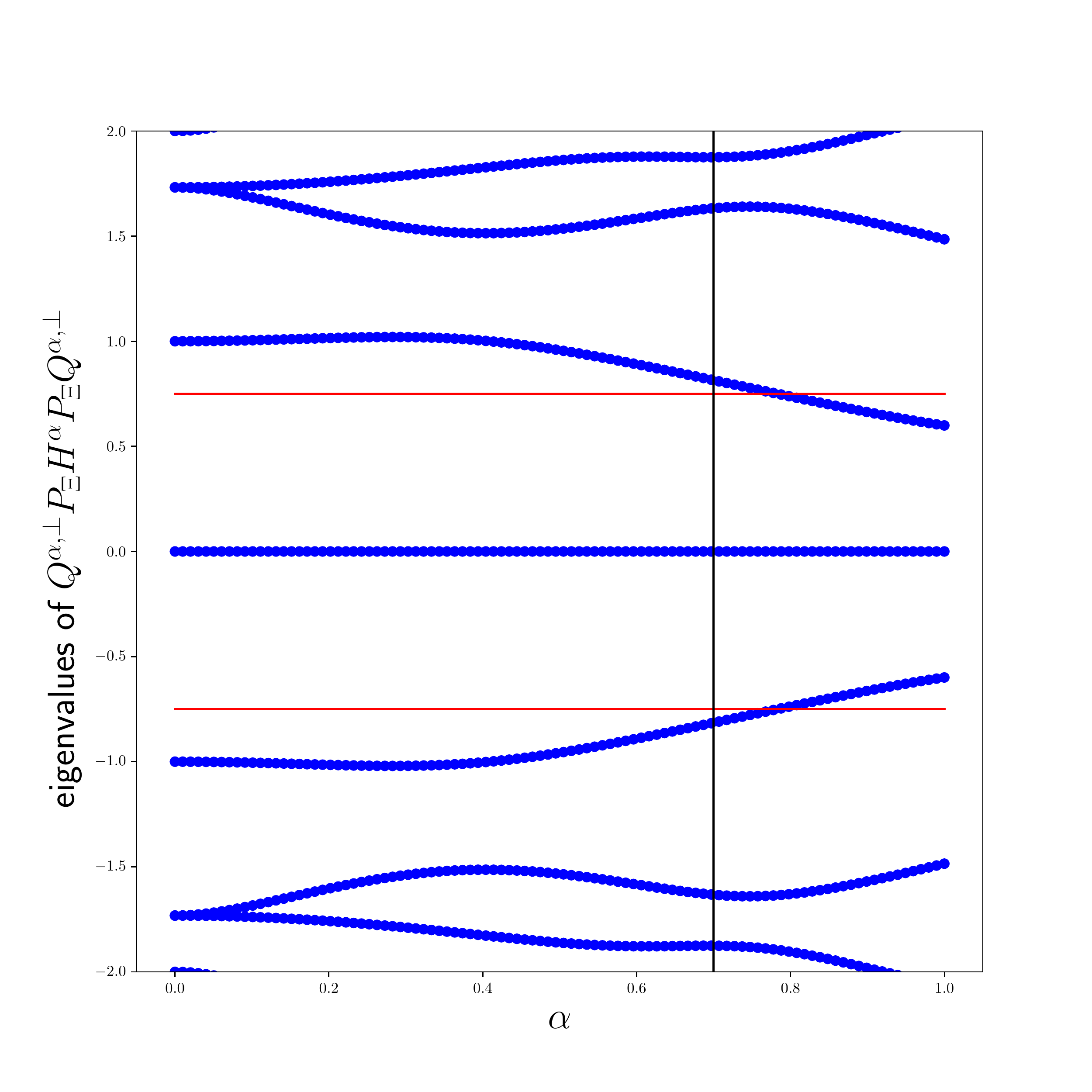}
\caption{
Plot of numerically computed eigenvalues of the 81$\times$81 matrix $H^\alpha_{\Xi}$ acting on $P_\Xi L^2_{K,1}$ (blue lines), showing the first non-zero eigenvalues are bounded away from $0$ by $\frac{3}{4}$ (red lines) when $\alpha$ is less than $\frac{7}{10}$ (black line). The zero eigenvalue corresponds to the subspace spanned by $\psi^{8,\alpha}$, and the non-zero eigenvalues equal those of the 80$\times$80 matrix $Q^{\alpha,\perp} P_\Xi H^\alpha P_\Xi Q^{\alpha,\perp}$ acting on $Q^{\alpha,\perp} P_\Xi L^2_{K,1}$ since non-zero eigenvectors $v$ of $Q^{\alpha,\perp} P_\Xi H^\alpha P_\Xi Q^{\alpha,\perp}$ must be orthogonal to $\psi^{8,\alpha}$ by orthogonality of eigenvectors corresponding to different eigenvalues. 
}
\label{fig:show_PHP_gap}
\end{figure}

Assuming Proposition \ref{prop:mu_choice} and Proposition \ref{as:PHP_gap}, the bound \eqref{eq:bound_below} becomes, for all $0 \leq \alpha \leq \frac{7}{10}$,
\begin{equation}
    \| Q^{\alpha,\perp} H^\alpha Q^{\alpha,\perp} \eta^\alpha \| \geq \left( \frac{3}{4} - \alpha \right) \| Q^{\alpha,\perp} \eta^\alpha \|.
\end{equation}
We now assume the following, proved in Appendix \ref{sec:TKV_expansion}.
\begin{proposition} \label{prop:bound_H1psi8}
$\| H^1 \Psi^8 \| \leq \frac{3}{20}$.
\end{proposition}
We can now give the proof of Theorem \ref{th:error_theorem_2}.
\begin{proof}[Proof of Theorem \ref{th:error_theorem_2}]
The proof follows immediately from Lemma \ref{lem:decompose}, Proposition \ref{prop:mu_choice}, Assumption \ref{as:PHP_gap}, and Proposition \ref{prop:bound_H1psi8}.
\end{proof}

\section*{Supplementary material}

In the supplementary material we list the chiral basis functions which span the space $\Xi$, list terms $\Psi^5-\Psi^8$ in the formal expansion of $\psi^\alpha$, and derive the TKV Hamiltonian from the Bistritzer-MacDonald model.

\appendix


\section{Derivation of expression for Fermi velocity in terms of $L^2_{K,1,1}$ zero mode of $H^\alpha$} \label{sec:K_Dirac}

The Bloch eigenvalue problem for the TKV Hamiltonian at quasi-momentum $\vec{k}$ is
\begin{equation}
    H^\alpha \psi^\alpha_{\vec{k}} = E_{\vec{k}} \psi^\alpha_{\vec{k}}
\end{equation}
where $H^\alpha$ is as in \eqref{eq:chiral_H} and 
\begin{equation}
    \psi^\alpha_{\vec{k}}(\vec{r} + \vec{v}) = e^{i \vec{k} \cdot \vec{v}} \diag(1,e^{i \vec{q}_1 \cdot \vec{v}},1,e^{i \vec{q}_1 \cdot \vec{v}}) \psi^\alpha_{\vec{k}}(\vec{r}) \quad \forall \vec{v} \in \Lambda.
\end{equation}
By Propositions \ref{prop:analytic_zeromodes} and \ref{prop:K_symmetry}, $0$ is a two-fold (at least) degenerate eigenvalue at the moir\'e $K$ point $\vec{k} = 0$, with associated eigenfunctions $\psi_{\pm 1}^\alpha$ as in Proposition \ref{prop:K_symmetry}. In what follows we assume that $0$ is \emph{exactly} two-fold degenerate so that $\psi_{\pm 1}^\alpha$ form a basis of the degenerate eigenspace. This assumption is clearly true for small $\alpha$ but could in principle be violated for $\alpha > 0$.

Introducing $\chi^\alpha_{\vec{k}} := e^{- i \vec{k} \cdot \vec{r}} \psi^\alpha_{\vec{k}}$, we derive the equivalent Bloch eigenvalue problem with $\vec{k}$-independent boundary conditions 
\begin{equation} \label{eq:periodic_Bloch_prob}
    H^\alpha_{\vec{k}} \chi^\alpha_{\vec{k}} = E_{\vec{k}} \chi^\alpha_{\vec{k}},
\end{equation}
where
\begin{equation}
    H^\alpha_{\vec{k}} := \begin{pmatrix} 0 & D_{\vec{k}}^{\alpha\dagger} \\ D_{\vec{k}}^\alpha & 0 \end{pmatrix}, \quad D_{\vec{k}}^\alpha = \begin{pmatrix} D_{x} + k_x + i( D_y + k_y ) & \alpha {U}(\vec{r}) \\ \alpha {U}(-\vec{r}) & D_x + k_x + i( D_y + k_y ) \end{pmatrix},
\end{equation}
where $D_{x,y} := - i \de_{x,y}$, 
and
\begin{equation}
    \chi^\alpha_{\vec{k}}(\vec{r} + \vec{v}) = \diag(1,e^{i \vec{q}_1 \cdot \vec{v}},1,e^{i \vec{q}_1 \cdot \vec{v}}) \chi^\alpha_{\vec{k}}(\vec{r}) \quad \forall \vec{v} \in \Lambda.
\end{equation}
Clearly $\psi_{\pm 1}^\alpha$ remain a basis of the zero eigenspace for the problem \eqref{eq:periodic_Bloch_prob} at $\vec{k} = 0$.

Differentiating the operator $D^\alpha_{\vec{k}}$ we find $\de_{k_x} D^\alpha_{\vec{k}} = I_2$ and $\de_{k_y} D^\alpha_{\vec{k}} = i I_2$, where $I_2$ denotes the $2 \times 2$ identity matrix, so that
\begin{equation} \label{eq:diff_H}
    \de_{k_x} H^\alpha_{\vec{k}} = \begin{pmatrix} 0 & I_2 \\ I_2 & 0 \end{pmatrix}, \quad \de_{k_y} H^\alpha_{\vec{k}} = \begin{pmatrix} 0 & - i I_2 \\ i I_2 & 0 \end{pmatrix}.
\end{equation}
By degenerate perturbation theory \cite{messiah1962quantum}, for small $\vec{k}$ we have that eigenfunctions $\chi^\alpha_{\vec{k}}$ of \eqref{eq:periodic_Bloch_prob} are given by 
\begin{equation}
    \chi^\alpha_{\vec{k}} \approx \sum_{\sigma = \pm 1} c_{\sigma,\vec{k}} \psi_\sigma^\alpha,
\end{equation}
where the coefficients $c_{\sigma,\vec{k}}$ and associated eigenvalues $E_{\vec{k}} \approx \epsilon_{\vec{k}}$ are found by solving the matrix eigenvalue problem
\begin{equation} \label{eq:k_dot_p}
    \begin{pmatrix} \frac{ \ip{ \psi_1^\alpha }{ \vec{k} \cdot  \nabla_{\vec{k}} H^\alpha_{\vec{0}} \psi_1^\alpha }}{ \ip{ \psi_1^\alpha }{ \psi_1^\alpha} } & \frac{ \ip{ \psi_1^\alpha }{ \vec{k} \cdot  \nabla_{\vec{k}} H^\alpha_{\vec{0}} \psi_{-1}^\alpha } }{ \ip{ \psi_1^\alpha }{ \psi_1^\alpha } } \\ \frac{ \ip{ \psi_{-1}^\alpha }{ \vec{k} \cdot  \nabla_{\vec{k}} H^\alpha_{\vec{0}} \psi_1^\alpha } }{ \ip{ \psi_{-1}^\alpha }{ \psi_{-1}^\alpha } } & \frac{ \ip{ \psi_{-1}^\alpha }{ \vec{k} \cdot \nabla_{\vec{k}} H^\alpha_{\vec{0}} \psi_{-1}^\alpha } }{ \ip{ \psi_{-1}^\alpha }{ \psi_{-1}^\alpha } } \end{pmatrix} \begin{pmatrix} c_{+1,\vec{k}} \\ c_{-1,\vec{k}} \end{pmatrix} = \epsilon_{\vec{k}} \begin{pmatrix} c_{+1,\vec{k}} \\ c_{-1,\vec{k}} \end{pmatrix}.
\end{equation}
Using \eqref{eq:diff_H} and the explicit forms of $\psi_{\pm 1}^\alpha$ given by Proposition \ref{prop:K_symmetry}, we find that the matrix on the left-hand side of \eqref{eq:k_dot_p} can be simplified to 
\begin{equation}
    \begin{pmatrix} 0 & \lambda(\alpha) ( k_x - i k_y ) \\ \lambda^*(\alpha) ( k_x + i k_y ) & 0 \end{pmatrix}, \quad \lambda(\alpha) := \frac{ \ip{ \psi^\alpha_1(\vec{r}) }{ \psi^{\alpha *}_1(-\vec{r}) } }{ \ip{ \psi_1^\alpha }{ \psi_1^\alpha } }.
\end{equation}
It follows that, for small $\vec{k}$, we have $E_{\vec{k}} \approx \pm v(\alpha) |\vec{k}|$, where $v(\alpha) = | \lambda(\alpha) |$ is as in \eqref{eq:Fermi_v_2}.

\section{The chiral basis of $L^2_{K,1}$ and action of $H^0$ and $H^1$ with respect to this basis} \label{sec:chiral}

\subsection{The spectrum and eigenfunctions of $H^0$ in $L^2_{K}$} \label{sec:H0_eigenfuncs}

The first task is to understand the spectrum and eigenfunctions of $H^0$ in $L^2_{K}$. In the next section we will discuss the spectrum and eigenfunctions of $H^0$ in $L^2_{K,1}$. Recall that
\begin{equation}
    H^0 = \begin{pmatrix} 0 & D^{0 \dagger} \\ D^{0} & 0 \end{pmatrix}, \quad D^0 = \begin{pmatrix} - 2 i \overline{\de} & 0 \\ 0 & - 2 i \overline{\de} \end{pmatrix},
\end{equation}
where $\overline{\de} = \frac{1}{2} ( \de_x + i \de_y )$. To describe the eigenfunctions of $H^0$ in $L^2_K$ we introduce some notation. Let $\vec{v} = \begin{pmatrix} v_1 , v_2 \end{pmatrix}$ be a vector in $\field{R}^2$. Then we will write
\begin{equation}
    z_{\vec{v}} = v_1 + i v_2, \quad \oldhat{z}_{\vec{v}} = \frac{ v_1 + i v_2 }{ |\vec{v}| }.
\end{equation}
Finally, let $V$ denote the area of the moir\'e cell $\Omega$.
\begin{proposition} \label{prop:L2K}
The zero eigenspace of $H^0$ in $L^2_K$ is spanned by
\begin{equation}
    \chi_\pm^{\vec{0}} = \frac{1}{\sqrt{2 V}} \begin{pmatrix} 1 , 0 , \pm 1 , 0 \end{pmatrix}.
\end{equation}
For all $\vec{G} \neq 0$ in the reciprocal lattice, then
\begin{equation} \label{eq:chi_pm}
    \chi_\pm^{\vec{G}}(\vec{r}) = \frac{1}{\sqrt{2 V}} \begin{pmatrix} 1 , 0 , \pm \oldhat{z}_{\vec{G}} , 0 \end{pmatrix} e^{i \vec{G} \cdot \vec{r}}
\end{equation}
are eigenfunctions with eigenvalues $\pm |\vec{G}|$. For all $\vec{G}$ in the reciprocal lattice,
\begin{equation} \label{eq:chi_pm_2}
    \chi_\pm^{\vec{q}_1 + \vec{G}}(\vec{r}) = \frac{1}{\sqrt{2 V}} \begin{pmatrix} 0 , 1 , 0 , \pm \oldhat{z}_{\vec{G} + \vec{q}_1} \end{pmatrix} e^{i (\vec{q}_1 + \vec{G}) \cdot \vec{r}}
\end{equation}
are eigenfunctions with eigenvalues $\pm |\vec{q}_1 + \vec{G}|$. The operator $H^0$ has no other eigenfunctions in $L^2_K$ other than linear combinations of these, and hence the spectrum of $H^0$ in $L^2_K$ is
\begin{equation}
    \sigma_{L^2_K}(H^0) = \left\{ \pm |\vec{G}|, \pm |\vec{q}_1 + \vec{G}| : \vec{G} \in \Lambda^* \right\}.
\end{equation}
\end{proposition}
\begin{proof}
The proof is a straightforward calculation taking into account the $L^2_K$ boundary conditions given by \eqref{eq:L2k_space} with $\vec{k} = 0$. For example, ${e}_2$ and ${e}_4$ are zero eigenfunctions of $H^0$ but in $L^2_{K'}$, not $L^2_K$.
\end{proof}

Note that (as it must be because of the chiral symmetry) the spectrum is symmetric about $0$ and all of the eigenfunctions with negative eigenvalues are given by applying $\mathcal{S}$ to the eigenfunctions with positive eigenvalues. 

The union of the lattices $\Lambda^*$ and $\Lambda^* + \vec{q}_1$ has the form of a honeycomb lattice in momentum space, where the lattice $\Lambda^*$ corresponds to ``A'' sites and $\Lambda^* + \vec{q}_1$ corresponds to ``B'' sites (or vice versa), see Figure \ref{fig:H0_eigenvalues}.

\begin{figure}
\centering
\includegraphics[scale=.5]{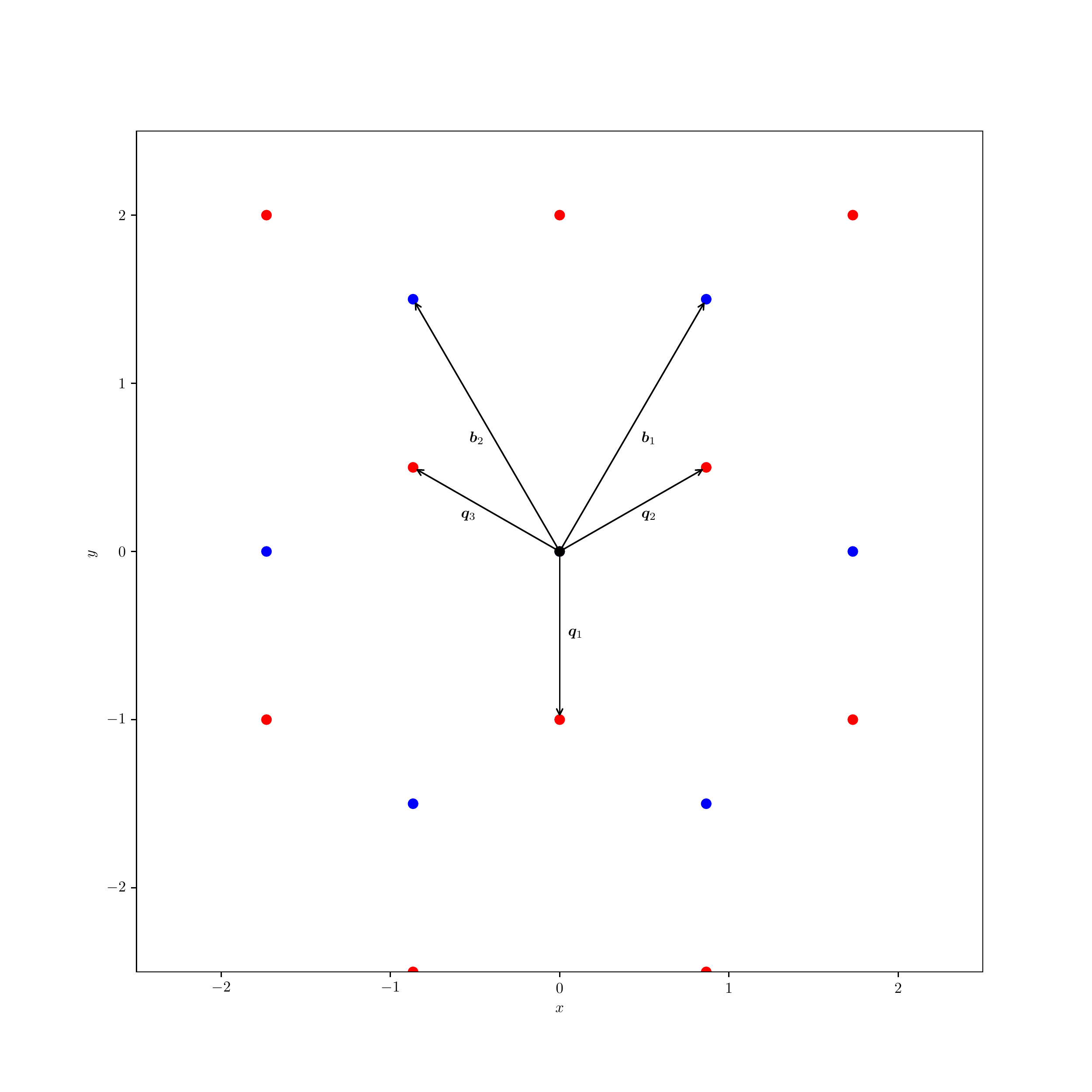}
\caption{Diagram showing $A$ (blue) and $B$ (red) sites of the momentum space lattice. Each site corresponds to two $L^2_K$-eigenvalues of $H^0$, given by $\pm$ the distance between the site and the origin (black). The lattice vectors $\vec{b}_1$ and $\vec{b}_2$ are shown, as well as the $A$ site nearest-neighbor vectors $\vec{q}_1$, $\vec{q}_2$, $\vec{q}_3$.}
\label{fig:H0_eigenvalues}
\end{figure}

\begin{figure}
\centering
\includegraphics[scale=.5]{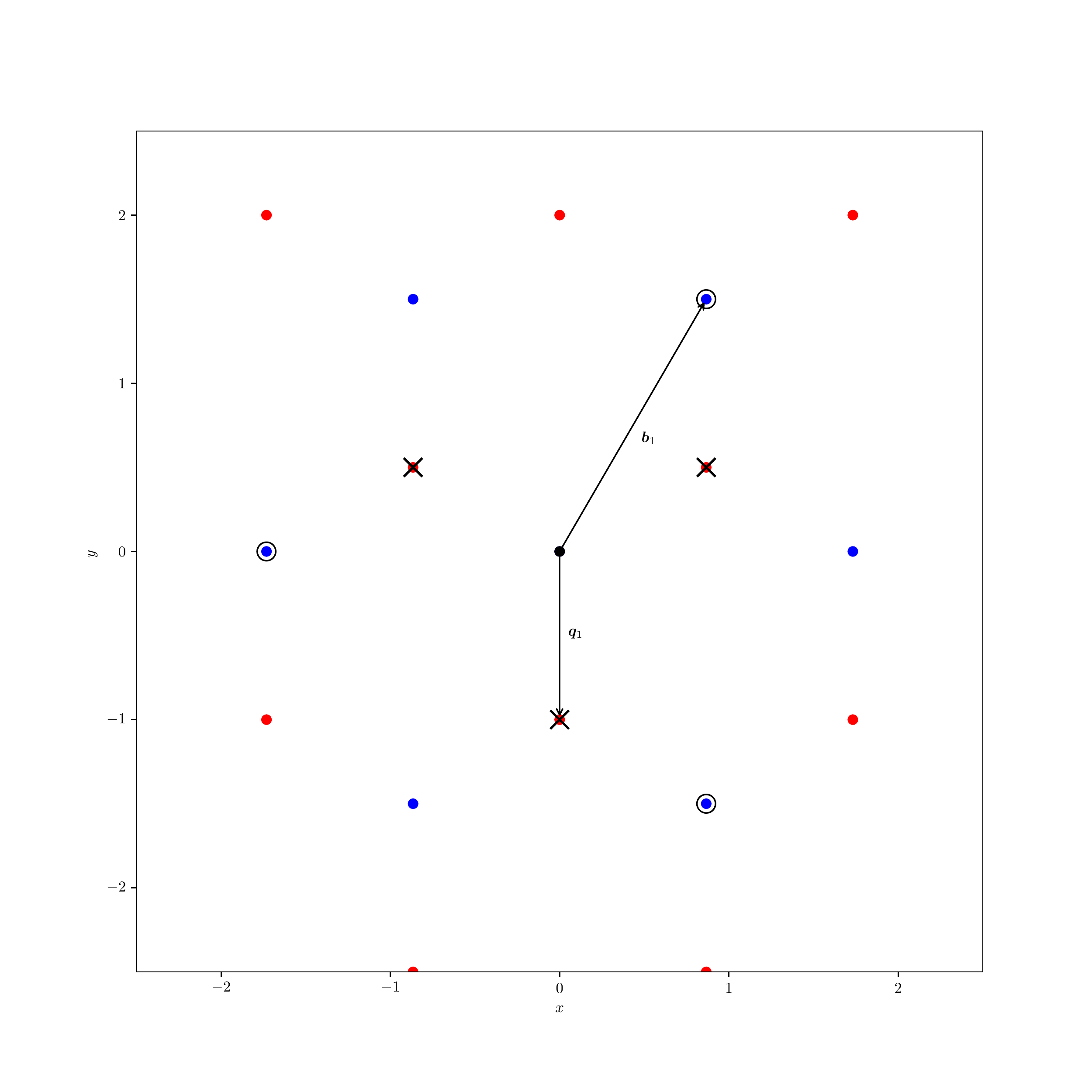}
\caption{Diagram showing support of $L^2_{K,1}$-eigenfunctions of $H^0$ superposed on the momentum space lattice. Each eigenfunction is given by superposing an $L^2_{K}$-eigenfunction of $H^0$ with its rotations by $\frac{2 \pi}{3}$ and $\frac{4 \pi}{3}$. The support of the eigenfunctions $\chi^{\pm \widetilde{\vec{q}_1}}$ with eigenvalues $\pm 1$ is shown with black crosses, while the support of the eigenfunctions $\chi^{\pm \widetilde{\vec{b}_1}}$ with eigenvalues $\pm \sqrt{3}$ is shown with black circles.
}
\label{fig:H0_L2K1_eigenvalues}
\end{figure}


\subsection{The spectrum and eigenfunctions of $H^0$ in $L^2_{K,1}$} \label{sec:H0_eigenfuncs_rotation}

We now discuss the spectrum of $H^0$ in $L^2_{K,1}$. 
\begin{proposition} \label{prop:L2K1}
The zero eigenspace of $H^0$ in $L^2_{K,1}$ is spanned by
\begin{equation}
    \chi^{\widetilde{\vec{0}}} := \frac{1}{\sqrt{V}} e_1.
\end{equation}
For all $\vec{G} \neq 0$ in the reciprocal lattice $\Lambda^*$,
\begin{equation}
    \chi^{\widetilde{\vec{G}}}_\pm := \frac{1}{\sqrt{3}} \sum_{k = 0}^2 \mathcal{R}^k \chi_\pm^{\vec{G}} = \frac{1}{\sqrt{3}}  \sum_{k = 0}^2 \chi_\pm^{ (R_\phi^*)^k \vec{G}}
\end{equation}
are eigenfunctions of $H^0$ in $L^2_{K,1}$ with associated eigenvalues $\pm |\vec{G}|$. For all $\vec{G}$ in the reciprocal lattice $\Lambda^*$,
\begin{equation}
    \chi^{\pm \widetilde{\vec{G} + \vec{q}_1}} = \frac{1}{\sqrt{3}} \sum_{k = 0}^2 \mathcal{R}^k \chi_\pm^{\vec{G} + \vec{q}_1} = \frac{1}{\sqrt{3}}  \sum_{k = 0}^2 \chi_\pm^{ (R_\phi^*)^k (\vec{G} + \vec{q}_1) }
\end{equation}
are eigenfunctions of $H^0$ in $L^2_{K,1}$ with associated eigenvalues $\pm |\vec{q}_1 + \vec{G}|$. The operator $H^0$ has no other eigenfunctions in $L^2_{K,1}$ other than linear combinations of these, and hence the spectrum of $H^0$ in $L^2_{K,1}$ is
\begin{equation}
    \sigma_{L^2_{K,1}}(H^0) = \left\{ \pm |\vec{G}|, \pm |\vec{q}_1 + \vec{G}| : \vec{G} \in \Lambda^* \right\}.
\end{equation}
\end{proposition}
\begin{proof}
The proof is another straightforward calculation starting from Proposition \ref{prop:L2K}. 
\end{proof}
For an illustration of the support of the $L^2_{K,1}$-eigenfunctions of $H^0$ on the momentum space lattice, see Figure \ref{fig:H0_L2K1_eigenvalues}. It is important to note that the notation introduced in Proposition \ref{prop:L2K1} is not one-to-one, because for example
\begin{equation}
    \chi^{\pm \widetilde{\vec{G}}} = \chi^{\pm \widetilde{R_\phi^* \vec{G}}} = \chi^{\pm \widetilde{ (R_\phi^*)^2 \vec{G} }}
\end{equation}
for any $\vec{G} \neq 0$ in $\Lambda^*$. 

\subsection{The chiral basis of $L^2_{K,1}$}

Recall that zero modes of $H^\alpha$ can be assumed to be eigenfunctions of the chiral symmetry operator $\mathcal{S}$. It follows that the most convenient basis for our purposes is not be the spectral basis just introduced but the basis of $L^2_{K,1}$ consisting of eigenfunctions of $\mathcal{S}$. We call this basis the chiral basis.
\begin{definition}
The chiral basis of $L^2_{K,1}$ is defined as the union of the functions
\begin{equation} \label{eq:chiral_zero}
    \chi^{\widetilde{\vec{0}}} = \frac{1}{\sqrt{V}} {e}_1,
\end{equation}
\begin{equation} \label{eq:chiral_G}
    \chi^{\widetilde{\vec{G}},\pm 1} := \frac{1}{\sqrt{2}} \left( \chi^{\widetilde{\vec{G}}} \pm \chi^{-\widetilde{\vec{G}}} \right), \quad \vec{G} \in \Lambda^* \setminus \{ \vec{0} \},
\end{equation}
and 
\begin{equation} \label{eq:chiral_q_G}
    \chi^{\widetilde{\vec{q}_1 + \vec{G}},\pm 1} := \frac{1}{\sqrt{2}} \left( \chi^{\widetilde{\vec{q}_1 + \vec{G}}} \pm \chi^{-\widetilde{\vec{q}_1 + \vec{G}}} \right), \quad \vec{G} \in \Lambda^*.
\end{equation}
\end{definition}
The following is straightforward.
\begin{proposition} \label{prop:L2K11}
The chiral basis is an orthonormal basis of $L^2_{K,1}$. The modes $\chi^{\widetilde{\vec{0}}}$, $\chi^{\widetilde{\vec{G}},1}$, and $\chi^{\widetilde{\vec{q}_1 + \vec{G}},1}$ are $+1$ eigenfunctions of $\mathcal{S}$, while the modes $\chi^{\widetilde{\vec{G}},-1}$ and $\chi^{\widetilde{\vec{q}_1 + \vec{G}},-1}$ are $-1$ eigenfunctions of $\mathcal{S}$.
\end{proposition}
Written out, chiral basis functions have a very simple form. We have
\begin{equation} \label{eq:A0_chiral_efuncs}
    \chi^{\widetilde{\vec{0}}} = \frac{1}{\sqrt{V}} {e}_1,
\end{equation}
and for all $\vec{G} \in \Lambda^* \setminus \{ \vec{0} \}$,
\begin{equation} \label{eq:A_chiral_efuncs}
    \chi^{\widetilde{\vec{G}},1}(\vec{r}) = \frac{1}{\sqrt{3 V}} {e}_1 \sum_{k = 0}^2 e^{i ( (R_\phi^*)^k \vec{G} ) \cdot \vec{r}}, \quad \chi^{\widetilde{\vec{G}},-1}(\vec{r}) = \frac{1}{\sqrt{3 V}} \oldhat{z}_{\vec{G}} {e}_3 \sum_{k = 0}^2 e^{- i k \phi} e^{i ( (R_\phi^*)^k \vec{G} ) \cdot \vec{r}},
\end{equation}
and for all $\vec{G} \in \Lambda^*$,
\begin{equation} \label{eq:B_chiral_efuncs}
\begin{split}
    \chi^{\widetilde{\vec{G}+\vec{q}_1},1}(\vec{r}) &= \frac{1}{\sqrt{3 V}} {e}_2 \sum_{k = 0}^2 e^{i ( (R_\phi^*)^k (\vec{q}_1 + \vec{G}) ) \cdot \vec{r}}, \\
    \chi^{\widetilde{\vec{G} + \vec{q}_1},-1}(\vec{r}) &= \frac{1}{\sqrt{3 V}} \oldhat{z}_{\vec{G} + \vec{q}_1} {e}_4 \sum_{k = 0}^2 e^{- i k \phi} e^{i ( (R_\phi^*)^k (\vec{q}_1 + \vec{G}) ) \cdot \vec{r}}.
\end{split}
\end{equation}
We use the chiral basis to divide up $L^2_{K,1}$ as follows. 
\begin{definition}
We define spaces $L^2_{K,1,\pm 1}$ to be the spans of the $\pm 1$ eigenfunctions of $\mathcal{S}$ in $L^2_{K,1}$, respectively. 
\end{definition}
Clearly we have
\begin{equation}
    L^2_{K,1} = L^2_{K,1,1} \oplus L^2_{K,1,-1}.
\end{equation}
We can divide up the chiral basis more finely as follows.
\begin{definition} \label{def:A_B_momentum}
We define
\begin{equation}
\begin{split}
    L^2_{K,1,1,A} &:= \left\{ \chi^{\widetilde{\vec{0}}} \right\} \cup \left\{ \chi^{\widetilde{\vec{G}},1} : \vec{G} \in \Lambda^* \setminus \{ \vec{0} \} \right\} , \\
    L^2_{K,1,1,B} &:= \left\{ \chi^{\widetilde{\vec{G} + \vec{q}_1},1} : \vec{G} \in \Lambda^* \right\} , \\
    L^2_{K,1,-1,A} &:= \left\{ \chi^{\widetilde{\vec{G}},-1} : \vec{G} \in \Lambda^* \setminus \{ \vec{0} \} \right\} , \\
    L^2_{K,1,-1,B} &:= \left\{ \chi^{\widetilde{\vec{G} + \vec{q}_1},-1} : \vec{G} \in \Lambda^* \right\}.
\end{split}
\end{equation}
\end{definition}
\begin{remark}
Note that the notation $A$ and $B$ in Definition \ref{def:A_B_momentum} refers to $A$ and $B$ sites of the momentum space lattice, not to the $A$ and $B$ sites of the real space lattice. Recalling Remark \ref{rem:A_B_sites} and comparing \eqref{eq:A_chiral_efuncs}-\eqref{eq:B_chiral_efuncs} with \eqref{eq:psi_densities}, we see that $L^2_{K,1,1,A}$ corresponds to wave-functions supported on $A$ sites of layer $1$, $L^2_{K,1,1,B}$ corresponds to wave-functions supported on $A$ sites of layer $2$, $L^2_{K,1,-1,A}$ corresponds to wave-functions supported on $B$ sites of layer $1$, and $L^2_{K,1,-1,B}$ corresponds to wave-functions supported on $B$ sites of layer $2$. 
\end{remark}
Clearly we have
\begin{equation}
    L^2_{K,1} = L^2_{K,1,1,A} \oplus L^2_{K,1,1,B} \oplus L^2_{K,1,-1,A} \oplus L^2_{K,1,-1,B}.
\end{equation}
The following propositions are straightforward to prove. For the first claim, note that $\{ \mathcal{S} , H^0 \}=0$.
\begin{proposition} \label{prop:H0_prop}
The operator $H^0$ maps $L^2_{K,1,\pm 1,\sigma} \rightarrow L^2_{K,1,\mp 1,\sigma}$ for $\sigma = A,B$. The action of $H^0$ on chiral basis functions is as follows
\begin{equation}
    H^0 \chi^{\widetilde{\vec{0}}} = 0,
\end{equation}
for all $\vec{G} \in \Lambda^*$ with $\vec{G} \neq 0$
\begin{equation}
    H^0 \chi^{\widetilde{\vec{G}},\pm 1} = |\vec{G}| \chi^{\widetilde{\vec{G}},\mp 1},
\end{equation}
and for all $\vec{G} \in \Lambda^*$
\begin{equation}
    H^0 \chi^{\widetilde{\vec{q}_1 + \vec{G}},\pm 1} = |\vec{q}_1 + \vec{G}| \chi^{\widetilde{\vec{q}_1 + \vec{G}}, \mp 1}.
\end{equation}
\end{proposition}
\begin{proposition} \label{prop:H0_inv_prop}
Let $P$ denote the projection operator onto $\chi^{\widetilde{\vec{0}}}$ in $L^2_{K,1}$, and $P^\perp = 1 - P$. Then the operator $P^\perp (H^0)^{-1} P^\perp$ maps $L^2_{K,1,\pm 1,\sigma} \rightarrow L^2_{K,1,\mp 1,\sigma}$ for $\sigma = A,B$, and
\begin{equation}
    P^\perp (H^0)^{-1} P^\perp \chi^{\widetilde{\vec{G}},\pm 1} = \frac{1}{|\vec{G}|} \chi^{\widetilde{\vec{G}},\mp 1}
\end{equation}
for all $\vec{G} \in \Lambda^*$ with $\vec{G} \neq 0$, and
\begin{equation}
    P^\perp (H^0)^{-1} P^\perp \chi^{\widetilde{\vec{q}_1 + \vec{G}},\pm 1} = \frac{1}{|\vec{q}_1 + \vec{G}|} \chi^{\widetilde{\vec{q}_1 + \vec{G}}, \mp 1}
\end{equation}
for all $\vec{G} \in \Lambda^*$.
\end{proposition}
In the coming sections we will study the action of the operator $H^1$ on $L^2_{K,1}$ with respect to the chiral basis.

\subsection{The spectrum of $H^1$ in $L^2_K$ and $L^2_{K,1}$}

Recall that
\begin{equation}
    H^1 = \begin{pmatrix} 0 & D^{1\dagger} \\ D^1 & 0 \end{pmatrix}, \quad D^1 = \begin{pmatrix} 0 & U(\vec{r}) \\ U(-\vec{r}) & 0 \end{pmatrix},
\end{equation}
where $U(\vec{r}) = e^{- i \vec{q}_1 \cdot \vec{r}} + e^{i \phi} e^{- i \vec{q}_2 \cdot \vec{r}} + e^{- i \phi} e^{- i \vec{q}_3 \cdot \vec{r}}$. We claim the following.
\begin{proposition} \label{prop:H1_action_L2K}
For each $\vec{r}_0 \in \Omega$, $\pm |U(\vec{r}_0)|$ and $\pm |U(-\vec{r}_0)|$ are eigenvalues of $H^1 : L^2_K \rightarrow L^2_K$. For $\vec{r}_0$ such that $U(\vec{r}_0) \neq 0$, the $\pm |U(\vec{r}_0)|$ eigenvectors are
\begin{equation} \label{eq:H1_evec_1}
    \begin{pmatrix} 0, 1, \pm \frac{U(\vec{r}_0)}{|U(\vec{r}_0)|}, 0 \end{pmatrix} \delta(\vec{r} - \vec{r}_0).
\end{equation}
For $\vec{r}_0$ such that $U(-\vec{r}_0) \neq 0$, the $\pm |U(-\vec{r}_0)|$ eigenvectors are
\begin{equation} \label{eq:H1_evec_2}
    \begin{pmatrix} 1, 0, 0, \pm \frac{U(-\vec{r}_0)}{|U(-\vec{r}_0)|} \end{pmatrix} \delta(\vec{r} - \vec{r}_0).
\end{equation}
When $U(\vec{r}_0) = 0$, zero is a degenerate eigenvalue with associated eigenfunctions ${e}_2 \delta(\vec{r} - \vec{r}_0)$ and ${e}_3 \delta(\vec{r} - \vec{r}_0)$. When $U(-\vec{r}_0) = 0$, zero is a degenerate eigenvalue with associated eigenfunctions ${e}_1 \delta(\vec{r} - \vec{r}_0)$ and ${e}_4 \delta(\vec{r} - \vec{r}_0)$. Finally,
\begin{equation} \label{eq:spec_H1}
    \sigma_{L^2_K}(H^1) = [-3,3].
\end{equation}
\end{proposition}
\begin{proof}
We prove only \eqref{eq:spec_H1} since the other assertions are clear. The triangle inequality yields the obvious bound 
\begin{equation}
    |U(\vec{r}_0)| \leq 3,
\end{equation}
so that the $L^2_K$ spectrum of $H^1$ must be contained in the interval $[-3,3]$. To see that the spectrum actually equals $[-3,3]$, note that if $\vec{r}_0 = \begin{pmatrix} \frac{4 \pi}{3 \sqrt{3}}, 0 \end{pmatrix}$ then
\begin{equation}
    \vec{q}_1 \cdot \vec{r}_0 = 0, (\vec{q}_1 + \vec{b}_1) \cdot \vec{r}_0 = \frac{1}{2} \begin{pmatrix} \sqrt{3}, 1 \end{pmatrix} \cdot \vec{r}_0 = \frac{2 \pi}{3}, (\vec{q}_1 + \vec{b}_2) \cdot \vec{r}_0 = \frac{1}{2} \begin{pmatrix} - \sqrt{3}, 1 \end{pmatrix} \cdot \vec{r}_0 = - \frac{2 \pi}{3}
\end{equation}
and hence $U(\vec{r}_0) = 3$. On the other hand, when $\vec{r}_0 = 0$ we have $U(\vec{r}_0) = 0$ so that the spectrum of $H^1$ in $L^2_K$ equals $[-3,3]$.
\end{proof}
By taking linear combinations of rotated copies of the $H^1$ eigenfunctions, just as we did with the $H^0$ eigenfunctions, it is straightforward to prove an analogous result to Proposition \ref{prop:H1_action_L2K} in $L^2_{K,1}$. We record only the following.
\begin{proposition} \label{prop:H1_norm_prop}
\begin{equation}
    \sigma_{L^2_{K,1}}(H^1) = [-3,3].
\end{equation}
\end{proposition}

\subsection{The action of $H^1$ on $L^2_{K,1}$ with respect to the chiral basis}

We now want to study the action of $H^1$ on $L^2_{K,1}$ with respect to the chiral basis. We will prove two propositions, which parallel Proposition \ref{prop:H0_prop}.
\begin{proposition} \label{prop:H1_L2K1}
The operator $H^1$ maps $L^2_{K,1,1,A} \rightarrow L^2_{K,1,-1,B}$, and $L^2_{K,1,1,B} \rightarrow L^2_{K,1,-1,A}$. The action of $H^1$ on chiral basis functions is as follows:
\begin{equation} \label{eq:H1_0}
    H^1 \chi^{\widetilde{\vec{0}}} = \sqrt{3} \overline{ \oldhat{z}_{\vec{q}_1} } \chi^{\widetilde{\vec{q}_1},-1},
\end{equation}
and
\begin{equation} \label{eq:H1_q1}
\begin{split}
    H^1 \chi^{\widetilde{\vec{q}_1},1} = e^{i \phi} \overline{ \oldhat{z}_{\vec{q}_1 - \vec{q}_2} } \chi^{\widetilde{\vec{q}_1-\vec{q}_2},-1} + e^{- i \phi} \overline{ \oldhat{z}_{\vec{q}_1 - \vec{q}_3} } \chi^{\widetilde{\vec{q}_1 - \vec{q}_3},-1}  
\end{split}
\end{equation}
For all $\vec{G} \in \Lambda^* \setminus \{ \vec{0} \}$,
\begin{equation} \label{eq:H1_chi_G}
\begin{split}
    H^1 \chi^{\widetilde{\vec{G}},1} = \overline{\oldhat{z}_{\vec{G} + \vec{q}_1}} \chi^{\widetilde{\vec{G} + \vec{q}_1},-1} + e^{i \phi} \overline{\oldhat{z}_{\vec{G} + \vec{q}_2}} \chi^{\widetilde{\vec{G} + \vec{q}_2},-1} + e^{- i \phi} \overline{\oldhat{z}_{\vec{G} + \vec{q}_3}} \chi^{\widetilde{\vec{G} + \vec{q}_3},-1}  
\end{split}
\end{equation}
For all $\vec{G} \in \Lambda^* \setminus \{ \vec{0} \}$,
\begin{equation} \label{eq:H1_q1_G}
\begin{split}
    H^1 \chi^{\widetilde{\vec{G} + \vec{q}_1},1} = \overline{ \oldhat{z}_{\vec{G}} } \chi^{\widetilde{\vec{G}},-1} + e^{i \phi} \overline{ \oldhat{z}_{\vec{G} + \vec{q}_1 - \vec{q}_2} } \chi^{\widetilde{\vec{G} + \vec{q}_1 - \vec{q}_2},-1} + e^{- i \phi} \overline{ \oldhat{z}_{\vec{G} + \vec{q}_1 - \vec{q}_3} } \chi^{\widetilde{\vec{G} + \vec{q}_1 - \vec{q}_3},-1}.
\end{split}
\end{equation}
\end{proposition}
Note that $H^1$ exchanges chirality ($\mathcal{S}$ eigenvalue) and the $A$ and $B$ momentum space sublattices, while $H^0$ only exchanges chirality. Proposition \ref{prop:H1_L2K1} has a simple interpretation in terms of nearest-neighbor hopping in the momentum space lattice, see Figures \ref{fig:H1_chi_G} and \ref{fig:H1_0}.
\begin{remark}
At first glance, equations \eqref{eq:H1_0} and \eqref{eq:H1_q1} appear different from \eqref{eq:H1_chi_G} and \eqref{eq:H1_q1_G}, because they appear to violate $\frac{2 \pi}{3}$ rotation symmetry. But this is not the case, since every chiral basis function individually respects this symmetry. For example, using $\chi^{\widetilde{\vec{q}_1},-1} = \chi^{\widetilde{\vec{q}_2},-1} = \chi^{\widetilde{\vec{q}_3},-1}$ and $\overline{\oldhat{z}_{\vec{q}_1}} = e^{i \phi} \overline{\oldhat{z}_{\vec{q}_2}} = e^{- i \phi} \overline{\oldhat{z}_{\vec{q}_3}}$, we can re-write \eqref{eq:H1_0} in a way that manifestly respects the $\frac{2 \pi}{3}$ rotation symmetry as
\begin{equation} \label{eq:H1_0_2}
    H^1 \chi^{\widetilde{\vec{0}}} = \frac{1}{\sqrt{3}} \left( \overline{\oldhat{z}_{\vec{q}_1}} \chi^{\widetilde{\vec{q}_1},-1} + e^{i \phi} \overline{\oldhat{z}_{\vec{q}_2}} \chi^{\widetilde{\vec{q}_2},-1} + e^{- i \phi} \overline{\oldhat{z}_{\vec{q}_3}} \chi^{\widetilde{\vec{q}_3},-1} \right).
\end{equation}
Equation \eqref{eq:H1_q1} can also be written in a manifestly rotationally invariant way but the expression is long and hence we omit it. Note that \eqref{eq:H1_q1} cannot have a term proportional to $\chi^{\widetilde{\vec{0}}}$ since $\chi^{\widetilde{\vec{0}}} \in L^2_{K,1,1}$ and $H^1$ maps $L^2_{K,1,1} \rightarrow L^2_{K,1,-1}$.
\end{remark}
\begin{figure}
\centering
\includegraphics[scale=.3]{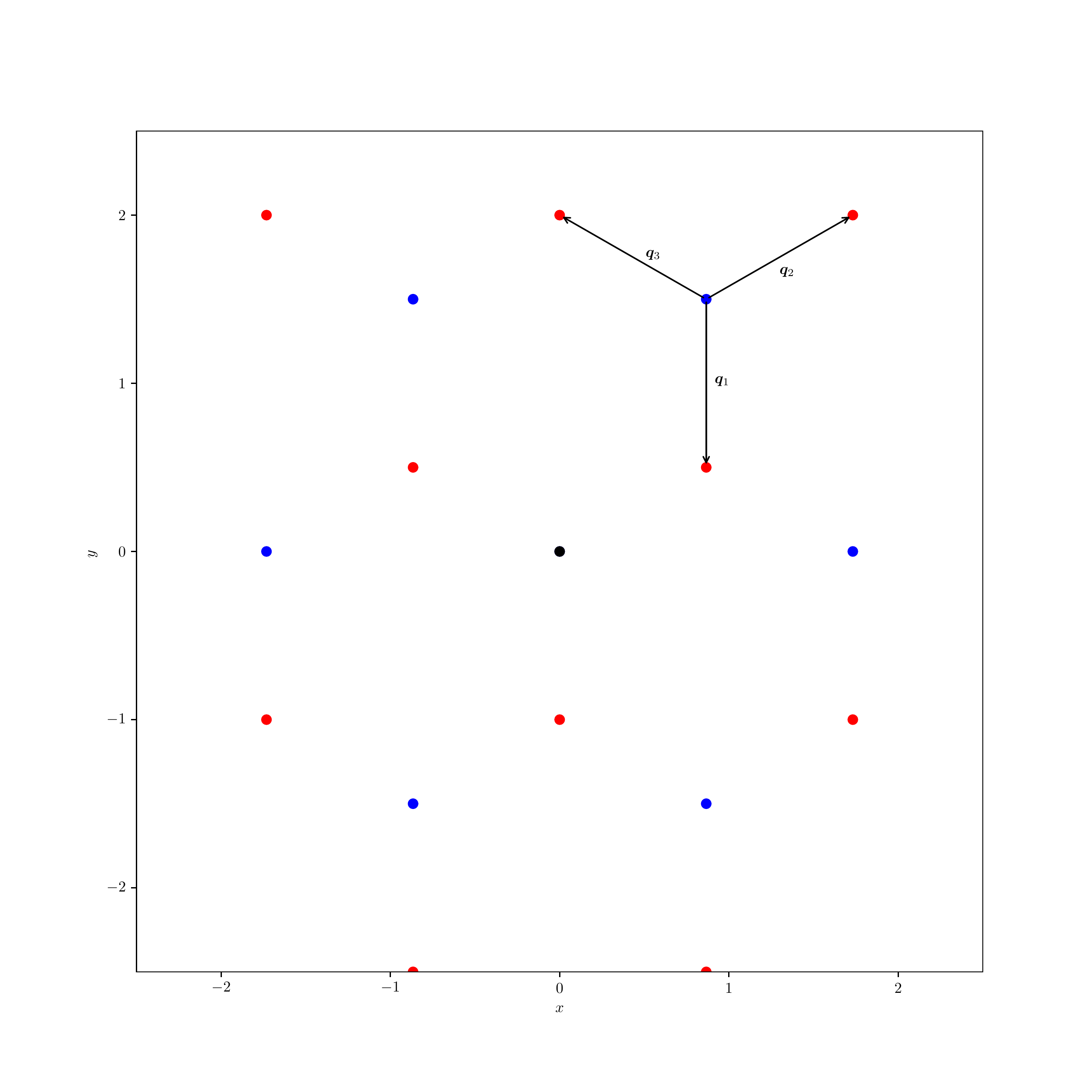}
\includegraphics[scale=.3]{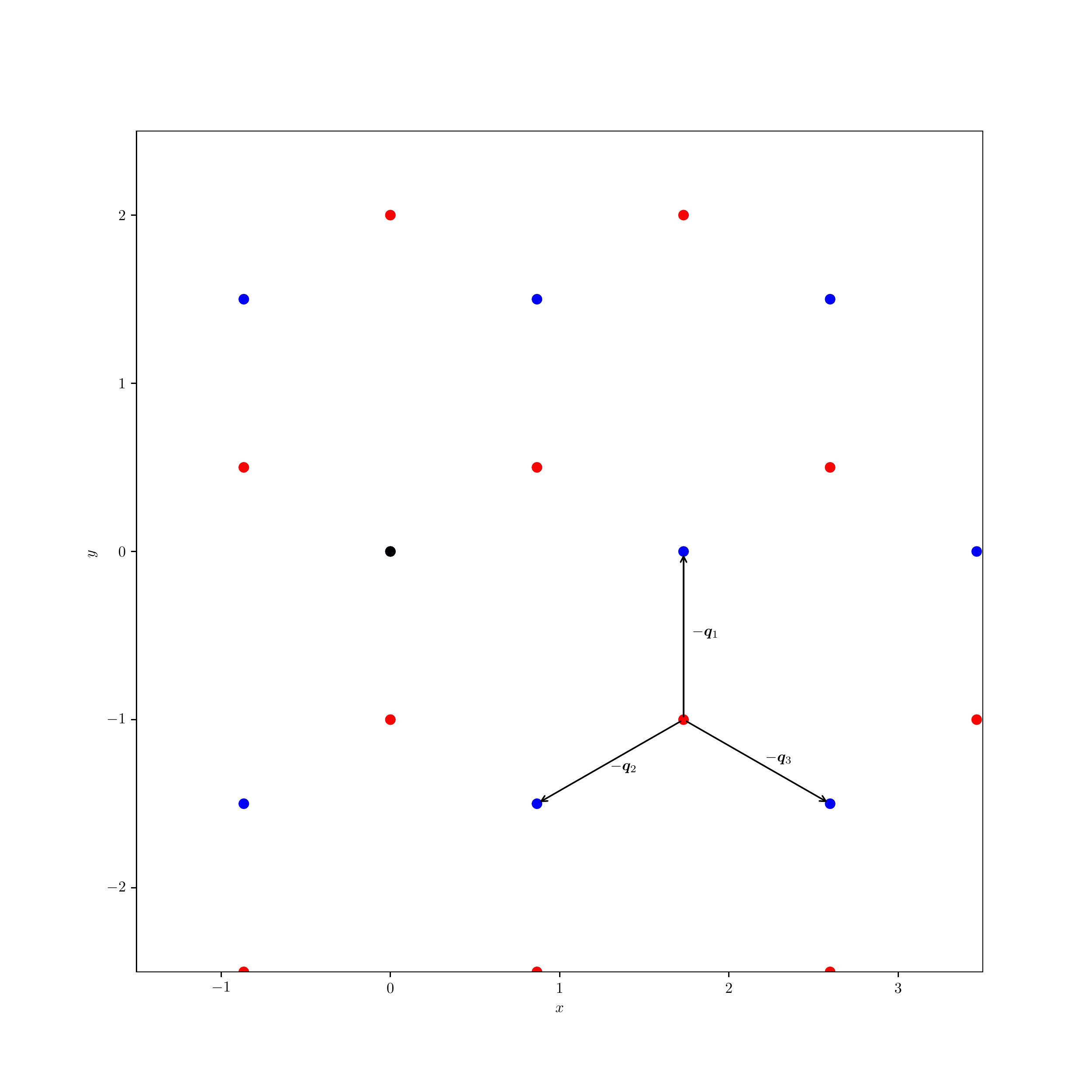}
\caption{Illustration of the action of $H^1$ in $L^2_{K,1}$ as hopping in the momentum space lattice described by equations \eqref{eq:H1_chi_G} (left, starting at $\vec{b}_1$) and \eqref{eq:H1_q1_G} (right, starting at $\vec{q}_1 + \vec{b}_1 - \vec{b}_2$). The origin is marked by a black dot. }
\label{fig:H1_chi_G}
\end{figure}
\begin{figure}
\centering
\includegraphics[scale=.3]{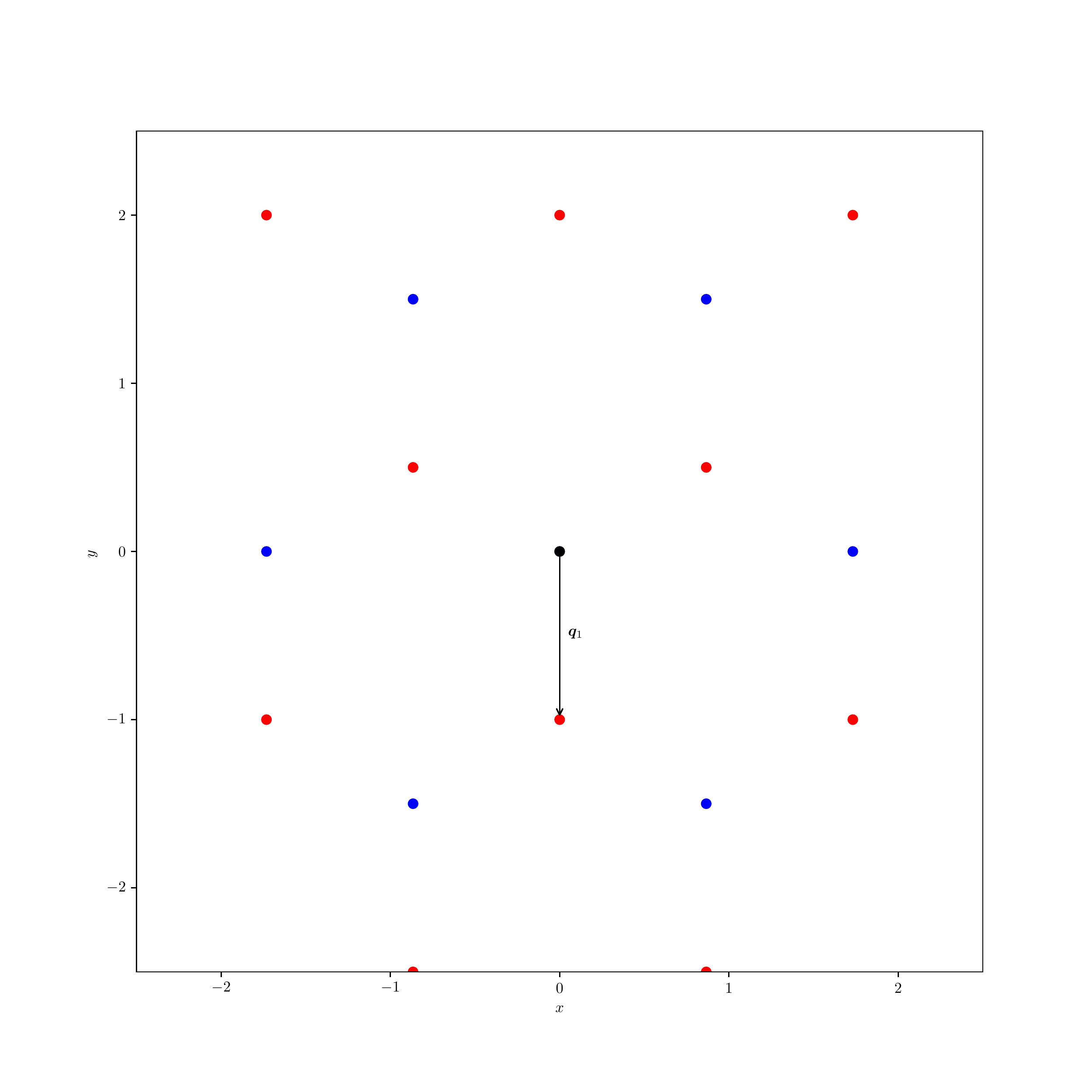}
\includegraphics[scale=.3]{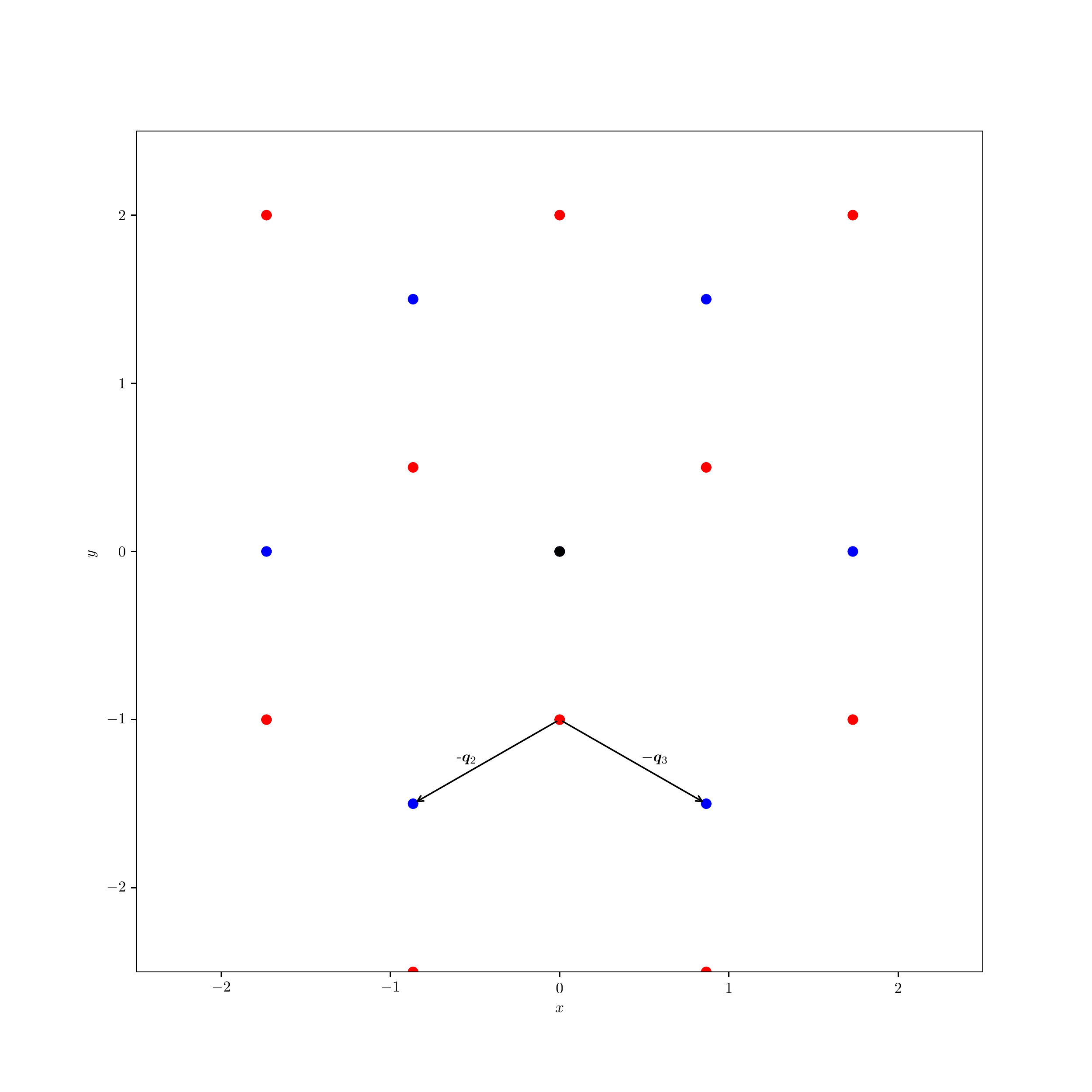}
\caption{Illustration of the action of $H^1$ as hopping in the momentum space lattice described by equations \eqref{eq:H1_0} (left, starting at $\vec{0}$) and \eqref{eq:H1_q1} (right, starting at $\vec{q}_1$). Although it appears that the hopping in these cases does not respect $\frac{2 \pi}{3}$ rotation symmetry, this is an artifact of working with chiral basis functions which individually respect the rotation symmetry, see \eqref{eq:H1_0_2}. 
}
\label{fig:H1_0}
\end{figure}
\begin{proof}[Proof of Proposition \ref{prop:H1_L2K1}]
We will prove \eqref{eq:H1_chi_G}, the proofs of the other identities are similar and hence omitted. We have 
\begin{equation*}
    H^1 \chi^{\widetilde{\vec{G}},1} = \frac{1}{\sqrt{3 V}} \left( e^{i \vec{q}_1 \cdot \vec{r}} + e^{i \phi} e^{i (\vec{q}_1 + \vec{b}_1) \cdot \vec{r}} + e^{- i \phi} e^{i (\vec{q}_1 + \vec{b}_2) \cdot \vec{r}} \right) \left( e^{i \vec{G} \cdot \vec{r}} + e^{i (R^*_\phi \vec{G}) \cdot \vec{r}} + e^{i ( (R^*_\phi)^2 \vec{G}) \cdot \vec{r}} \right) {e}_4.
\end{equation*}
Multiplying out we have
\begin{equation}
\begin{split}
    &\frac{1}{\sqrt{3 V}} \left( e^{i (\vec{q}_1 + \vec{G}) \cdot \vec{r}} + e^{i \phi} e^{i (\vec{q}_1 + \vec{G} + \vec{b}_1) \cdot \vec{r}} + e^{- i \phi} e^{i (\vec{q}_1 + \vec{G} + \vec{b}_2) \cdot \vec{r}} \right. \\
    &+ e^{i ( \vec{q}_1 + ( R^*_\phi \vec{G} ) ) \cdot \vec{r}} + e^{i \phi} e^{i (\vec{q}_1 + (R^*_\phi \vec{G}) + \vec{b}_1) \cdot \vec{r}} + e^{- i \phi} e^{i (\vec{q}_1 + (R^*_\phi \vec{G}) + \vec{b}_2) \cdot \vec{r}}   \\
    &\left. + e^{i ( \vec{q}_1 + ( (R^*_\phi)^2 \vec{G} ) )\cdot \vec{r}} + e^{i \phi} e^{i (\vec{q}_1 + ( (R^*_\phi)^2 \vec{G} ) + \vec{b}_1 ) \cdot \vec{r}} + e^{- i \phi} e^{i (\vec{q}_1 + ( (R^*_\phi)^2 \vec{G} ) + \vec{b}_2 ) \cdot \vec{r}} \right)   \\
    = &\frac{1}{\sqrt{3 V}} \left( e^{i (\vec{q}_1 + \vec{G}) \cdot \vec{r}} + e^{i \phi} e^{i (\vec{q}_1 + \vec{G} + \vec{b}_1) \cdot \vec{r}} + e^{- i \phi} e^{i (\vec{q}_1 + \vec{G} + \vec{b}_2) \cdot \vec{r}} \right. \\
    &e^{i ( R^*_\phi ( \vec{q}_1 + \vec{G} + \vec{b}_1 ) ) \cdot \vec{r}} + e^{i \phi} e^{i ( R^*_\phi ( \vec{q}_1 + \vec{G} + \vec{b}_2) ) \cdot \vec{r}} + e^{- i \phi} e^{i ( R^*_\phi (\vec{q}_1 + \vec{G}) ) \cdot \vec{r}}   \\
    &\left. + e^{i ( (R^*_\phi)^2 ( \vec{q}_1 + \vec{G} + \vec{b}_2 ) ) \cdot \vec{r}} + e^{i \phi} e^{i ( (R^*_\phi)^2 ( \vec{q}_1 + \vec{G} ) ) \cdot \vec{r}} + e^{- i \phi} e^{i ( (R^*_\phi)^2 (\vec{q}_1 + \vec{G} + \vec{b}_1 ) \cdot \vec{r}} \right).   \\
    = &\frac{1}{\sqrt{3 V}} \left( e^{i (\vec{q}_1 + \vec{G}) \cdot \vec{r}} + e^{- i \phi} e^{i ( R^*_\phi ( \vec{q}_1 + \vec{G}) ) \cdot \vec{r}} + e^{i \phi} e^{i ( (R^*_\phi)^2 ( \vec{q}_1 + \vec{G} ) ) \cdot \vec{r}} \right)   \\
    &+ \frac{1}{\sqrt{3 V}} e^{i \phi} \left( e^{i (\vec{q}_1 + \vec{G} + \vec{b}_1) \cdot \vec{r}} + e^{- i \phi} e^{i ( R^*_\phi ( \vec{q}_1 + \vec{G} + \vec{b}_1) ) \cdot \vec{r}} + e^{i \phi} e^{i ( (R^*_\phi)^2 ( \vec{q}_1 + \vec{G} + \vec{b}_1 ) ) \cdot \vec{r}} \right)   \\
    &+ \frac{1}{\sqrt{3 V}} e^{- i \phi} \left( e^{i (\vec{q}_1 + \vec{G} + \vec{b}_2) \cdot \vec{r}} + e^{- i \phi} e^{i ( R^*_\phi ( \vec{q}_1 + \vec{G} + \vec{b}_2) ) \cdot \vec{r}} + e^{i \phi} e^{i ( (R^*_\phi)^2 ( \vec{q}_1 + \vec{G} + \vec{b}_2 ) ) \cdot \vec{r}} \right),
\end{split}
\end{equation}
from which \eqref{eq:H1_chi_G} follows.
\end{proof}
\begin{proposition} \label{prop:H1_L2Kminus1}
The operator $H^1$ maps $L^2_{K,1,-1,A} \rightarrow L^2_{K,1,1,B}$, and $L^2_{K,1,-1,B} \rightarrow L^2_{K,1,1,A}$. The action of $H^1$ on chiral basis functions is as follows:
\begin{equation} \label{eq:H1_q1_-}
    H^1 \chi^{\widetilde{\vec{q}_1},-1} = \oldhat{z}_{\vec{q}_1} \left( \sqrt{3} \chi^{\widetilde{\vec{0}}} + e^{- i \phi} \chi^{\widetilde{\vec{q}_1-\vec{q}_2},1} + e^{i \phi} \chi^{\widetilde{\vec{q}_1-\vec{q}_3},1} \right).
\end{equation}
For all $\vec{G} \in \Lambda^* \setminus \{ \vec{0} \}$,
\begin{equation} \label{eq:H1_G_-}
    H^1 \chi^{\widetilde{\vec{G}},-1} = \oldhat{z}_{\vec{G}} \left( \chi^{\widetilde{\vec{G} + \vec{q}_1},1} + e^{- i \phi} \chi^{\widetilde{\vec{G} + \vec{q}_2},1} + e^{i \phi} \chi^{\widetilde{\vec{G} + \vec{q}_3},1} \right).
\end{equation}
For all $\vec{G} \in \Lambda^*$,
\begin{equation} \label{eq:H1_q1_G_-}
    H^1 \chi^{\widetilde{\vec{G} + \vec{q}_1},-1} = \oldhat{z}_{\vec{G} + \vec{q}_1} \left( \chi^{\widetilde{\vec{G}},1} + e^{- i \phi} \chi^{\widetilde{\vec{G} + \vec{q}_1 - \vec{q}_2},1} + e^{i \phi} \chi^{\widetilde{\vec{G} + \vec{q}_1 - \vec{q}_3},1} \right).
\end{equation}
\end{proposition}
\begin{proof}
The proof is similar to that of Proposition \ref{prop:H1_L2K1} and is hence omitted.
\end{proof}

\section{Formal expansion of the zero mode} \label{sec:TKV_expansion}

We now bring to bear the developments of the preceding sections on the asymptotic expansion of the zero mode $\psi^\alpha \in L^2_{K,1,1}$ starting from $\Psi^0 = e_1 = \chi^{\widetilde{\vec{0}}}$. We first give the proof of Proposition \ref{prop:series_prop}.
\begin{proof}[Proof of Proposition \ref{prop:series_prop}]
We have seen that $\chi^{\widetilde{\vec{0}}} \in L^2_{K,1,1}$. By the calculations of the previous section, $H^1 \chi^{\widetilde{\vec{0}}} \in L^2_{K,1,-1}$ which is orthogonal to the null space of $H^0$. The general solution of $H^0 \Psi^1 = - H^1 \Psi^0$ is
\begin{equation}
    \Psi^1 = - P^\perp (H^0)^{-1} P^\perp H^1 \Psi^0 + C \Psi^0,
\end{equation}
where $C$ is an arbitrary constant, which is in $L^2_{K,1,1}$ by Proposition \ref{prop:H0_prop}. To ensure that $\Psi^1$ is orthogonal to $\Psi^0$ we take $C = 0$. It is clear that this procedure can be repeated to derive an expansion to all orders satisfying the conditions of Proposition \ref{prop:series_prop}.
\end{proof}
Our goal is to calculate $\Psi^n \in L^2_{K,1,1}$ satisfying the conditions of Proposition \ref{prop:H0_prop} up to $n = 8$. This amounts to calculating, for $n = 1$ to $n = 8$,
\begin{equation}
    \Psi^n = - P^\perp (H^0)^{-1} P^\perp H^1 \Psi^{n-1}.
\end{equation}
We do this algorithmically by repeated application of the following proposition, which combines Proposition \ref{prop:H1_L2K1} and Proposition \ref{prop:H0_inv_prop}.
\begin{proposition} \label{prop:H0_inv_H1}
The operator $- P^\perp (H^0)^{-1} P^\perp H^1$ maps $L^2_{K,1,1,A} \rightarrow L^2_{K,1,1,B}$ and $L^2_{K,1,1,B} \rightarrow L^2_{K,1,1,A}$. Its action on chiral basis functions is as follows: 
\begin{equation} \label{eq:H0_inv_H1_chi_0}
    - P^\perp (H^0)^{-1} P^\perp H^1 \chi^{\widetilde{\vec{0}}} = - \sqrt{3} \overline{\oldhat{z}_{\vec{q}_1}} \chi^{\widetilde{\vec{q}_1},1},
\end{equation}
and
\begin{equation} \label{eq:H0_inv_H1_chi_1}
    - P^\perp (H^0)^{-1} P^\perp H^1 \chi^{\widetilde{\vec{q}_1},1} = - \frac{e^{i \phi} \overline{ \oldhat{z}_{\vec{q}_1 - \vec{q}_2} }}{|\vec{q}_1 - \vec{q}_2|} \chi^{\widetilde{\vec{q}_1-\vec{q}_2},1} - \frac{e^{- i \phi} \overline{ \oldhat{z}_{\vec{q}_1 - \vec{q}_3} }}{|\vec{q}_1 - \vec{q}_3|} \chi^{\widetilde{\vec{q}_1-\vec{q}_3},1}.
\end{equation}
For all $\vec{G} \in \Lambda^* \setminus \{ \vec{0} \}$,  
\begin{equation} \label{eq:H0_inv_H1_chi_G}
\begin{split}
    &- P^\perp (H^0)^{-1} P^\perp H^1 \chi^{\widetilde{\vec{G}},1} = \\
    &- \frac{\overline{ \oldhat{z}_{\vec{G} + \vec{q}_1} }}{|\vec{G} + \vec{q}_1|} \chi^{\widetilde{\vec{G} + \vec{q}_1},1} - \frac{e^{i \phi} \overline{ \oldhat{z}_{\vec{G} + \vec{q}_2} }}{|\vec{G} + \vec{q}_2|} \chi^{\widetilde{\vec{G} + \vec{q}_2},1} - \frac{e^{- i \phi} \overline{ \oldhat{z}_{\vec{G} + \vec{q}_3} }}{|\vec{G} + \vec{q}_3|} \chi^{\widetilde{\vec{G} + \vec{q}_3},1}.
\end{split}
\end{equation}
For all $\vec{G} \in \Lambda^* \setminus \{ \vec{0} \}$,
\begin{equation} \label{eq:H0_inv_H1_chi_G_plus_q}
    - P^\perp (H^0)^{-1} P^\perp H^1 \chi^{\widetilde{\vec{G} + \vec{q}_1},1} = - \frac{\overline{\oldhat{z}_{\vec{G}}}}{|\vec{G}|} \chi^{\widetilde{\vec{G}},1} - \frac{e^{i \phi} \overline{ \oldhat{z}_{\vec{G} + \vec{q}_1 - \vec{q}_2} }}{|\vec{G} + \vec{q}_1 - \vec{q}_2|} \chi^{\widetilde{\vec{G} + \vec{q}_1 - \vec{q}_2},1} - \frac{e^{- i \phi} \overline{ \oldhat{z}_{\vec{G} + \vec{q}_1 - \vec{q}_3} } }{|\vec{G} + \vec{q}_1 - \vec{q}_3|} \chi^{\widetilde{\vec{G} + \vec{q}_1 - \vec{q}_3},1}.
\end{equation}
\end{proposition}
We now claim the following.
\begin{proposition} \label{prop:psi_alpha_expansion}
Let $\Psi^n$ be the sequence defined by Proposition \ref{prop:series_prop}. Then
\begin{equation} \label{eq:Psi_1}
    \Psi^1 = - \sqrt{3} i \chi^{\widetilde{\vec{q}_1},1},
\end{equation}
\begin{equation} \label{eq:Psi_2}
    \Psi^2 = \left( \frac{ \sqrt{3} - i }{ 2 } \right) \chi^{\widetilde{-\vec{b}_1},1} + \left( \frac{ \sqrt{3} + i }{ 2 } \right) \chi^{\widetilde{-\vec{b}_2},1},
\end{equation}
\begin{equation} \label{eq:Psi_3}
    \Psi^3 = \frac{1}{\sqrt{7}} \left( \frac{ \sqrt{7} - 3 \sqrt{21} i }{ 14 } \right) \chi^{\widetilde{\vec{q}_1 - \vec{b}_2},1} + \frac{1}{\sqrt{7}} \left( \frac{-\sqrt{7} - 3 \sqrt{21} i}{ 14 } \right) \chi^{\widetilde{\vec{q}_1 - \vec{b}_1},1},
\end{equation}
\begin{equation} \label{eq:Psi_4}
\begin{split}
    \Psi^4 = &\frac{1}{\sqrt{21}} \left( \frac{ -5 \sqrt{7} + \sqrt{21} i }{14} \right) \chi^{\widetilde{-\vec{b}_2},1} + \frac{1}{2 \sqrt{21}} \left( \frac{2 \sqrt{7} + \sqrt{21} i}{7} \right) \chi^{\widetilde{-2 \vec{b}_2},1}  \\
    &+ \frac{1}{\sqrt{21}} \left( \frac{-5\sqrt{7} - \sqrt{21}i}{14} \right) \chi^{\widetilde{-\vec{b}_1},1} + \frac{1}{2 \sqrt{21}} \left( \frac{2 \sqrt{7} - \sqrt{21} i}{7} \right) \chi^{\widetilde{-2 \vec{b}_1}}  \\
    &+ \frac{2 \sqrt{3}}{21} \chi^{\widetilde{-\vec{b}_1 - \vec{b}_2},1},
\end{split}
\end{equation}
\end{proposition}
\begin{proof}
Equations \eqref{eq:Psi_1} and \eqref{eq:Psi_2} follow immediately from \eqref{eq:H0_inv_H1_chi_0} and \eqref{eq:H0_inv_H1_chi_1} and using $\vec{q}_2 = \vec{q}_1 + \vec{b}_1$ and $\vec{q}_3 = \vec{q}_1 + \vec{b}_2$. The derivation of equation \eqref{eq:Psi_3} is more involved, so we give details. Using linearity, and applying \eqref{eq:H0_inv_H1_chi_G} twice, we find
\begin{equation}
\begin{split}
    &- P^\perp (H^0)^{-1} P^\perp H^1 \Psi^2 =   \\
    &\left( \frac{\sqrt{3} - i}{2} \right) \left( \frac{ \overline{ \oldhat{z}_{\vec{q}_1 - \vec{b}_2} } }{ |\vec{q}_1 - \vec{b}_2| } \chi^{\widetilde{\vec{q}_1 - \vec{b}_2},1} + \frac{ e^{i \phi} \overline{ \oldhat{z}_{\vec{q}_1 + \vec{b}_1 - \vec{b}_2} } }{ |\vec{q}_1 + \vec{b}_1 - \vec{b}_2| } \chi^{\widetilde{\vec{q}_1 + \vec{b}_1 - \vec{b}_2},1} + e^{- i \phi} \overline{ \oldhat{z}_{\vec{q}_1} } \chi^{\widetilde{\vec{q}_1},1} \right)    \\
    &+ \left( \frac{\sqrt{3} + i}{2} \right) \left( \frac{ \overline{ \oldhat{z}_{\vec{q}_1 - \vec{b}_1} } }{ |\vec{q}_1 - \vec{b}_1| } \chi^{\widetilde{\vec{q}_1 - \vec{b}_1},1} + e^{i \phi} \overline{ \oldhat{z}_{\vec{q}_1} } \chi^{\widetilde{\vec{q}_1},1} + \frac{ e^{- i \phi} \overline{ \oldhat{z}_{\vec{q}_1 + \vec{b}_2 - \vec{b}_1} } }{ |\vec{q}_1 + \vec{b}_2 - \vec{b}_1| } \chi^{\widetilde{\vec{q}_1 + \vec{b}_2 - \vec{b}_1},1} \right).
\end{split}
\end{equation}
First, the terms proportional to $\chi^{\widetilde{\vec{q}_1},1}$ cancel. Next, since $R_\phi ( \vec{q}_1 + \vec{b}_1 - \vec{b}_2 ) = \vec{q}_1 + \vec{b}_2 - \vec{b}_1$, we have $\chi^{\widetilde{\vec{q}_1 + \vec{b}_1 - \vec{b}_2},1} = \chi^{\widetilde{\vec{q}_1 + \vec{b}_2 - \vec{b}_1},1}$. These terms also cancel, leaving \eqref{eq:Psi_3}. The derivation of \eqref{eq:Psi_4} (and the higher corrections) is involved but does not depend on any new ideas, and is therefore omitted.
\end{proof}
We give the explicit forms of $\Psi^5$-$\Psi^8$ in the Supplementary Material.
\begin{remark}
Written out, \eqref{eq:Psi_1} and \eqref{eq:Psi_2} become 
\begin{equation}
    \Psi^1 = - \sqrt{3} i \frac{1}{\sqrt{3 V}} {e}_2 \left( e^{i \vec{q}_1 \cdot \vec{r}} + e^{i \vec{q}_2 \cdot \vec{r}} + e^{i \vec{q}_3 \cdot \vec{r}} \right),
\end{equation}
and
\begin{equation}
    \Psi^2 = - i e^{i \phi} \frac{1}{\sqrt{3 V}} {e}_1 \left( e^{i \vec{b}_1 \cdot \vec{r}} + e^{i (\vec{b}_2 - \vec{b}_1) \cdot \vec{r}} + e^{- i \vec{b}_2 \cdot \vec{r}} \right) + i e^{- i \phi} \frac{1}{\sqrt{3 V}} {e}_1 \left( e^{i \vec{b}_2 \cdot \vec{r}} + e^{- i \vec{b}_1 \cdot \vec{r}} + e^{i (\vec{b}_1 - \vec{b}_2) \cdot \vec{r}} \right),
\end{equation}
which agree with equation (24) of Tarnopolsky et al. \cite{Tarnopolsky2019} up to an overall factor of $\sqrt{V}$ (this factor cancels in the Fermi velocity so there is no discrepancy).
\end{remark}
Using orthonormality of the chiral basis functions, it is straightforward to calculate the norms of each of the $\Psi^n$. We have
\begin{proposition} \label{prop:Psi_n_norms}
\begin{equation}
    \| \Psi^0 \| = 1, \| \Psi^1 \| = \sqrt{3}, \| \Psi^2 \| = \sqrt{2}, \| \Psi^3 \| = \frac{\sqrt{14}}{7}, \| \Psi^4 \| = \frac{\sqrt{258}}{42}, \| \Psi^5 \| = \frac{ \sqrt{1968837} }{3458}
\end{equation}
\begin{equation}
    \| \Psi^6 \| = \frac{ \sqrt{106525799} }{31122}, \| \Psi^7 \| = \frac{ 2\sqrt{2129312323981473} }{ 624696345 }, \| \Psi^8 \| = \frac{ \sqrt{183643119755214454} }{ 4997570760 }.
\end{equation}
\end{proposition}
\begin{remark}\label{rem:remark_on_PsiN_norms}
Note that the sequence of norms of the expansion functions grows much slower than the pessimistic bound $\| \Psi^{N+1} \| \leq 3 \| \Psi^N \|, N = 0,1,2,...$ guaranteed by Proposition \ref{prop:H1H0bound}. The reason is that the bounds \eqref{eq:H0_worst} and \eqref{eq:H1_worst} are never attained. As $N$ becomes larger, the bound \eqref{eq:H0_worst} is very pessimistic because $\Psi^N$ is mostly made up of eigenfunctions of $H^0$ with eigenvalues strictly larger than $1$. The bound \eqref{eq:H1_worst} is also very pessimistic because it is attained only at delta functions, which can only be approximated with a superposition of a large number of eigenfunctions of $H^0$. It seems possible that a sharper bound could be proved starting from these observations, but we do not pursue this in this work.
\end{remark}

We finally give the proof of Proposition \ref{prop:bound_H1psi8}.
\begin{proof}[Proof of Proposition \ref{prop:bound_H1psi8}]
Explicit computation using Proposition \ref{prop:H1_L2K1} and orthonormality of the chiral basis functions gives
\begin{equation}
    \| H^1 \Psi^8 \| = \frac{ \sqrt{4855076200233765642} }{14992712280} \approx 0.147 \leq \frac{3}{20}.
\end{equation}
\end{proof}

\section{Proof of Proposition \ref{prop:mu_choice}} \label{sec:prop_proof}

We choose $\Xi$ as
\begin{equation} \label{eq:Omega}
    \Xi := \left\{ \parbox{15em}{ $L^2_{K,1}$-eigenfunctions of $H^0$ with \\ eigenvalues with magnitude $\leq 4 \sqrt{3}$ } \right\} \bigcup \left\{ \chi^{\widetilde{\vec{q}_1 - 4 \vec{b}_1 + \vec{b}_2},\pm 1}, \chi^{\widetilde{\vec{q}_1 + \vec{b}_1 - 4 \vec{b}_2},\pm 1} \right\}.
\end{equation}
Part 1. of Proposition \ref{prop:mu_choice} follows immediately from observing that $\chi^{\widetilde{\vec{q}_1 - 2 \vec{b}_1 - 2 \vec{b}_2},\pm 1}$ is not in $\Xi$ but $|\vec{q}_1 - 2 \vec{b}_1 - 2 \vec{b}_2| = 7$. That $\mu = 7$ is optimal can be seen from Figure \ref{fig:H0_evals}.

Part 2. follows from the fact that $\psi^{8,\alpha}$ depends only on eigenfunctions of $H^0$ with eigenvalues with magnitude less than or equal to $4 \sqrt{3}$. The largest eigenvalue is $4 \sqrt{3}$, coming from dependence of $\Psi^8$ on $\chi^{\widetilde{- 4 \vec{b}_2},1}$, since $|- 4 \vec{b}_2| = 4 \sqrt{3}$.

Part 3. can be seen from Figure \ref{fig:H0_evals}.

\begin{figure}
\centering
\includegraphics[scale=.35]{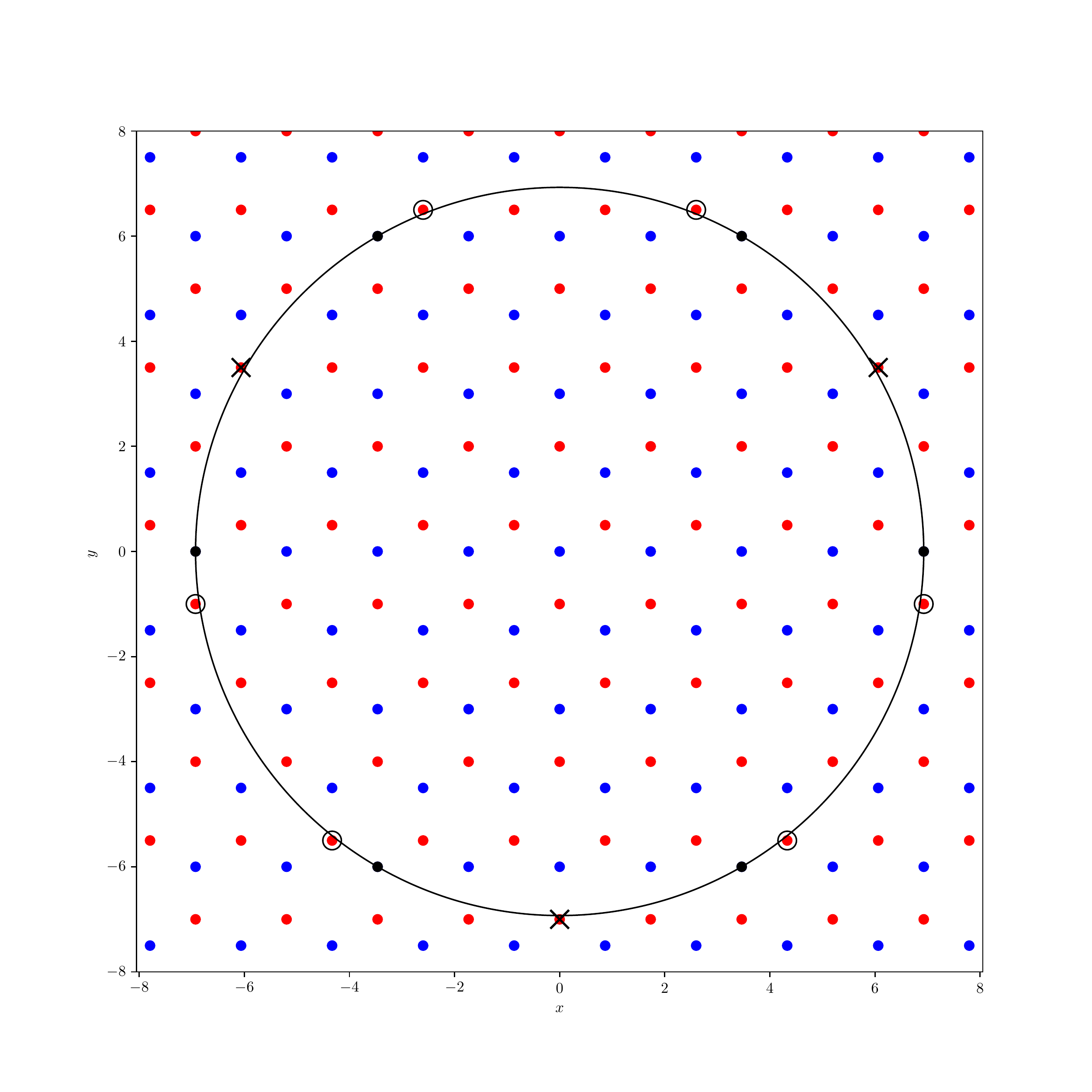}
\caption{
Illustration of $\Xi$ in the momentum space lattice. The circle has radius $4 \sqrt{3}$, so that every dot within the circle corresponds to two chiral basis vectors included in $\Xi$. Chiral basis vectors exactly $4 \sqrt{3}$ away from the origin, marked with black dots, are also included in $\Xi$. We also include in $\Xi$ the chiral basis vectors $\left\{ \chi^{\widetilde{\vec{q}_1 - 4 \vec{b}_1 + \vec{b}_2},\pm 1}, \chi^{\widetilde{\vec{q}_1 + \vec{b}_1 - 4 \vec{b}_2},\pm 1} \right\}$, which correspond to the dots marked with circles, which are distance $7$ (NB. $7 > 4 \sqrt{3}$) from the origin. We do not include the chiral basis vectors $\chi^{\widetilde{\vec{q}_1 - 2 \vec{b}_1 - 2 \vec{b}_2},\pm 1}$, marked with black crosses, which are also a distance $7$ from the origin. The reason for this is so that part 3 of Proposition \ref{prop:mu_choice} holds. With this choice, every dot in $\Xi$ has \emph{at most one} nearest neighbor lattice point outside of $\Xi$. It follows immediately from Propositions \ref{prop:H1_L2K1} and \ref{prop:H1_L2Kminus1} ($H^1$ acts by nearest neighbor hopping in the momentum space lattice) that $\| P_\Xi H^1 P_\Xi^\perp \| = 1$. Note that if we chose $\Xi$ to include $\chi^{\widetilde{\vec{q}_1 - 2 \vec{b}_1 - 2 \vec{b}_2},\pm 1}$ this would no longer hold because these basis functions would have two nearest neighbors outside $\Xi$, resulting in the worse bound $\| P_\Xi H^1 P_\Xi^\perp \| \leq \sqrt{2}$.
}
\label{fig:H0_evals}
\end{figure}

\section{Proof of Proposition \ref{as:PHP_gap}} \label{sec:verify_gap_assump}

\subsection{Proof of Theorem \ref{th:first_1}} \label{sec:prove_first_1}

We will prove Theorem \ref{th:first_1} starting from Theorem 11.5.1 of Parlett \cite{doi:10.1137/1.9781611971163}, where the proof can be found.
\begin{lemma} \label{lem:third}
Let $Q$ be an $n \times m$ matrix which satisfies $Q^\dagger Q = I_{m}$. Define $H = Q^\dagger A Q$ and $R = A Q - Q H$. Let $\{ \theta_j \}_{1 \leq j \leq m}$ denote the eigenvalues of $H$ (the Ritz values). Then $m$ of $A$'s eigenvalues $\{ \alpha_j \}_{1 \leq j \leq m}$ can be put into one-to-one correspondence with the $\{ \theta_j \}_{1 \leq j \leq m}$ in such a way that
\begin{equation}
    | \theta_j - \alpha_j | \leq \| R \|_2 \quad 1 \leq j \leq m.
\end{equation}
\end{lemma}
\begin{proof}[Proof of Theorem \ref{th:first_1}]
Let $Q$ be the matrix whose columns are ${v}_1,...,{v}_m$. Using orthonormality of the $v_j$, $Q^\dagger Q = I_m$ and
\begin{equation}
\begin{split}
    H = Q^\dagger A Q = \diag({\lambda}_1 , ... , {\lambda}_m) + \begin{pmatrix} \ip{ {v}_1 }{ {r}_1 } & ... & \ip{ {v}_1 }{ {r}_m } \\ ... & ... & ... \\ \ip{ v_n }{ r_1 } & ... & \ip{ v_n }{ r_n } \end{pmatrix}.
\end{split}
\end{equation}
We now prove that the eigenvalues of $H$, denoted by $\theta_j$, are close to the $\lambda_j$s. By the Gershgorin circle theorem, we have
\begin{equation}
    | \theta_i - \left( {\lambda}_i + \ip{ v_i }{ r_i } \right) | \leq \sum_{j \neq i}^m | \ip{ {v}_i }{ {r}_j } |,
\end{equation}
which implies, using $\| v_j \|_2 = 1$,
\begin{equation}
    | \theta_i - \lambda_i | = | \theta_i - \lambda_i - \ip{v_i}{r_i} + \ip{v_i}{r_i} | \leq \sum_{j = 1}^m | \ip{ v_i }{ r_j } | \leq m \sup_{1 \leq i \leq m} \| r_i \|_2. 
\end{equation}
We can now use Lemma \ref{lem:third} to bound the difference between the $\lambda_j$s and exact eigenvalues $\alpha_j$
\begin{equation}
    | {\lambda}_j - \alpha_j | = | {\lambda}_j - \theta_j + \theta_j - \alpha_j | \leq m \sup_{1 \leq i \leq m} \| r_i \|_2 + \| R \|_2,
\end{equation}
where $R := A Q - Q H = ( I - Q Q^\dagger) A Q$. Since $Q Q^\dagger$ projects onto the $v_j$, $R$ simplifies to
\begin{equation}
    R = ( I - Q Q^\dagger) R', \quad R' := \begin{pmatrix} {r}_1 & ... & r_m \end{pmatrix}.
\end{equation}
Since $Q Q^\dagger$ is a projection, so is $I - Q Q^\dagger$, and hence $\| R \|_2 \leq \| R' \|_2$. To bound $\| R' \|_2$, note that for any $v$ with $\| v \|_2 = 1$ we have
\begin{equation}
    \| R' v \|_2 = \ip{e_1}{v} r_1 + ... + \ip{e_m}{v} r_m \leq m \sup_{1 \leq i \leq m} \| r_i \|_2,
\end{equation}
where $e_j$ denote the standard orthonormal basis vectors. The result now follows.
\end{proof}

\subsection{Proof of Theorem \ref{th:second_1}} \label{sec:proof_second_1}

\begin{proof}[Proof of Theorem \ref{th:second_1}]
We start with the following Lemma which guarantees that numerically computed approximately orthonormal sets can be approximated by exactly orthonormal sets.
\begin{lemma} \label{lem:orthonormal}
Let $\tilde{v}_1,...,\tilde{v}_m$ be $n$-dimensional vectors, let $\ip{ \tilde{v}_i }{ \tilde{v}_j }_{comp}$ for $1 \leq i, j \leq m$ denote their numerically computed inner products, let $\epsilon$ denote machine epsilon, and assume $n \epsilon < 0.01$. Define
\begin{equation} \label{eq:mu_2}
    \mu := (1.01) n^2 \epsilon \sup_{1 \leq i \leq m} \| \tilde{v}_i \|_\infty^2 + \sup_i | \ip{ \tilde{v}_i }{ \tilde{v}_i }_{comp} - 1 | + \sup_{i \neq j} | \ip{ \tilde{v}_i }{ \tilde{v}_j }_{comp} |.
\end{equation}
Then, as long as $m \mu < \frac{1}{2}$, there is a set of $n$-dimensional orthonormal vectors $\oldhat{v}_1,...,\oldhat{v}_m$ which satisfy
\begin{equation}
    \| \oldhat{v}_j - \tilde{v}_j \|_2 \leq 2^{-1/2} m \mu, \quad 1 \leq j \leq m.
\end{equation}
\end{lemma}
\begin{proof}
Bounding the round-off error in computing inner products in the usual way (see, for example, Golub and Van Loan \cite{GoluVanl13} Chapter 2.7) and assuming that $n \epsilon < 0.01$ we have that for each $1 \leq i,j \leq m$, $\ip{ \tilde{v}_i }{ \tilde{v}_j } = \ip{ \tilde{v}_i }{ \tilde{v}_j }_{comp} + e_{ij}$
where $|e_{ij}| \leq (1.01) n \epsilon |\tilde{v}_i|^\top |\tilde{v}_j| \leq (1.01) n \epsilon \| \tilde{v}_i \|_2 \| \tilde{v}_j \|_2$.
Letting $\tilde{Q}$ denote the matrix whose columns are the $\tilde{v}_i$s, then $\tilde{Q}^\dagger \tilde{Q} - I_m = E$ where, for all $i \neq j$, $| E_{ij} | \leq | \ip{ \tilde{v}_i }{ \tilde{v}_j }_{comp} | + (1.01) n \epsilon \| \tilde{v}_i \|_2 \| \tilde{v}_j \|_2$, and, for all $i$, $| E_{ii} | \leq | 1 - \ip{ \tilde{v}_i }{ \tilde{v}_i }_{comp} | + (1.01) n \epsilon \| \tilde{v}_i \|_2^2$. Paying the price of factors of $\sqrt{n}$ to replace $\| \cdot \|_2$ norms by $\| \cdot \|_\infty$ norms, we can obtain a trivial bound on the maximal element of $E$, denoted $\| E \|_{max}$, by
\begin{equation}
    \| E \|_{max} \leq (1.01) n^2 \epsilon \left( \sup_i \| \tilde{v}_i \|_\infty \right)^2 + \sup_{i} |\ip{ \tilde{v}_i }{ \tilde{v}_i }_{comp} - 1| + \sup_{i,j} |\ip{ \tilde{v}_i }{ \tilde{v}_j }_{comp}|.
\end{equation}
Note this is nothing but $\mu$ in the statement of the theorem. Using the Gershgorin circle theorem we then have that the eigenvalues $\lambda$ of $\tilde{Q}^\dagger \tilde{Q}$ satisfy $| \lambda - 1 | \leq m \| E \|_{max}$. We claim that there are exact orthonormal vectors $\oldhat{v}_j$ near (in the $\| \cdot \|_2$-norm) to the $\tilde{v}_j$. To see this note that $\oldhat{Q} := \tilde{Q} ( \tilde{Q}^\dagger \tilde{Q} )^{-1/2}$ satisfies $\oldhat{Q}^\dagger \oldhat{Q} = I_m$, and
\begin{equation}
    \| \oldhat{Q} - \tilde{Q} \|_2 \leq \| \oldhat{Q} \|_2 \| ( \tilde{Q}^\dagger \tilde{Q} )^{1/2} - I_m \|_2 = \| ( \tilde{Q}^\dagger \tilde{Q} )^{1/2} - I_m \|_2.
\end{equation}
Let $\lambda_{max}$ and $\lambda_{min}$ denote the maximum and minimum eigenvalues of $\tilde{Q}^\dagger \tilde{Q}$ respectively. Then $\left\| ( \tilde{Q}^\dagger \tilde{Q} )^{1/2} - I_m \right\|_2 \leq \max\left\{ | \lambda_{min}^{1/2} - 1 |, | \lambda_{max}^{1/2} - 1 | \right\}$. Since $\lambda_{min}$ is bounded below by $1 - m \| E \|_{max}$ and $\lambda_{max}$ is bounded above by $1 + m \| E \|_{max}$ we have 
\begin{equation}
    \| ( \tilde{Q}^\dagger \tilde{Q} )^{1/2} - I_m \|_2 \leq \max \left\{ | ( 1 + m \| E \|_{max} )^{1/2} - 1 |, | ( 1 - m \| E \|_{max} )^{1/2} - 1 | \right\}.
\end{equation}
Using Taylor's theorem, for $|x| < \frac{1}{2}$ we have that $| ( 1 + x )^{1/2} - 1 | \leq 2^{-1/2} |x|$ and $| ( 1 - x )^{1/2} - 1 | \leq 2^{-1/2} |x|$. Since by assumption $m \| E \|_{max} < \frac{1}{2}$ we conclude
\begin{equation}
    \| \tilde{Q} - \oldhat{Q} \|_2 \leq \| ( \tilde{Q}^\dagger \tilde{Q} )^{1/2} - I_m \|_2 \leq 2^{-1/2} m \| E \|_{max}.
\end{equation}
Letting $\oldhat{v}_j$ denote the columns of $\oldhat{Q}$ and noting that $\| \oldhat{v}_j - \tilde{v}_j \|_2 \leq \| \oldhat{Q} - \tilde{Q} \|_2$ for all $1 \leq j \leq m$ the result is proved.
\end{proof}
Using Lemma \ref{lem:orthonormal}, we have that there exists an exactly orthonormal set $\{ \oldhat{v}_j \}_{1 \leq j \leq m}$ nearby to the set $\{ \tilde{v}_j \}_{1 \leq j \leq m}$. We now want to bound the residuals of the $\oldhat{v}_j$ in terms of numerically computable quantities.
We start with the following easy lemma whose proof is a straightforward manipulation.
\begin{lemma} \label{lem:easy}
Let $A$ be an $n \times n$ Hermitian matrix and suppose that $\oldhat{r} := ( A - \tilde{\lambda} I ) \oldhat{v}$ and $\tilde{r} := ( A - \tilde{\lambda} I ) \tilde{v}$. Then
\begin{equation}
    \| \oldhat{r} \|_2 \leq \left( \| A \|_2 + | \tilde{\lambda} | \right)  \| \oldhat{v} - \tilde{v} \|_2 + \| \tilde{r} \|_2.
\end{equation}
\end{lemma}
The following lemma quantifies the error in approximating exact residuals by numerically computed values.
\begin{lemma} \label{lem:residual}
Let $A$ be a Hermitian $n \times n$ matrix and let $\tilde{A}$ denote the matrix whose entries are those of $A$ evaluated as floating-point numbers. Let $\left[ \left( \tilde{A} - \tilde{\lambda} I \right) \tilde{v} \right]_{comp}$ denote the numerically computed value of $\left( \tilde{A} - \tilde{\lambda} I \right) \tilde{v}$ in floating-point arithmetic. Then $\tilde{r} := ( A - \tilde{\lambda} I ) \tilde{v}$ satisfies
\begin{equation}
    \| \tilde{r} \|_2 \leq 
    n^{1/2} \left\| \left[ (\tilde{A} - \tilde{\lambda} I) \tilde{v} \right]_{comp} \right\|_{max} + (1.01) n^{5/2} \epsilon \| \tilde{A} - \tilde{\lambda} I \|_{max} \| \tilde{v} \|_{\infty} + n \epsilon \| A \|_{max} \| \tilde{v} \|_\infty.
\end{equation}
\end{lemma}
\begin{proof}
For matrices $A$ and $B$ with entries $A_{ij}$ and $B_{ij}$ we will write $|A|$ to denote the matrix with entries $|A_{ij}|$ for all $i,j$, and $|A| \leq |B|$ if $|A_{ij}| \leq |B_{ij}|$ for all $i,j$. It is straightforward to see that (see Chapter 2.7 of Golub and Van Loan \cite{GoluVanl13}) $\tilde{A} = A + F$, where $|F| < \epsilon |A|$. Also, $( \tilde{A} - \tilde{\lambda} I ) \tilde{v} = \left[ ( \tilde{A} - \tilde{\lambda} I ) \tilde{v} \right]_{comp} + g$, where $|g| \leq (1.01) n \epsilon | (\tilde{A} - \tilde{\lambda} I) | | \tilde{v} |$. Now note that $(A - \tilde{\lambda} I) \tilde{v} = (\tilde{A} - \tilde{\lambda} I) \tilde{v} - F \tilde{v}$, so that 
\begin{equation}
    \| \tilde{r} \|_2 \leq \left\| \left[ (\tilde{A} - \tilde{\lambda} I) \tilde{v} \right]_{comp} \right\|_2 + \| g \|_2 + \| F \tilde{v} \|_2. 
\end{equation}
Noting that
\begin{equation}
    \| g \|_2 = (1.01) n \epsilon \| | \tilde{A} - \tilde{\lambda} I | | \tilde{v} | \|_2 \leq (1.01) n^{5/2} \epsilon \| \tilde{A} - \tilde{\lambda} I \|_{max} \| \tilde{v} \|_{\infty},
\end{equation}
\begin{equation}
    \| F \tilde{v} \|_2 \leq \epsilon \| |A| |\tilde{v}| \|_2 \leq n \epsilon \| A \|_{max} \| \tilde{v} \|_\infty,
\end{equation}
and
\begin{equation}
    \left\| \left[ (\tilde{A} - \tilde{\lambda} I) \tilde{v} \right]_{comp} \right\|_2 \leq n^{1/2} \left\| \left[ (\tilde{A} - \tilde{\lambda} I) \tilde{v} \right]_{comp} \right\|_{max},
\end{equation}
the result is proved.
\end{proof}
We now prove estimate \eqref{eq:bound_1}. Applying Lemma \ref{lem:orthonormal} to the set $\{ \tilde{v}_j \}_{1 \leq j \leq m}$ yields an orthonormal set $\{ \oldhat{v}_j \}_{1 \leq j \leq m}$ such that $\| \oldhat{v}_j - \tilde{v}_j \|_2 \leq 2^{-1/2} m \mu$ where $\mu$ is as in \eqref{eq:mu_2}. By Lemma \ref{lem:easy} we have that 
\begin{equation}
    \| \oldhat{r}_j \|_2 \leq 2^{-1/2} m \left( \| A \|_2 + |\tilde{\lambda}_j| \right) \mu + \| \tilde{r}_j \|_2, \quad 1 \leq j \leq m.
\end{equation}
The estimate now follows easily upon applying Lemma \ref{lem:residual} and taking the sup over $j$.
\end{proof}

\subsection{Proof of Theorem \ref{thm:Taylor_for_eigenvalues_1}} \label{sec:proof_Taylor_eigenvalues} 

\begin{proof}[Proof of Theorem \ref{thm:Taylor_for_eigenvalues_1}]
The proof is a simple consequence of the min-max characterization of eigenvalues of Hermitian matrices. By min-max (here $U$ denotes a subspace of $\field{C}^n$),
\begin{equation}
    |\lambda_j^\alpha - \lambda_j^{\alpha_0}| = \min_{\dim U = j} \max_{\substack{v \in U \\ v \neq 0}} \left| \frac{ \ip{ v }{ ( A^\alpha - A^{\alpha_0} ) v } }{ \ip{ v }{ v } } \right|,
\end{equation}
on the other hand, for any fixed $v \neq 0$ we have by Taylor's theorem
\begin{equation}
    \left| \frac{ \ip{ v }{ ( A^\alpha - A^{\alpha_0} ) v } }{ \ip{ v }{ v } } \right| \leq |\alpha - \alpha_0| \sup_{\beta \in [\alpha_0,\alpha]} \left| \de_\beta \frac{ \ip{ v }{ ( A^\beta - A^{\alpha_0} ) v } }{ \ip{ v }{ v } } \right| \leq |\alpha - \alpha_0| \sup_{\beta \in [\alpha_0,\alpha]} \| \de_\beta A^\beta \|_2,
\end{equation}
and the result follows immediately.
\end{proof}

\subsection{Proof of Proposition \ref{prop:bound_de_alpha_H}} \label{sec:proof_de_alpha_H} 

\begin{proof}[Proof of Proposition \ref{prop:bound_de_alpha_H}]
Differentiating $H^\alpha_{\Xi}$ yields
\begin{equation}
    \de_\alpha H^\alpha_{\Xi} = ( - \de_\alpha Q^\alpha ) P_{\Xi} H^\alpha P_{\Xi} Q^{\alpha,\perp} + Q^{\alpha,\perp} P_{\Xi} H^1 P_{\Xi} Q^{\alpha,\perp} + Q^{\alpha,\perp} P_{\Xi} H^\alpha P_{\Xi} ( - \de_\alpha Q^\alpha ).
\end{equation}
For $\alpha < 1$ we have $\| P_{\Xi} H^\alpha P_{\Xi} \|_2 \leq 10$, and $\| H^1 \|_2 \leq 3$. It remains only to estimate $\| \de_\alpha Q^\alpha \|_2$. Using Dirac notation to represent $L^2_{K,1}$-projections we have
\begin{equation}
    Q^\alpha = \ket{\Psi^{(8)}}\bra{\Psi^{(8)}} = \ket{ \sum_{m = 0}^8 \alpha^m \Psi^m }\bra{ \sum_{n = 0}^8 \alpha^n \Psi^n } = \sum_{m = 0}^8 \sum_{n = 0}^8 \alpha^{m + n} \ket{ \Psi^m } \bra{ \Psi^n },
\end{equation}
so that
\begin{equation}
    \de_\alpha Q^\alpha = \sum_{m = 0}^8 \sum_{n = 0}^8 (m + n) \alpha^{m + n - 1}  \ket{ \Psi^m }\bra{ \Psi^n }.
\end{equation}
Using $\| \ket{\Psi^m} \bra{\Psi^n} \|_2 \leq \| \Psi^m \|_2 \| \Psi^n \|_2$, and $\max_{0 \leq j \leq 8} \| \Psi^j \|_2 \leq \sqrt{3}$ by Proposition \ref{prop:Psi_n_norms}, we have, for $\alpha \leq 1$,
\begin{equation}
    \| \de_\alpha Q^\alpha \|_2 \leq 3 \sum_{m = 0}^8 \sum_{n = 0}^8 (m + n) = 1944. 
\end{equation}
Putting everything together we conclude
\begin{equation}
    \sup_{0 \leq \alpha \leq \frac{7}{10}} \| \de_\alpha H^\alpha_{\Xi} \|_2 \leq 2 \times 10 \times 1944 + 3 = 38883.
\end{equation}
\end{proof}

\section{Proof of Proposition \ref{prop:Fermi_v_expansion}} \label{sec:Fermi_v_expansion}

We can now prove Proposition \ref{prop:Fermi_v_expansion}. We start by proving \eqref{eq:denominator_exp}. 

\subsection{Proof of \eqref{eq:denominator_exp}}

We now prove \eqref{eq:denominator_exp}. It is straightforward to derive
\begin{equation} \label{eq:expansion_denom}
\begin{split}
    \ip{ \sum_{n = 0}^8 \alpha^n \Psi^n }{ \sum_{n = 0}^8 \alpha^n \Psi^n } = &\sum_{n = 0}^{8} \sum_{j = 0}^n \ip{\Psi^j}{\Psi^{n-j}} \alpha^n \\
    &+ \sum_{n = 0}^{7} \sum_{j = 0}^n \ip{ \Psi^{8-j} }{ \Psi^{8 - (n - j)} } \alpha^{16-n}.   \\
\end{split}
\end{equation}
We now make two observations which simplify the computation. First, recall that the operator $-P^\perp (H^0)^{-1} P^\perp H^1$ maps $L^2_{K,1,1,A} \rightarrow L^2_{K,1,1,B}$ and $L^2_{K,1,1,B} \rightarrow L^2_{K,1,1,A}$. It follows that $\Psi^0 \in L^2_{K,1,1,A}$, $\Psi^1 \in L^2_{K,1,1,B}$, $\Psi^2 \in L^2_{K,1,1,A}$, and so on, and hence
\begin{equation}
    \ip{ \Psi^{2 i} }{ \Psi^{2 j + 1} } = 0 \quad \forall i, j \in \{0,1,2,...\}.
\end{equation}
It follows that all terms in \eqref{eq:expansion_denom} with odd powers of $\alpha$ vanish. Second, note that since $\Psi^0 \in \ran P$ while $\Psi^n \in \ran P^\perp$ for all $n \geq 1$, we have that
\begin{equation}
    \ip{ \Psi^n }{ \Psi^0 } = \ip{ \Psi^0 }{ \Psi^n } = 0 \quad \forall n \in \{1,2,...\}.
\end{equation}
Deriving \eqref{eq:denominator_exp} is then just a matter of computation using the properties of the chiral basis. For the leading term, we have
\begin{equation}
    \ip{ \Psi^0 }{ \Psi^0 } = \ip{ \chi^{\widetilde{\vec{0}}} }{ \chi^{\widetilde{\vec{0}}} } = 1.
\end{equation}
For the $\alpha^2$ term the only non-zero term is 
\begin{equation}
    \ip{ \Psi^1 }{ \Psi^1 } = \ip{ - \sqrt{3} i \chi^{\widetilde{\vec{q}_1},1} }{ - \sqrt{3} i \chi^{\widetilde{\vec{q}_1},1} } = 3,
\end{equation}
using \eqref{eq:Psi_1}. For the $\alpha^4$ term, the possible non-zero terms are 
\begin{equation}
    \ip{ \Psi^3 }{ \Psi^1 } + \ip{ \Psi^2 }{ \Psi^2 } + \ip{ \Psi^1 }{ \Psi^3 },
\end{equation}
but $\Psi^3$ and $\Psi^1$ depend on orthogonal chiral basis vectors (see \eqref{eq:Psi_1} and \eqref{eq:Psi_3}) so we are left with
\begin{equation}
\begin{split}
    &\ip{ \Psi^2 }{ \Psi^2 } \\
    &= \ip{ \left(\frac{\sqrt{3} - i}{2}\right) \chi^{\widetilde{-\vec{b}_1},1} + \left( \frac{\sqrt{3} + i}{2} \right) \chi^{\widetilde{-\vec{b}_2},1} }{ \left(\frac{\sqrt{3} - i}{2}\right) \chi^{\widetilde{-\vec{b}_1},1} + \left( \frac{\sqrt{3} + i}{2} \right) \chi^{\widetilde{-\vec{b}_2},1} } = 2,
\end{split}
\end{equation}
using \eqref{eq:Psi_2} and orthgonality of $\chi^{\widetilde{-\vec{b}_1},1}$ and $\chi^{\widetilde{-\vec{b}_2},1}$. We omit the derivation of the higher terms since the derivations do not require any new ideas.

\subsection{Proof of \eqref{eq:numerator_exp}}

It is straightforward to derive
\begin{equation} \label{eq:expansion_num}
\begin{split}
    \ip{ \sum_{n = 0}^8 \alpha^n \Psi^{n*}(-\vec{r}) }{ \sum_{n = 0}^8 \alpha^n \Psi^n(\vec{r}) } = &\sum_{n = 0}^{8} \sum_{j = 0}^n \ip{\Psi^{j*}(-\vec{r})}{\Psi^{n-j}(\vec{r})} \alpha^n \\
    &+ \sum_{n = 0}^{7} \sum_{j = 0}^n \ip{ \Psi^{8-j*}(-\vec{r}) }{ \Psi^{8 - (n - j)}(\vec{r}) } \alpha^{16-n}.   \\
\end{split}
\end{equation}
We now note the following.
\begin{proposition} \label{prop:chiral_basis_conj}
Let $\chi$ be a chiral basis function in $L^2_{K,1,1}$. Then $\chi^*(-\vec{r}) = \chi(\vec{r})$.
\end{proposition}
\begin{proof}
The proof follows immediately from the explicit forms of the chiral basis functions in $L^2_{K,1,1}$ given by \eqref{eq:A0_chiral_efuncs}-\eqref{eq:A_chiral_efuncs}-\eqref{eq:B_chiral_efuncs} and the observation that for any $\vec{k} \in \field{R}^2$, $\left( e^{i \vec{k} \cdot (-\vec{r})} \right)^* = e^{i \vec{k} \cdot \vec{r}}$.
\end{proof}
Using Proposition \ref{prop:chiral_basis_conj} and the same two observations as in the previous section we have that the only non-zero terms in \eqref{eq:expansion_num} are those with even powers of $\alpha$, and that other than the leading term, terms involving $\Psi^0$ do not contribute. The calculation is then similar to the previous case. For the leading order term we have
\begin{equation}
    \ip{ \Psi^{0*}(-\vec{r}) }{ \Psi^0(\vec{r}) } = \ip{ \chi^{\widetilde{\vec{0}}} }{ \chi^{\widetilde{\vec{0}}} } = 1.
\end{equation}
The only non-zero $\alpha^2$ term is
\begin{equation}
    \ip{ \Psi^{1*}(-\vec{r}) }{ \Psi^1(\vec{r}) } = \ip{ \sqrt{3} i \chi^{\widetilde{\vec{q}_1},1} }{ - \sqrt{3} i \chi^{\widetilde{\vec{q}_1},1} } = - 3.
\end{equation}
The only non-zero $\alpha^4$ term is
\begin{equation}
\begin{split}
    &\ip{ \Psi^{2*}(-\vec{r}) }{ \Psi^2(\vec{r}) } \\
    &= \ip{ \left(\frac{\sqrt{3} + i}{2}\right) \chi^{\widetilde{-\vec{b}_1},1} + \left( \frac{\sqrt{3} - i}{2} \right) \chi^{\widetilde{-\vec{b}_2},1} }{ \left(\frac{\sqrt{3} - i}{2}\right) \chi^{\widetilde{-\vec{b}_1},1} + \left( \frac{\sqrt{3} + i}{2} \right) \chi^{\widetilde{-\vec{b}_2},1} }    \\
    &= \left(\frac{\sqrt{3} - i}{2}\right)^2 + \left(\frac{\sqrt{3} + i}{2}\right)^2 = 1.
\end{split}
\end{equation}
We omit the derivation of the higher terms since the derivations do not require any new ideas.

Proposition \ref{prop:series_prop} implies that the series expansion of $\psi^\alpha$ exists up to any order. We can therefore define formal infinite series by
\begin{equation} \label{eq:inf_num}
    \ip{ \sum_{n = 0}^\infty \alpha^n \Psi^{n*}(-\vec{r}) }{ \sum_{n = 0}^\infty \alpha^n \Psi^n(\vec{r}) }
\end{equation}
\begin{equation} \label{eq:inf_den}
    \ip{ \sum_{n = 0}^\infty \alpha^n \Psi^{n} }{ \sum_{n = 0}^\infty \alpha^n \Psi^n }.
\end{equation}
We then have the following.
\begin{proposition} \label{prop:approx_of_exp}
The expansions \eqref{eq:numerator_exp} and \eqref{eq:denominator_exp} approximate the formal series \eqref{eq:inf_num} and \eqref{eq:inf_den} up to terms of order $\alpha^{10}$.
\end{proposition}
\begin{proof}
The series agree exactly without any simplifications up to terms of $\alpha^9$. However, because the even and odd terms in the expansion of $\psi^\alpha$ are orthogonal (since they lie in $L^2_{K,1,1,A}$ and $L^2_{K,1,1,B}$ respectively), all terms with odd powers of $\alpha$ vanish in the expansions \eqref{eq:inf_num}-\eqref{eq:inf_den}. The series may disagree at order $\alpha^{10}$ because the infinite series includes terms arising from inner products of $\Psi^1$ and $\Psi^9$.
\end{proof}

\begin{acknowledgments}
M.L. and A.W.  were supported, in part, by ARO MURI Award No. W911NF-14-0247 and NSF DMREF Award No. 1922165.
\end{acknowledgments}

\section*{Data Availability}
The data that supports the findings of this study are available within the article and its supplementary material.

\bibliography{bibliography}

\section{Supplementary material}

\subsection{Chiral basis functions spanning the subspace $\Xi$} \label{sec:chiral_modes}

The chiral basis functions spanning the subspace $\Xi$ are as follows. We note which of the subspaces of $H^0$ acting on $L^2_{K,1}$ are spanned by the chiral basis vectors at the right.
\begin{align*}
    &\chi^{\widetilde{\vec{0}}} &\text{$0$ eigenspace}    \\
    &\chi^{\widetilde{\vec{q}_1},\pm 1} = \chi^{\widetilde{\vec{q}_1 + \vec{b}_1},\pm 1} = \chi^{\widetilde{\vec{q}_1 + \vec{b}_2},\pm 1} &\text{$\pm 1$ eigenspace}  \\
    &\chi^{\widetilde{-\vec{b}_1},\pm 1} = \chi^{\widetilde{\vec{b}_2},\pm 1} = \chi^{\widetilde{\vec{b}_1 - \vec{b}_2},\pm 1} &  \\
    &\chi^{\widetilde{-\vec{b}_2},\pm 1} = \chi^{\widetilde{\vec{b}_1},\pm 1} = \chi^{\widetilde{\vec{b}_2 - \vec{b}_1},\pm 1} &\text{$\pm \sqrt{3}$ eigenspace} \\
    &\chi^{\widetilde{\vec{q}_1 + \vec{b}_1 + \vec{b}_2},\pm 1} = \chi^{\widetilde{\vec{q}_1 + \vec{b}_1 - \vec{b}_2},\pm 1} = \chi^{\widetilde{\vec{q}_1 + \vec{b}_2 - \vec{b}_1},\pm 1} &\text{$\pm 2$ eigenspace}  \\
    &\chi^{\widetilde{\vec{q}_1 - \vec{b}_1},\pm 1} = \chi^{\widetilde{\vec{q}_1 + 2 \vec{b}_2},\pm 1} = \chi^{\widetilde{\vec{q}_1 + 2 \vec{b}_1 - \vec{b}_2},\pm 1} &   \\
    &\chi^{\widetilde{\vec{q}_1 - \vec{b}_2},\pm 1} = \chi^{\widetilde{\vec{q}_1 + 2 \vec{b}_1},\pm 1} = \chi^{\widetilde{\vec{q}_1 + 2 \vec{b}_2 - \vec{b}_1},\pm 1} &\text{$\pm \sqrt{7}$ eigenspace}   \\
    &\chi^{\widetilde{\vec{b}_1 + \vec{b}_2},\pm 1} = \chi^{\widetilde{\vec{b}_1 - 2 \vec{b}_2},\pm 1} = \chi^{\widetilde{\vec{b}_2 - 2 \vec{b}_1},\pm 1} & \\
    &\chi^{\widetilde{-\vec{b}_1 - \vec{b}_2},\pm 1} = \chi^{\widetilde{2 \vec{b}_2 - \vec{b}_1},\pm 1} = \chi^{\widetilde{2 \vec{b}_1 - \vec{b}_2},\pm 1} &\text{$\pm 3$ eigenspace} \\
    &\chi^{\widetilde{-2 \vec{b}_1},\pm 1} = \chi^{\widetilde{2 \vec{b}_2},\pm 1} = \chi^{\widetilde{2 \vec{b}_1 - 2 \vec{b}_2},\pm 1} &  \\
    &\chi^{\widetilde{-2 \vec{b}_2},\pm 1} = \chi^{\widetilde{2 \vec{b}_1},\pm 1} = \chi^{\widetilde{2 \vec{b}_2 - 2 \vec{b}_1},\pm 1} &\text{$\pm 2 \sqrt{3}$ eigenspace}  \\
    &\chi^{\widetilde{\vec{q}_1 + \vec{b}_1 - 2 \vec{b}_2},\pm 1} = \chi^{\widetilde{\vec{q}_1 - 2 \vec{b}_1 + 2 \vec{b}_2},\pm 1} = \chi^{\widetilde{\vec{q}_1 + 2 \vec{b}_1 + \vec{b}_2},\pm 1} & \\
    &\chi^{\widetilde{\vec{q}_1 + \vec{b}_2 - 2 \vec{b}_1},\pm 1} = \chi^{\widetilde{\vec{q}_1 - 2 \vec{b}_2 + 2 \vec{b}_1},\pm 1} = \chi^{\widetilde{\vec{q}_1 + 2 \vec{b}_2 + \vec{b}_1},\pm 1} &\text{$\pm \sqrt{13}$ eigenspace} \\
    &\chi^{\widetilde{\vec{q}_1 - \vec{b}_1 - \vec{b}_2},\pm 1} = \chi^{\widetilde{\vec{q}_1 - \vec{b}_1 + 3 \vec{b}_2},\pm 1} = \chi^{\widetilde{\vec{q}_1 + 3 \vec{b}_1 - \vec{b}_2},\pm 1} &\text{$\pm 4$ eigenspace} \\
    &\chi^{\widetilde{\vec{q}_1 - 2 \vec{b}_1},\pm 1} = \chi^{\widetilde{\vec{q}_1 + 3 \vec{b}_2},\pm 1} = \chi^{\widetilde{\vec{q}_1 + 3 \vec{b}_1 - 2 \vec{b}_2},\pm 1}& \\
    &\chi^{\widetilde{\vec{q}_1 - 2 \vec{b}_2},\pm 1} = \chi^{\widetilde{\vec{q}_1 + 3 \vec{b}_1},\pm 1} = \chi^{\widetilde{\vec{q}_1 + 3 \vec{b}_2 - 2 \vec{b}_1},\pm 1}&\text{$\pm \sqrt{19}$ eigenspace} \\
    &\chi^{\widetilde{-3\vec{b}_1 + \vec{b}_2},\pm 1} = \chi^{\widetilde{2 \vec{b}_1 - 3 \vec{b}_2},\pm 1} = \chi^{\widetilde{\vec{b}_1 + 2 \vec{b}_2},\pm 1}   \\
    &\chi^{\widetilde{-3 \vec{b}_1 + 2 \vec{b}_2},\pm 1} = \chi^{\widetilde{\vec{b}_1 - 3 \vec{b}_2},\pm 1} = \chi^{\widetilde{2 \vec{b}_1 + \vec{b}_2},\pm 1}  \\
    &\chi^{\widetilde{-\vec{b}_1- 2 \vec{b}_2},\pm 1} = \chi^{\widetilde{-2 \vec{b}_1 + 3 \vec{b}_2},\pm 1} = \chi^{\widetilde{3 \vec{b}_1 - \vec{b}_2},\pm 1}&   \\
    &\chi^{\widetilde{-\vec{b}_2- 2 \vec{b}_1},\pm 1} = \chi^{\widetilde{-2 \vec{b}_2 + 3 \vec{b}_1},\pm 1} = \chi^{\widetilde{3 \vec{b}_2 - \vec{b}_1},\pm 1}&\text{$\pm \sqrt{21}$ eigenspace} \\
    &\chi^{\widetilde{\vec{q}_1 + 2 \vec{b}_1 + 2 \vec{b}_2},\pm 1} = \chi^{\widetilde{\vec{q}_1 + 2 \vec{b}_1 -3 \vec{b}_2},\pm 1} = \chi^{\widetilde{\vec{q}_1 - 3 \vec{b}_1 + 2 \vec{b}_2},\pm 1}&\text{$\pm 5$ eigenspace}    \\
    &\chi^{\widetilde{-3 \vec{b}_1},\pm 1} = \chi^{\widetilde{3 \vec{b}_2},\pm 1} = \chi^{\widetilde{3 \vec{b}_1 - 3 \vec{b}_2},\pm 1}& \\
    &\chi^{\widetilde{-3 \vec{b}_2},\pm 1} = \chi^{\widetilde{3 \vec{b}_1},\pm 1} = \chi^{\widetilde{3 \vec{b}_2 - 3 \vec{b}_1},\pm 1}&\text{$\pm 3 \sqrt{3}$ eigenspace} \\
    &\chi^{\widetilde{\vec{q}_1 - 3 \vec{b}_1 + \vec{b}_2},\pm 1} = \chi^{\widetilde{\vec{q}_1 + 3 \vec{b}_1 - 3 \vec{b}_2},\pm 1} = \chi^{\widetilde{\vec{q}_1 + \vec{b}_1 + 3 \vec{b}_2},\pm 1}   \\
    &\chi^{\widetilde{\vec{q}_1 - 3 \vec{b}_1 + 3 \vec{b}_2}, \pm 1} = \chi^{\widetilde{\vec{q}_1 + \vec{b}_1 - 3 \vec{b}_2}, \pm 1}&\text{$\pm 2 \sqrt{7}$ eigenspace} \\
    &\chi^{\widetilde{\vec{q}_1 - 2 \vec{b}_1 - \vec{b}_2},\pm 1} = \chi^{\widetilde{\vec{q}_1 + 4 \vec{b}_1 - 2 \vec{b}_2},\pm 1} = \chi^{\widetilde{\vec{q}_1 - \vec{b}_1 + 4 \vec{b}_2},\pm 1}   \\
    &\chi^{\widetilde{\vec{q}_1 - 2 \vec{b}_1 + 4 \vec{b}_2},\pm 1} = \chi^{\widetilde{\vec{q}_1 - \vec{b}_1 - 2 \vec{b}_2},\pm 1} = \chi^{\widetilde{\vec{q}_1 + 4 \vec{b}_1 - \vec{b}_2}, \pm 1}&\text{$\pm \sqrt{31}$ eigenspace}    \\
    &\chi^{\widetilde{- 4 \vec{b}_1 + 2 \vec{b}_2}, \pm 1} = \chi^{\widetilde{2 \vec{b}_1 - 4 \vec{b}_2}, \pm 1} = \chi^{\widetilde{2 \vec{b}_1 + 2 \vec{b}_2}, \pm 1}  \\
    &\chi^{\widetilde{- 2 \vec{b}_1 - 2 \vec{b}_2}, \pm 1} = \chi^{\widetilde{4 \vec{b}_1 - 2 \vec{b}_2}, \pm 1} = \chi^{\widetilde{- 2 \vec{b}_1 + 4 \vec{b}_2}, \pm 1}&\text{$\pm 6$ eigenspace}    \\
    &\chi^{\widetilde{\vec{q}_1 - 3 \vec{b}_1}, \pm 1} = \chi^{\widetilde{\vec{q}_1 + 4 \vec{b}_1 - 3 \vec{b}_2}, \pm 1} = \chi^{\widetilde{\vec{q}_1 + 4 \vec{b}_2}, \pm 1}    \\
    &\chi^{\widetilde{\vec{q}_1 - 3 \vec{b}_1 + 4 \vec{b}_2}, \pm 1} = \chi^{\widetilde{\vec{q}_1 - 3 \vec{b}_2}, \pm 1} = \chi^{\widetilde{\vec{q}_1 + 4 \vec{b}_1}, \pm 1}&\text{$\pm \sqrt{37}$ eigenspace}   \\
    &\chi^{\widetilde{- 4 \vec{b}_1 + \vec{b}_2}, \pm 1} = \chi^{\widetilde{3 \vec{b}_1 - 4 \vec{b}_2}, \pm 1} = \chi^{\widetilde{\vec{b}_1 + 3 \vec{b}_2}, \pm 1}  \\
    &\chi^{\widetilde{- 4 \vec{b}_1 + 3 \vec{b}_2}, \pm 1} = \chi^{\widetilde{\vec{b}_1 - 4 \vec{b}_2}, \pm 1} = \chi^{\widetilde{3 \vec{b}_1 + \vec{b}_2}, \pm 1}  \\
    &\chi^{\widetilde{- 3 \vec{b}_1 - \vec{b}_2}, \pm 1} = \chi^{\widetilde{4 \vec{b}_1 - 3 \vec{b}_2}, \pm 1} = \chi^{\widetilde{-\vec{b}_1 + 4 \vec{b}_2}, \pm 1} \\
    &\chi^{\widetilde{- 3 \vec{b}_1 + 4 \vec{b}_2}, \pm 1} = \chi^{\widetilde{-\vec{b}_1 - 3 \vec{b}_2}, \pm 1} = \chi^{\widetilde{4 \vec{b}_1 - \vec{b}_2}, \pm 1}&\text{$\pm \sqrt{39}$ eigenspace}   \\
    &\chi^{\widetilde{\vec{q}_1 - 4 \vec{b}_1 + 2 \vec{b}_2}, \pm 1} = \chi^{\widetilde{\vec{q}_1 + 3 \vec{b}_1 - 4 \vec{b}_2}, \pm 1} = \chi^{\widetilde{\vec{q}_1 + 2 \vec{b}_1 + 3 \vec{b}_2}, \pm 1}    \\
    &\chi^{\widetilde{\vec{q}_1 - 4 \vec{b}_1 + 3 \vec{b}_2}, \pm 1} = \chi^{\widetilde{\vec{q}_1 + 2 \vec{b}_1 - 4 \vec{b}_2}, \pm 1} = \chi^{\widetilde{\vec{q}_1 + 3 \vec{b}_1 + 2 \vec{b}_2}, \pm 1}&\text{$\pm \sqrt{43}$ eigenspace}  \\
    &\chi^{\widetilde{-4\vec{b}_1}, \pm 1} = \chi^{\widetilde{4 \vec{b}_1 - 4 \vec{b}_2}, \pm 1} = \chi^{\widetilde{4 \vec{b}_2}, \pm 1}    \\
    &\chi^{\widetilde{- 4 \vec{b}_1 + 4 \vec{b}_2}, \pm 1} = \chi^{\widetilde{- 4 \vec{b}_2}, \pm 1} = \chi^{\widetilde{4 \vec{b}_1}, \pm 1}&\text{$\pm 4 \sqrt{3}$ eigenspace}.
\end{align*}
We finally add four out of the six modes which span the $\pm 7$ eigenspace
\begin{equation}
\begin{split}
    &\chi^{\widetilde{\vec{q}_1 - 4 \vec{b}_1 + \vec{b}_2}, \pm 1} = \chi^{\widetilde{\vec{q}_1 + 4 \vec{b}_1 - 4 \vec{b}_2}, \pm 1} = \chi^{\widetilde{\vec{q}_1 + \vec{b}_1 + 4 \vec{b}_2}, \pm 1}    \\
    &\chi^{\widetilde{\vec{q}_1 - 4 \vec{b}_1 + 4 \vec{b}_2}, \pm 1} = \chi^{\widetilde{\vec{q}_1 + \vec{b}_1 - 4 \vec{b}_2}, \pm 1} = \chi^{\widetilde{\vec{q}_1 + 4 \vec{b}_1 + \vec{b}_2}, \pm 1}.
\end{split}
\end{equation}

\subsection{Terms $\Psi^5$-$\Psi^8$ in the expansion} \label{sec:higher_terms}

Here we list terms $\Psi^5$-$\Psi^8$ in the expansion of $\psi^\alpha$ in powers of $\alpha$. The calculations were assisted by Sympy\cite{10.7717/peerj-cs.103}.
\begin{equation} \label{eq:Psi_5}
\begin{split}
    &\Psi^5 =    \\
    &\frac{ \sqrt{21} }{ 42 } \left( \frac{\sqrt{21} + 2 \sqrt{7} i}{7} \right) \chi^{\widetilde{\vec{q}_1 - \vec{b}_2},1} + \frac{ \sqrt{21} }{ 42 } \left( \frac{- \sqrt{21} + 2 \sqrt{7} i}{7} \right) \chi^{\widetilde{\vec{q}_1 - \vec{b}_1},1} \\
    &+ \frac{ 2 \sqrt{3} i }{ 21 } \chi^{\widetilde{\vec{q}_1 + \vec{b}_1 - \vec{b}_2},1} - \frac{ 4 \sqrt{3} i }{ 21 } \chi^{\widetilde{\vec{q}_1},1} - \frac{\sqrt{3} i}{42} \chi^{\widetilde{\vec{q}_1 - \vec{b}_2 - \vec{b}_1},1} \\
    &+ \frac{ \sqrt{273} }{ 546 } \left( \frac{ 5\sqrt{273} + 4\sqrt{91}i }{ 91 } \right) \chi^{\widetilde{\vec{q}_1 + \vec{b}_1 - 2 \vec{b}_2},1} + \frac{ \sqrt{399} }{ 798 } \left( \frac{2 \sqrt{399} - 11 \sqrt{133} i}{133} \right) \chi^{\widetilde{\vec{q}_1 - 2 \vec{b}_2},1} \\
    &+ \frac{ \sqrt{273} }{ 546 } \left( \frac{-5 \sqrt{273} + 4 \sqrt{91} i}{91} \right) \chi^{\widetilde{\vec{q}_1 + \vec{b}_2 - 2 \vec{b}_1},1} + \frac{ \sqrt{399} }{ 798 } \left( \frac{ -2\sqrt{399} - 11 \sqrt{133} i}{133} \right) \chi^{\widetilde{\vec{q}_1 - 2 \vec{b}_1},1},
\end{split}
\end{equation}
\begin{equation} \label{eq:Psi_6}
\begin{split}
    &\Psi^6 =    \\
    &\frac{\sqrt{91}}{42} \left( \frac{9 \sqrt{273} - 11 \sqrt{91} i}{182} \right) \chi^{\widetilde{-\vec{b}_1},1} + \frac{4 \sqrt{1729}}{5187} \left( \frac{-45\sqrt{5187} - 29\sqrt{1729}i}{3458} \right) \chi^{\widetilde{-2\vec{b}_1},1}    \\
    &+ \frac{\sqrt{91}}{42} \left( \frac{9 \sqrt{273} + 11\sqrt{91}i}{182} \right) \chi^{\widetilde{-\vec{b}_2},1} - \frac{\sqrt{3}}{26} \chi^{\widetilde{-2 \vec{b}_1 + \vec{b}_2},1} + \frac{\sqrt{133}}{2394} \left( \frac{9 \sqrt{399} - 17 \sqrt{133} i}{266} \right) \chi^{\widetilde{ - 3 \vec{b}_1 },1}   \\
    &+ \frac{ \sqrt{57} }{798} \left( \frac{59 \sqrt{19} - 9 \sqrt{57} i}{266} \right) \chi^{\widetilde{-2 \vec{b}_1 - \vec{b}_2},1} + \frac{ \sqrt{13} }{546} \left( \frac{- 17 \sqrt{39} - 41 \sqrt{13} i}{182} \right) \chi^{\widetilde{-3 \vec{b}_1 + \vec{b}_2},1}     \\
    &+ \frac{\sqrt{57}}{798} \left( \frac{59\sqrt{19} + 9 \sqrt{57} i}{266} \right) \chi^{\widetilde{-\vec{b}_1 - 2 \vec{b}_2},1} + \frac{4 \sqrt{1729}}{5187} \left( \frac{-45 \sqrt{5187} + 29 \sqrt{1729}i }{3458} \right) \chi^{\widetilde{-2\vec{b}_2},1}    \\
    &+ \frac{\sqrt{133}}{2394} \left( \frac{9 \sqrt{399} + 17 \sqrt{133} i}{266} \right) \chi^{\widetilde{-3 \vec{b}_2},1} + \frac{ \sqrt{13} }{546} \left( \frac{ -17 \sqrt{39} + 41 \sqrt{13} i }{ 182 } \right) \chi^{\widetilde{\vec{b}_1 - 3 \vec{b}_2},1},
\end{split}
\end{equation}
\begin{equation} \label{eq:Psi_7}
\begin{split}
    &\Psi^7 =   \\
    &\frac{ \sqrt{1032213} }{ 10374 } \left( \frac{-97\sqrt{1032213} - 562\sqrt{344071}i}{344071} \right) \chi^{\widetilde{\vec{q}_1 - \vec{b}_1},1} - \frac{\sqrt{3} i}{42} \chi^{\widetilde{\vec{q}_1},1} - \frac{2 \sqrt{3} i}{273} \chi^{\widetilde{\vec{q}_1 - \vec{b}_1 + \vec{b}_2},1} \\
    &+ \frac{ \sqrt{3549637} }{ 217854 } \left( \frac{-2621 \sqrt{3549637} + 1563\sqrt{10648911}i}{7099274} \right) \chi^{\widetilde{\vec{q}_1 - 2 \vec{b}_1},1} \\
    &+ \frac{ \sqrt{178087} }{ 24206 } \left( \frac{- 241\sqrt{178087} + 467 \sqrt{534261}i}{356174} \right) \chi^{\widetilde{\vec{q}_1 - 2 \vec{b}_1 + \vec{b}_2},1} \\
    &+ \frac{ \sqrt{1032213} }{10374} \left( \frac{97\sqrt{1032213} - 562\sqrt{344071}i}{344071} \right) \chi^{\widetilde{\vec{q}_1 - \vec{b}_2},1} \\
    &+ \frac{ \sqrt{178087} }{ 24206 } \left( \frac{ 241 \sqrt{178087} + 467 \sqrt{534261} i }{ 356174 } \right) \chi^{\widetilde{\vec{q}_1 - 2 \vec{b}_1 + 2 \vec{b}_2},1}    \\
    &+ \frac{\sqrt{4921}}{88578} \left( \frac{-53\sqrt{4921} - 75\sqrt{14763}i}{9842} \right) \chi^{\widetilde{\vec{q}_1 - 3 \vec{b}_1},1} \\
    &+ \frac{2 \sqrt{247}}{15561} \left( \frac{ -215 \sqrt{247} + 27 \sqrt{741}i }{3458} \right) \chi^{\widetilde{\vec{q}_1 - 3 \vec{b}_1 + \vec{b}_2},1} \\
    &+ \frac{ \sqrt{1767} }{24738} \left( \frac{-10\sqrt{1767} - 169\sqrt{589}i}{4123} \right) \chi^{\widetilde{\vec{q}_1 - 2 \vec{b}_1 - \vec{b}_2},1} + \frac{2 \sqrt{3} i }{2793} \chi^{\widetilde{\vec{q}_1 - \vec{b}_1 - \vec{b}_2},1}   \\
    &+ \frac{29 \sqrt{3} i}{19110} \chi^{\widetilde{\vec{q}_1 - 3 \vec{b}_1 + 2 \vec{b}_2},1} + \frac{\sqrt{1767}}{24738} \left( \frac{10 \sqrt{1767} - 169\sqrt{589}i}{4123} \right) \chi^{\widetilde{\vec{q}_1 - \vec{b}_1 - 2 \vec{b}_2},1}  \\
    &+ \frac{ \sqrt{3549637} }{ 217854 } \left( \frac{ 2621 \sqrt{3549637} + 1563\sqrt{10648911} i}{7099274} \right) \chi^{\widetilde{\vec{q}_1 - 2 \vec{b}_2},1} \\
    &+ \frac{\sqrt{4921}}{88578} \left( \frac{53 \sqrt{4921} - 75 \sqrt{14763} i}{ 9842 } \right) \chi^{\widetilde{\vec{q}_1 - 3 \vec{b}_2},1} \\
    &+ \frac{ 2 \sqrt{247} }{ 15561 } \left( \frac{215\sqrt{247} + 27\sqrt{741}i}{3458} \right) \chi^{\widetilde{\vec{q}_1 + \vec{b}_1 - 3 \vec{b}_2},1},
\end{split}
\end{equation}
\begin{equation} \label{eq:Psi_8}
\begin{split}
    &\Psi^8 =    \\
    &\frac{\sqrt{160797}}{10374} \left( \frac{ - 206\sqrt{53599} - 61\sqrt{160797}i }{ 53599 } \right) \chi^{\widetilde{-\vec{b}_1},1}    \\
    &+ \frac{\sqrt{1694251299}}{1307124} \left(\frac{16249\sqrt{564750433} - 10012\sqrt{1694251299}i}{564750433} \right) \chi^{\widetilde{-2 \vec{b}_1},1}    \\
    &+ \frac{317\sqrt{3}}{11466} \chi^{\widetilde{-\vec{b}_1 - \vec{b}_2},1} + \frac{\sqrt{160797}}{10374} \left( \frac{-206\sqrt{53599} + 61\sqrt{160797}i}{53599} \right) \chi^{\widetilde{-\vec{b}_2},1} \\
    &+ \frac{ 67\sqrt{3} }{ 16758 } \chi^{\widetilde{-2 \vec{b}_1 + \vec{b}_2},1} + \frac{ \sqrt{837273} }{620046} \left( \frac{-496\sqrt{279091} - 105\sqrt{837273}i}{279091} \right) \chi^{\widetilde{-3\vec{b}_1},1} \\
    &+ \frac{ \sqrt{997694607} }{20260422} \left( \frac{5849\sqrt{332564869} - 20785\sqrt{997694607}i}{665129738} \right) \chi^{\widetilde{-2\vec{b}_1 -\vec{b}_2},1} \\
    &+ \frac{ \sqrt{2667} }{13230} \left( \frac{-59\sqrt{889} - 5\sqrt{2667}i}{1778} \right) \chi^{\widetilde{-3 \vec{b}_1 + \vec{b}_2},1}    \\
    &+ \frac{ \sqrt{1694251299} }{1307124} \left( \frac{16249\sqrt{564750433} + 10012\sqrt{1694251299}i}{564750433} \right) \chi^{\widetilde{-2\vec{b}_2},1}  \\
    &+ \frac{ \sqrt{2667} }{13230} \left( \frac{-59\sqrt{889} + 5\sqrt{2667}i}{1778} \right) \chi^{\widetilde{-3 \vec{b}_1 + 2 \vec{b}_2},1}  \\
    &+ \frac{ \sqrt{14763} }{ 1062936 } \left( \frac{ 43 \sqrt{4921} - 32 \sqrt{14763} }{4921} \right) \chi^{\widetilde{-4 \vec{b}_1},1}  \\
    &+ \frac{ \sqrt{114919077} }{ 39454506 } \left( \frac{11413\sqrt{38306359} - 2767\sqrt{114919077}i}{76612718} \right) \chi^{\widetilde{-3 \vec{b}_1 - \vec{b}_2},1}   \\
    &+ \frac{ 2 \sqrt{57} }{ 46683 } \left( \frac{- 29\sqrt{19} - 31 \sqrt{57}i }{266} \right) \chi^{\widetilde{-4 \vec{b}_1 + \vec{b}_2},1}  \\
    &+ \frac{ 199 \sqrt{3} }{ 1038996 } \chi^{\widetilde{-2 \vec{b}_1 - 2 \vec{b}_2},1} - \frac{29 \sqrt{3}}{114660} \chi^{\widetilde{-4\vec{b}_1 + 2 \vec{b}_2},1}  \\
    &+ \frac{\sqrt{997694607}}{20260422} \left(\frac{5849\sqrt{332564869} + 20785\sqrt{997694607}i}{665129738}\right) \chi^{\widetilde{-\vec{b}_1 - 2 \vec{b}_2},1} \\
    &+ \frac{\sqrt{114919077}}{39454506} \left(\frac{11413\sqrt{38306359} + 2767\sqrt{114919077}i}{76612718}\right) \chi^{\widetilde{-\vec{b}_1-3 \vec{b}_2},1}   \\
    &+ \frac{\sqrt{837273}}{620046} \left( \frac{-496\sqrt{279091} + 105\sqrt{837273}i}{279091} \right) \chi^{\widetilde{-3\vec{b}_2},1}  \\
    &+ \frac{\sqrt{14763}}{1062936} \left( \frac{43\sqrt{4921} + 32\sqrt{14763}i}{4921} \right) \chi^{\widetilde{-4\vec{b}_2},1} + \frac{2 \sqrt{57}}{46683} \left( \frac{-29\sqrt{19} + 31\sqrt{57}i}{266} \right) \chi^{\widetilde{\vec{b}_1 - 4 \vec{b}_2},1}
\end{split}
\end{equation}

\subsection{Derivation of the TKV Hamiltonian from the Bistritzer-MacDonald model} \label{sec:TKV_from_BM}

The Bistritzer-MacDonald model of bilayer graphene, with relative twist angle $\theta$, is as follows\cite{Bistritzer2011}. Starting from two graphene layers laid exactly on top of each other (i.e., $AA$ stacking configuration), we rotate one layer (call this layer 1) clockwise by $\frac{\theta}{2}$, and the other layer (call this layer 2) counter-clockwise by $\frac{\theta}{2}$. Concentrating on layer 1 for a moment, and making the standard Dirac approximation for wavefunctions at the Dirac points, we have that when $\theta = 0$ there exist co-ordinate axes such that the effective Hamiltonian describing electrons near to the $K$-point of the graphene layers is $- i v_0 \vec{\sigma} \cdot \vec{\nabla}$, where $\vec{\sigma} = \left( \sigma_1, \sigma_2 \right)$ is the vector of Pauli matrices\cite{2009Castro-NetoGuineaPeresNovoselovGeim}. If we rotated the layer clockwise by $\theta$, the effective Hamiltonian would become $- i v_0 \vec{\sigma} \cdot \vec{\nabla}'$ where $\vec{\nabla}'$ is the gradient with respect to variables measured with respect to co-ordinate axes rotated clockwise by $\theta$, i.e.,
\begin{equation}
    \vec{\nabla}' = R_{\theta} \vec{\nabla}, \quad R_{\theta} = \begin{pmatrix} \cos \theta & - \sin \theta \\ \sin \theta & \cos \theta \end{pmatrix},
\end{equation}
and the effective Hamiltonian would be, in terms of the original variables, $- i v_0 \vec{\sigma}_{\theta} \cdot \vec{\nabla}$, where $\vec{\sigma}_{\theta} = e^{- i \frac{\theta}{2} \sigma_3} \vec{\sigma} e^{i \frac{\theta}{2} \sigma_3}$. We are thus lead to the following Hamiltonian describing electrons near to the ${K}$-points of the respective layers which are coupled through an ``inter-layer coupling potential'' $T(\vec{r})$
\begin{equation} \label{eq:bilayer_H}
    H = \begin{pmatrix} - i v_0 \vec{\sigma}_{\theta/2} \cdot \vec{\nabla} & T(\vec{r}) \\ T^\dagger(\vec{r}) & - i v_0 \vec{\sigma}_{-\theta/2} \cdot \vec{\nabla} \end{pmatrix},
\end{equation}
acting on $L^2(\field{R}^2;\field{C}^4)$ with domain $H^1(\field{R}^2;\field{C}^4)$. Note that $H$ ignores possible interactions between electrons with quasi-momentum away from the ${K}$-points of each layer, e.g., with the ${K}'$-points of each layer. Since the Fermi level occurs at the Dirac energy and interactions between ${K}$ and ${K}'$ points are small for small twist angles \cite{Catarina2019}, this is a reasonable simplification. The Hamiltonian \eqref{eq:bilayer_H} acts on wavefunctions
\begin{equation}
    \psi(\vec{r}) = \begin{pmatrix} \psi_1^A(\vec{r}), \psi_1^B(\vec{r}), \psi_2^A(\vec{r}), \psi_2^B(\vec{r}) \end{pmatrix}
\end{equation}
where $\psi^\sigma_\tau(\vec{r})$ represents the electron density near to the $K$ point (in momentum space) on sublattice $\sigma$ and on layer $\tau$. 

Under quite general assumptions, the inter-layer coupling has the following form \cite{Catarina2019}:
\begin{equation} \label{eq:T}
    T(\vec{r}) =
    \begin{pmatrix} w_{AA} ( e^{- i \vec{q}_1 \cdot \vec{r}} + e^{- i \vec{q}_2 \cdot \vec{r}} + e^{- i \vec{q}_3 \cdot \vec{r}} ) & w_{AB} ( e^{- i \vec{q}_1 \cdot \vec{r}} + e^{- i \vec{q}_2 \cdot \vec{r}} e^{- i \phi} + e^{- i \vec{q}_3 \cdot \vec{r}} e^{i \phi} )\\
    w_{AB} ( e^{- i \vec{q}_1 \cdot \vec{r}} + e^{- i \vec{q}_2 \cdot \vec{r}} e^{i \phi} + e^{- i \vec{q}_3 \cdot \vec{r}} e^{- i \phi} ) & w_{AA} ( e^{- i \vec{q}_1 \cdot \vec{r}} + e^{- i \vec{q}_2 \cdot \vec{r}} + e^{- i \vec{q}_3 \cdot \vec{r}} ) \end{pmatrix},
\end{equation}
where
\begin{equation}
    \vec{q}_1 = k_\theta \begin{pmatrix} 0, -1 \end{pmatrix}, \quad \vec{q}_{2,3} = \frac{ k_\theta }{ 2 } \begin{pmatrix} \pm \sqrt{3}, 1 \end{pmatrix}.
\end{equation}
Here $k_\theta = 2 k_D \sin(\theta/2)$ is the distance between the ${K}$ points of the different layers, and $k_D = |{K}_1| = |{K}_2|$ is the distance from the origin to the ${K}$ point of either layer. Let $\phi := \frac{2 \pi}{3}$, then $\vec{q}_2 = R_\phi \vec{q}_1$ and $\vec{q}_3 = R_\phi \vec{q}_2$ where $R_\phi$ is the matrix which rotates counterclockwise by $\phi$. Note that \eqref{eq:T} is written in such a way as to show clearly which couplings are between the $A$ lattices of the layers (proportional to $w_{AA}$ and occuring on the diagonal) and between the $A$ and $B$ lattices (proportional to $w_{AB}$ and occuring off the diagonal). 



\subsection{Translation and rotation symmetries of the Bistritzer-MacDonald model}

The operator $H$ essentially describes coupling on the scale of the bilayer moir\'e pattern. The moir\'e lattice vectors are
\begin{equation}
    \vec{a}_1 = \frac{2 \pi}{3 k_\theta} \begin{pmatrix} \sqrt{3}, 1 \end{pmatrix}, \quad \vec{a}_2 = \frac{2 \pi}{3 k_\theta} \begin{pmatrix} - \sqrt{3} , 1 \end{pmatrix}.
\end{equation}
We denote the moir\'e lattice generated by these vectors as $\Lambda$. 
It is straightforward to check that $H$ commutes with the ``phase-shifted'' moir\'e translation operators
\begin{equation}
    \tau_{\vec{v}} f(\vec{r}) := \diag( 1, 1, e^{i \vec{q}_1 \cdot \vec{v}}, e^{i \vec{q}_1 \cdot \vec{v}} ) \tilde{\tau}_{\vec{v}} f(\vec{r}), \quad \tilde{\tau}_{\vec{v}} f(\vec{r}) = f(\vec{r} + \vec{v}),
\end{equation}
for all $\vec{v} \in \Lambda$.

The operator also has rotational symmetry. Let $R_\phi$ be the matrix which rotates vectors by $\phi$ counter-clockwise
\begin{equation}
    R_\phi = \begin{pmatrix} - \frac{1}{2} & - \frac{\sqrt{3}}{2} \\ \frac{\sqrt{3}}{2} & - \frac{1}{2} \end{pmatrix}.
\end{equation}
Then $H$ commutes with the ``phase-shifted'' rotation operator
\begin{equation}
    \tilde{\mathcal{R}} f(\vec{r}) := \diag(1,e^{- i \phi},1,e^{- i \phi}) \mathcal{R} f(\vec{r}), \quad \mathcal{R} f(\vec{r}) = f( R_\phi \vec{r} ).
\end{equation}

\subsection{Deriving TKV from BM}

The first step to deriving Tarnopolsky-Kruchkov-Vishwanath's chiral model is to set $w_{AA} = 0$ in the Bistritzer-MacDonald model. Physically, this assumption is motivated by the observation that relaxation effects penalize the $AA$-stacking configuration, so that one expects \cite{PhysRevResearch.1.013001} $|w_{AA}| \ll |w_{AB}|$.

With this simplification, conjugating $H \rightarrow V_\theta H V_\theta^\dagger$ (here $\dagger$ represents the adjoint/Hermitian transpose) by 
\begin{equation}
    V_\theta := \diag( e^{i \theta/4}, e^{- i \theta/4}, e^{- i \theta/4}, e^{i \theta/4} )
\end{equation}
removes the explicit $\theta$ dependence of the Hamiltonian (although $H$ still depends on $\theta$ through $\vec{q}_1,\vec{q}_2,\vec{q}_3$) so that
\begin{equation}
    H = \begin{pmatrix} - i v_0 \vec{\sigma}_{\theta/2} \cdot \nabla & {T}_{AB}(\vec{r}) \\ T^\dagger_{AB}(\vec{r}) & - i v_0 \vec{\sigma}_{-\theta/2} \cdot \nabla \end{pmatrix}
\end{equation}
where
\begin{equation*}
    T_{AB} = \begin{pmatrix} 0 & w_{AB} ( e^{- i \vec{q}_1 \cdot \vec{r}} + e^{- i \vec{q}_2 \cdot \vec{r}} e^{- i \phi} + e^{- i \vec{q}_3 \cdot \vec{r}} e^{i \phi} ) \\
    w_{AB} ( e^{- i \vec{q}_1 \cdot \vec{r}} + e^{- i \vec{q}_2 \cdot \vec{r}} e^{i \phi} + e^{- i \vec{q}_3 \cdot \vec{r}} e^{- i \phi} ) & 0 \end{pmatrix}.
\end{equation*}

Conjugating once more $H \rightarrow \rho H \rho^\dagger$ by
\begin{equation}
    \rho = \begin{pmatrix} 1 & 0 & 0 & 0 \\ 0 & 0 & 1 & 0 \\ 0 & 1 & 0 & 0 \\ 0 & 0 & 0 & 1 \end{pmatrix}
\end{equation}
yields
\begin{equation} \label{eq:chiral_H}
    H = \begin{pmatrix} 0 & D^{\dagger} \\ D & 0 \end{pmatrix}, \quad D = \begin{pmatrix} - 2 i v_0 \overline{\de} & w_{AB} {U}(\vec{r}) \\ w_{AB} {U}(-\vec{r}) & - 2 i v_0 \overline{\de} \end{pmatrix},
\end{equation}
where $\overline{\de} = \frac{1}{2} ( \de_x + i \de_y )$ and ${U}(\vec{r}) = e^{- i \vec{q}_1 \cdot \vec{r}} + e^{i \phi} e^{- i \vec{q}_2 \cdot \vec{r}} + e^{- i \phi} e^{- i \vec{q}_3 \cdot \vec{r}}$. 

After changing variables $\vec{r} \rightarrow k_\theta \vec{r}$ and re-scaling the $\vec{q}_i \rightarrow \frac{ \vec{q}_i }{ k_\theta }, i = 1,2,3,$ we derive
\begin{equation} \label{eq:chiral_H_2}
    H = \begin{pmatrix} 0 & D^{\dagger} \\ D & 0 \end{pmatrix}, \quad D = \begin{pmatrix} - 2 i v_0 k_\theta \overline{\de} & w_{AB} {U}(\vec{r}) \\ w_{AB} {U}(-\vec{r}) & - 2 i v_0 k_\theta \overline{\de} \end{pmatrix}.
\end{equation}
Finally dividing by $v_0 k_\theta$ and defining
\begin{equation}
    \alpha := \frac{ w_{AB} }{ v_0 k_\theta }
\end{equation}
yields the TKV Hamiltonian stated in the main text.

\end{document}